%% file: main-journal.tex
\documentclass[11pt]{article}

\pdfoutput=1 

\usepackage[utf8]{inputenc}

\newcommand{\lv}[1]{{
#1}}
\newcommand{\lvbis}[1]{{
#1}}

\title{Certificates in P and Subquadratic-Time Computation of Radius, Diameter, and all Eccentricities in Graphs%
\thanks{This work was supported by the French ANR projects ANR-22-CE48-0001 (TEMPOGRAL), ANR-24-CE48-4377 (GODASse) and ANR-23-PEIA-005 (REDEEM).}
\thanks{This paper was first published in SODA (see \url{https://dx.doi.org/10.1137/1.9781611978322.70}) and then in Algorithmica with few additional results (see \url{https://doi.org/10.1007/s00453-025-01344-6}). This version additionally corrects values that were incorrectly reported in the experiments (see Table~\ref{tab:all}).}
}
\author{Feodor Dragan\thanks{Kent State University, Department of Computer Science, Kent, Ohio,  USA, \texttt{dragan@cs.kent.edu}.}
  \and Guillaume Ducoffe\thanks{University of Bucharest, Faculty of Mathematics and Computer Science, and National Institute for Research and Development in Informatics, Romania, \texttt{guillaume.ducoffe@ici.ro}.}
  \and Michel Habib\thanks{IRIF, Université Paris Cité \& CNRS, Paris,  France, \texttt{habib@irif.fr}.}
  \and Laurent Viennot\thanks{Inria, DI ENS, Paris, France, \texttt{laurent.viennot@inria.fr}.}
  }
\usepackage[papersize={8.5in,11in},margin=1in]{geometry}

\usepackage{color}
\usepackage{graphicx}
\usepackage{xspace}
\usepackage{url}
\usepackage{amsmath,amssymb}
\usepackage{hyperref}
\usepackage{cleveref}
\hypersetup{pdfstartview=XYZ}

\usepackage{thmtools}
\usepackage{capt-of}

\newcommand{\commentout}[1]{}

\usepackage[algo2e,vlined,english,boxed]{algorithm2e} 
      \SetAlgoLined
      \SetKwInput{Input}{Input}
      \SetKwInput{Output}{Output}
      \SetKw{Return}{Return}
      \SetKwIF{If}{ElseIf}{Else}{If}{then}{else if}{else}{endif}
      \SetKwFor{For}{For}{do}{done}
      \SetKwFor{for}{For}{}{}
      \SetKwFor{While}{While}{do}{done}
      \SetKwFor{Repeat}{Repeat}{}{}
      \SetKwRepeat{Do}{Do}{while}
      \SetKwFor{Procedure}{Procedure}{}{}
      \SetKwFor{Function}{Function}{}{}
      \SetKw{Return}{Return}%
      \SetKw{return}{return}%
      \SetKwComment{Comment}{/* }{ */}
      \SetSideCommentRight
      \DontPrintSemicolon

\usepackage{todonotes}

\def\guillaume#1{{#1}}
\long\def\jump#1\finjump{}
\long\def\beglongversion#1\endlongversion{#1}

\newtheorem{theorem}{Theorem}
\newtheorem{lemma}{Lemma}
\newtheorem{claim}{Claim}
\newtheorem{question}{Question}
\newtheorem{corollary}{Corollary}
\newtheorem{proposition}{Proposition}

\def\Box{\hbox{\hskip 1pt \vrule width 4pt height 8pt depth 1.5pt \hskip 1pt}}
\newenvironment{proof}{\medskip\noindent\textbf{Proof.}}{{}\hfill$\Box$\\}

\newcommand{\card}[1]{\left|{#1}\right|}

\newcommand{\ceil}[1]{\left\lceil{#1}\right\rceil}

\newcommand{\set}[1]{\{{#1}\}}

\let\eps=\varepsilon

\newcommand{\C}{{\cal C}}

\newcommand{\D}{{\cal D}}
\newcommand{\R}{{\cal R}}

\renewcommand{\R}{{\cal R}}
\renewcommand{\S}{{\cal S}}
\newcommand{\N}{\mathbb{N}}

\usepackage{amsfonts}
\DeclareMathOperator{\argmin}{argmin}
\DeclareMathOperator{\argmax}{argmax}
\DeclareMathOperator{\rad}{rad}
\DeclareMathOperator{\diam}{diam}

\DeclareMathOperator{\ecc}{ecc}
\DeclareMathOperator{\eccuntight}{ecc-untight}
\DeclareMathOperator{\eccslack}{ecc-slack-far}

\DeclareMathOperator{\distfrom}{DistFrom}

\DeclareMathOperator{\minecc}{minES}
\DeclareMathOperator{\argminecc}{argminES}

\DeclareMathOperator{\antipode}{Antipode}

\newcommand{\ovB}{\overline{B}}
\let\subsec=\subsection
\let\sec=\section

\begin{document}

\maketitle

\begin{abstract}
In the context of fine-grained complexity, we investigate the notion of certificate enabling faster polynomial-time algorithms. We specifically target radius (minimum eccentricity), diameter (maximum eccentricity), and all-eccentricity computations for which quadratic-time lower bounds are known under plausible conjectures. In each case, we introduce a notion of certificate as a specific set of nodes from which appropriate bounds on all eccentricities can be derived in subquadratic time when this set has sublinear size. The existence of small certificates for radius, diameter and all eccentricities is a barrier against SETH-based lower bounds for these problems. We indeed prove that for graph classes with certificates of bounded size, there exist randomized subquadratic-time algorithms for computing the radius, the diameter, and all eccentricities respectively.

Moreover, these notions of certificates are tightly related to algorithms probing the graph through one-to-all distance queries and allow to explain the efficiency of practical radius and diameter algorithms from the literature.
In particular, our formalization enables a novel primal-dual analysis of a classical approach for diameter computation. Based on our novel insights for these problems, we introduce several new algorithmic techniques related to eccentricity computation and propose algorithms for radius, diameter and all eccentricities with theoretical guarantees with respect to certain graph parameters. This is complemented by experimental results on various types of real-world graphs showing that these parameters appear to be low in practice. Finally, we obtain refined results in the case where the input graph is a power-law random graph, has low doubling dimension, has low hyperbolicity, is chordal, satisfies some Helly-type property, or has bounded asteroidal number.
\end{abstract}


\pagebreak
\tableofcontents

\pagebreak
\input{body/intro-soda.tex}


\input{body/certificates-soda.tex}
\input{body/smallcert-soda.tex}
\input{body/approx-certificate.tex}
\input{body/antipodes-journal.tex}
\input{body/graphclasses.tex}

\section{Conclusion}


In this paper we extensively study this idea of small certificates for radius, diameter and all-eccentricities. It gives us another view point on existing algorithms, somehow explaining why they are practically so efficient. It also leads to new algorithmic ideas to overcome the quadratic barrier for radius, diameter and all eccentricities. We are convinced that the celebrated notion of certificate can still lead to fruitful developments in the study of other problems in $P$ and their complexity barrier.









\bibliographystyle{plainurl}
\bibliography{biblio_clean}


\end{document}

%% file: body/intro-soda.tex
\sec{Introduction}

We investigate the notion of certificate in P, or more precisely how the existence of some small specific  certificate for a given problem can enable a polynomial-time algorithm with smaller exponent compared to the situation where no specific certificate is known. This question appears particularly interesting about problems in P where the best exponent of a polynomial-time algorithm is not completely settled. In particular, we target the problems of radius and diameter computation in graphs, as there is a large gap between known quadratic-time lower-bounds~\cite{RV13,AWV16} based on the Strong Exponential Time Hypothesis (SETH) or other plausible conjectures, and the efficiency of exact practical algorithms that appear to succeed in computing the radius~\cite{BCHKMT15} and diameter~\cite{TK11,CGHLM13,AIK15} of various types of real-world graphs with few Breadth First Search (BFS) traversals~\cite{BackstromBRUV2012}. The efficiency and the correctness of these algorithms suggest that such certificates exist in real-world graphs. \lv{Diameter computation is typically used when analyzing large real world graphs. For example, the iFub algorithm \cite{CGHLM13} was used to compute the diameter of Facebook graph~\cite{BBRUV12}. This algorithm and others are also part of the undirected graphs library of SageMath~\cite{sagemathGD2025}.}

The main approach of these algorithms dates back to \cite{TK11} and consists in maintaining for each vertex $v$ of the input graph $G$ a lower bound $\underline{e}(v)$ and/or an upper bound $\overline{e}(v)$ of its eccentricity $e(v)=\max_{w\in V(G)}d_G(v,w)$ which is the maximum distance of any node from $v$. Recall that the radius is the minimum eccentricity, while the diameter is the maximum eccentricity. Each time a BFS traversal is performed from a node $x$, the eccentricity $e(x)$ is obtained, and for each vertex $v$, its lower bound can be updated to $\max\set{\underline{e}(v),d_G(v,x)}$ while its upper bound can be replaced by $\min\set{\overline{e}(v),d_G(v,x)+e(x)}$ by triangle inequality. These practical algorithms~\cite{AIK15,BCHKMT15,CGHLM13,TK11,TK13} perform BFS traversals from specifically chosen vertices $x$ until some stopping condition is met. If some vertex $x$ with eccentricity $R$ was found and if all lower bounds are greater or equal to $R$, we can then conclude that the input graph has radius $R$. Similarly, if some vertex with eccentricity $D$ was encountered while at some point all eccentricity upper-bounds appear to be $D$ or less, we can stop and conclude that the diameter is $D$. A typical choice of appropriate sources $x$ for BFS traversals is to alternate between a vertex with minimum lower bound and a vertex with maximum upper bound~\cite{TK13,BCHKMT15}. 
In a seminal work on understanding the efficiency of practical algorithms on real-world graphs~\cite{BCT17}, an analysis of a randomized variant of this algorithm within power-law random graphs bounds the number of sources used to $o(n)$.

We see the set $X$ of sources $x$ from which such an algorithm performs BFS traversals before returning values $R$ and $D$ as a certificate that the input graph has radius $R$ and diameter $D$. Indeed, if $X$ is given, then we can perform a BFS traversal from each vertex $x\in X$, update lower bounds accordingly and check that the radius is at least $R$ in $O(|X|m)$ time where $m=|E(G)|$ denotes the number of edges in the input graph $G$. If $X$ also contains a vertex of eccentricity $R$, we can then certify that the radius of the input graph is indeed $R$. As an example, for odd $k$, we will see that a $k\times k$ square grid has a five-node certificate (its center and its four corners).
Similarly, we can check that $X$ contains a vertex of eccentricity $D$ and that all upper bounds are lower or equal to $D$ for certifying that the diameter is indeed $D$ in $O(|X|m)$ time. In the worst case, the algorithm may use $X=V(G)$ and require quadratic time, but it appears that various type of real-world graph do have such a set $X$ with few dozens of vertices. Coming back to a $k\times k$ square grid with odd $k$, it has a two-node certificate (its center and a corner).
It should be noted that this notion of certificate is independent of any algorithm: it is a graph property to have small or large certificates. 

Note that the existence of a certificate of size $o(n)$ enables a non-deterministic subquadratic-time algorithm for computing the diameter (by first guessing the certificate and then computing the appropriate bounds on all eccentricities). Following the results of~\cite{CGIMPS16}, this implies that for any graph class with truly sublinear certificates, the existence of meaningful lower bounds for diameter computation based on SETH is unlikely.
Furthermore, this opens the possibility of breaking the quadratic barrier for (deterministic or probabilistic) diameter computation (or similarly radius computation) if one can efficiently find a small certificate when there exists one.
We thus ask the following.
\begin{quote}
\textbf{Main question:} For a given class of graphs, does the existence of small specific certificates coincides with the existence of a truly subquadratic-time algorithm for computing either the radius, or the diameter, or even all eccentricities?
\end{quote}
We believe that the same question could be investigated in other classical problems studied in fine-grained complexity such as for example the All-Pairs Shortest Path Problem (APSP) in the light of practical algorithms for fast shortest-path computation such as Contraction Hierarchies~\cite{GeisbergerSSD2008} or Hub Labeling~\cite{AbrahamDGW2012}. 

\medskip

The paper is presented in the context of unweighted undirected graphs but all the notions and algorithms extend to the weighted and/or directed cases as shown in a short experimental part.
In the directed setting, we restrict ourselves to the strongly directed case where diameter is again well defined as the maximum (outward or inward) eccentricity while two notions of radius arise as minimum outward (resp. inward) eccentricity. A more general setting for diameter in weakly-connected directed graphs is explored in~\cite{AIK15,BCHKMT15}. We believe that our certificate notions could naturally be extended to this setting.

\subsec{Our contributions}

Our general idea of certificate consists in a set of nodes such that the distances from these nodes to all nodes (rather than all-to-all pairs) allow to deduce the value of either the radius, or the diameter, or even all eccentricities, with certainty. It is thus tightly related to \emph{one-to-all distance based algorithms}, that is algorithms probing the graph through one-to-all distances queries where a query for a vertex $x$ returns the vector of all distances from $x$. Several concepts of certificates for radius, diameter and all eccentricities arise from this general idea. We now detail them and our related results. When considering a graph $G=(V,E)$, we assume that it is connected, and we generally let $n=|V|$ and $m=|E|$ denote its number of vertices and edges respectively. We use the terms vertex and node interchangeably.

\paragraph*{Radius.}
Given a graph $G=(V,E)$ with radius $\rad(G)=r$, we define a \emph{radius certificate} as a set $L$ of nodes such that any node of $G$ is at distance at least $r$ from some node of $L$. 
Equivalently, it can be defined as a covering of the node set $V$ with complementary of open balls of radius $\rad(G)$ (excluding nodes at distance $\rad(G)$). As an example, for odd $k$, a $k\times k$ square grid has radius $k$ and its four corners form a radius certificate. 
(See also the set $R$ in the example of Figure~\ref{fig:bowtie}.)
In the case of radius, we obtain the following equivalence as an answer to our main question (as a consequence of Theorems~\ref{th:lb} and~\ref{thm:rad-cert-random}).

\begin{theorem}\label{th:main-radius}
    Given a class of graphs $\cal G$ and a sublinear function $\ell(n)$, we have:
    \begin{itemize}
        \item The existence of a one-to-all distance based randomized (Monte Carlo) algorithm for computing radius within $\cal G$ and running with $\ell(n)$ queries implies that every graph of $\cal G$ with $n$ vertices has a radius certificate of size $O(\ell(n))$.
        \item If every graph with $n$ vertices in $\cal G$ has a radius certificate of size $\ell(n)$ at most, then there exists a one-to-all distance based randomized (Monte Carlo) algorithm computing within $\cal G$ the radius and a radius certificate of size $O(\ell(n)\log n)$ through $O(\ell(n)\log^3n)$ queries (in $O(\ell(n)m\log^3n)$ time) with high probability.
    \end{itemize}
\end{theorem}

The problem of finding a minimum radius certificate is indeed shown to be equivalent to minimum set cover. It is thus NP-hard while $O(\log n)$-approximation (only) is doable in polynomial time. Compared to set cover, it has an additional difficulty: the sets are not directly available and computing all of them would require quadratic time at least. This setting for set cover has been considered in the literature~\cite{SM10}, and we show how these prior results can be exploited in the design of fast randomized approximation algorithms for minimum-size radius certificate.

\paragraph*{Diameter.}
Given a graph $G=(V,E)$ with diameter $\diam(G)=D$, we define a \emph{diameter certificate} as a set $U$ of nodes such that any node $v$ is at distance at most $D-e(x)$ from a node $x\in U$. In other words, it corresponds to a covering of $V$ with balls $B[x,D - e(x)]$ of radius $D - e(x)$ where $e(x)$ is the eccentricity of the center $x$ of the ball.
As an example, for odd $k$, a $k\times k$ square grid has diameter $2k$ and its center forms a diameter certificate. (See also the set $D$ in the example of Figure~\ref{fig:bowtie}.) Again, the problem of finding a minimum diameter certificate is equivalent to minimum set cover. However, $O(\log n)$-approximation in subquadratic time appears much more difficult as the radii of the balls considered depend from the (unknown) eccentricities of their centers.
Nevertheless, it is still possible to show that the existence of truly sublinear diameter certificates enables a truly subquadratic-time algorithm for diameter computation (see Theorem~\ref{thm:diam-cert-random}).

\begin{figure}[t]
  \begin{center}
    \includegraphics[width=\textwidth]{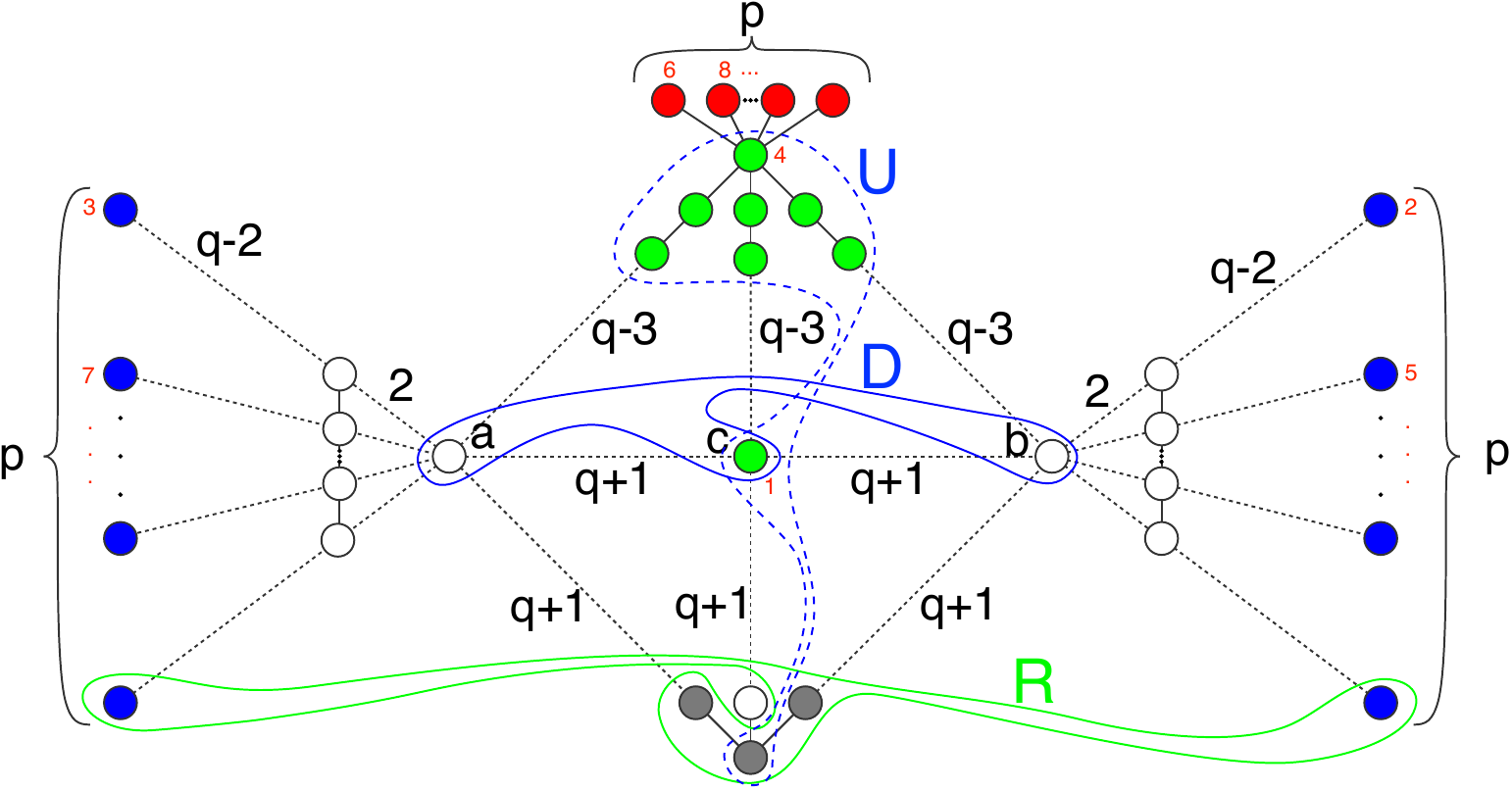}
  \end{center}
  \caption{An example of bow-tie shaped graph $BT_{p,q}$ (for $p\ge 2$ and $q\ge 6$) with radius certificate $R$, diameter certificate $D$ and all eccentricity certificate ($R,U$). Its diameter is $4q-2$ (eccentricity of blue nodes). Its radius is $2q+1$ (eccentricity of green nodes). Plain lines correspond to edges while a dashed line with label $\ell$ corresponds to a path of length $\ell$. All unlabeled dashed lines from $a$ and $b$ have length 2, while all unlabeled dashed lines from blue nodes have length $q-2$ (similarly to given labels on upper lines).}
  \label{fig:bowtie}
\end{figure}

\begin{theorem}
    For every class of graphs ${\cal G}$ such that every graph with $n$ vertices has a diameter certificate of size $\ell(n)$ at most, there exists a {\em randomized} (Monte Carlo) algorithm for computing the diameter within ${\cal G}$ in $O(m\sqrt{\ell(n)n}\log^{3/2}{n})$ time with high probability.
\end{theorem}

\lv{Furthermore, we show that the existence of a subquadratic-time one-to-all distance based algorithm implies the existence of a weaker type of certificate. More precisely, we define an \emph{extended diameter certificate} as a set $X$ of nodes such that for any pair $u,v$ of nodes, there exists a node $x\in X$ satisfying $d(u,x)+d(x,v)\le D$. As an example, in a cycle of length $4k$ for an integer $k$, four evenly spaced nodes form an extended diameter certificate. In contrast, the entire set of nodes constitutes the only diameter certificate. 
We can then state the following (see \Cref{th:lbdiam} and \Cref{prop:extended-diameter}).

\begin{theorem}
    Given a class of graphs $\cal G$ and a sublinear function $\ell(n)$, we have:
    \begin{itemize}
        \item The existence of a one-to-all distance based randomized (Monte Carlo) algorithm for computing diameter within $\cal G$ and running with $\ell(n)$ queries implies that every graph of $\cal G$ with $n$ vertices has an extended diameter certificate of size $\ell(n)$ at most.
        \item Given a graph $G$ with $n$ vertices and an extended diameter certificate of size $o(\log n)$, its diameter can be computed in subquadratic time.
    \end{itemize}
\end{theorem}

As another example, in a split graph, that is a graph that can be decomposed as the union of a clique and an independent set with possibly edges between them, the clique is an extended diameter certificate when the graph is connected. Indeed, such a graph has diameter two or three, and it has diameter two only if, for any two nodes in the independent set, they have a common neighbor in the clique. Distinguishing between diameter two and three in such graphs is proved to require quadratic time at least under SETH even if the clique has size $O(\log n)$~\cite{RV13}. We thus suspect that the complexity of computing the diameter of a graph has at least an exponential dependency in the minimum size of an extended diameter certificate. The possibility of a truly subquadratic-time algorithm for graphs having an extended diameter certificate of size $o(\log n)$ remains an open question.
}

\paragraph*{All eccentricities}
Given a graph $G=(V,E)$, we define a \emph{tight lower certificate} as a set $L$ of nodes such that any node $v$ is at distance $e(v)$ from some node $x$ of $L$. 
We define a \emph{tight upper certificate} as a set $U$ of nodes such that any node $v$ is at distance at most $e(v)-e(x)$ from some node $x\in U$. We call all-eccentricity certificate a couple $(L,U)$ of such tight-lower and tight-upper certificates. 
In a $k\times k$ square grid with odd $k$, the four corners and the center form such a couple.
(See also $(R,U)$ in the example of Figure~\ref{fig:bowtie}.)
The problem of finding a minimum tight lower certificate again appears to be equivalent to set cover. However, the notion of tight upper certificate is related to a partial order structure that makes the minimum tight upper certificate unique and enables its computation in polynomial time although subquadratic-time computation still remains open. We can nevertheless obtain subquadratic-time computation of all eccentricities when sufficiently small certificates exist (see Theorem~\ref{th:approx-all-cert}):

\begin{theorem}    
    Given an arbitrary $n$-node graph $G$, all eccentricities, a tight lower certificate $L$ of size $O(\ell^* \log n)$ and the minimum tight upper certificate $U^\preceq$ can be computed in $O(m \ell^* \card{U^\preceq} \log^2 n)$ time with high probability where $\ell^*$ denotes the size of a minimum tight lower certificate.
\end{theorem}

\lv{
\paragraph*{Approximating diameter.}
We also consider the problem of approximating the diameter of a graph and introduce a related notion of certificate. More precisely, given a graph $G=(V,E)$ with diameter $\diam(G)=D$ and a value $c\ge 1$, we define a \emph{$c$-approximate diameter certificate} as a set $U$ of nodes such that any node $v$ is at distance at most $c\cdot D-e(x)$ from a node $x\in U$. This notion is related to that of distance $k$ domination. Indeed, consider a set $U$ of nodes that dominates the graph at distance $(c-1)D$, that is such that any node $v$ is at distance at most $(c-1)D$ from some node $x$ of $U$. As the excentricity of such a node $x$ is $D$ at most, we then have $d(v,x)\le c\cdot D - e(x)$, proving that $U$ is a $c$-approximate diameter certificate. As an example, three evenly spaced nodes in a cycle $G$ of length $6k$ for an integer $k$ form a $4/3$-approximate diameter certificate as they dominate the graph at distance $\diam(G)/3$. We prove that any graph has a $3/2$-approximate diameter certificate of sublinear size (see Theorems~\ref{th:approx-cert-even} and~\ref{th:approx-cert}).

\begin{theorem}
    Any graph $G=(V,E)$ with even diameter has a $\frac{3}{2}$-approximate diameter certificate $X$ of size $O(\sqrt{n}\log n)$. Moreover, such a set $X$ and a $3/2$-approximation of the diameter can be computed in randomized $O(m\sqrt{n}\log n)$ time by using a Monte Carlo algorithm. 
\end{theorem}

The existence of a $(2-1/k)$-approximate diameter certificate of size $\widetilde O(n^{1/k})$ in any graph with diameter divisible by $k$ remains an open question for $k>2$.
}

\paragraph*{Practical algorithms.}
We propose a primal-dual analysis of a basic variant of the diameter algorithm of~\cite{TK11} which was not known before as far as we know. 
Indeed, a basic primal-dual argument implies that the maximum size $\pi_1$ of a packing for (closed) balls $B[u, \alpha(\diam(G)-e(u))]$ for $\alpha =1$ is a lower bound on the minimum size of a diameter certificate, as the latter corresponds to a covering with such balls. We prove that the nodes selected by this basic variant for performing BFS traversals form a packing $P$ for balls with radii reduced by a factor $\alpha =1/3$ (each such ball contains at most one node of $P$) implying that it indeed computes a $\pi_{1/3}/\pi_1$ approximation of the minimum diameter certificate.

We also design new algorithms for radius and all eccentricities. They both rely on a specific 
tight lower certificate, namely the set $A$ of all antipodes, which are defined as follows.
It is a subset of furthest nodes, i.e. the nodes that are at furthest distance from some node. 
More precisely, given a ranking of the nodes
(e.g., their ID order), we define the antipode of a node $u$ as the node at furthest distance from $u$ having highest rank (the ranking is used for breaking ties among nodes at the same distance). Our all-eccentricity algorithm relies on a characterization of the (unique) minimum tight upper certificate of a graph which allows the algorithm to compute it. We obtain the following guarantees
(as a consequence of Theorems~\ref{th:rad},~\ref{th:diam} and~\ref{th:all-ecc}).

\begin{theorem}\label{th:main}
  Given a graph $G$ having $\ell$ antipodes overall (according to a given ranking), it is possible to  compute:
  \begin{itemize}
  \item its radius, a center and a radius certificate of size $\ell$ at most
    in $O(\ell m)$ time,
  \item its diameter, a diametral node and a diameter certificate of size $\pi_{1/3}$ at most
    in $O(\pi_{1/3} m)$ time
    where $\pi_{1/3}$ is the maximum packing size for open balls $B(u,\frac{1}{3}(\diam(G)-e(u)))$,
  \item all eccentricities, a lower certificate of size $\ell$ at most and the minimum upper certificate $U^\preceq$
    in $O((\ell+\card{U^\preceq}) m)$ time.
  \end{itemize}
\end{theorem}

We provide small experiments on various types of real-world graphs that confirm that these graphs have extremely small radius and diameter certificates (less than 30 nodes for all graphs when they have from tens of thousands to more than ten million nodes). Surprisingly they also have a very small tight lower certificate: the set of all antipodes.
Although we do not know how to compute this set in subquadratic time, the fact that its size is very small in practical graphs guarantees fast termination of our radius algorithm. Although the idea of using  an antipode as a good candidate for a diametral node dates back to the TwoSweep heuristic~\cite{MLH09}, the idea of repeteadly using antipodes in radius and all-eccentricity computations is new. The observation of their sparsity in real-world graphs seems new also. Note that the set of furthest nodes (which contains antipodes) is also an obvious tight lower certificate, but it appears to be significantly larger in several networks.
We also observe that these graphs have relatively small coverings with balls of reduced radii compared to the balls required for a diameter certificate, allowing to provide support for the efficiency of diameter algorithms based on the approach of~\cite{TK11}. 
The size of the minimum tight upper certificate appears to be quite variable in our experiments, from few percents to a large fraction of the nodes.
Note that for graphs where it is larger than the number of antipodes (all our real-world graphs), our all-eccentricity algorithm is somehow optimal as its complexity is then the same as the algorithm checking the certificate.


\paragraph*{Classical graph classes with specific certificates.}
We additionally refine our results in specific graph classes:
\begin{enumerate}
    \item The analysis of~\cite{BCT17} allows to show that, for any $\eps>0$, sufficiently large \textbf{power-law random graphs} have $n^\eps$ furthest nodes at most asymptotically almost surely, and their set is thus a tight lower certificate of size $n^{\eps}$ at most.
    \item A refinement of our primal-dual algorithm for diameter leads to subquadratic $(1+\eps)$-approximation of both radius and diameter in \textbf{graphs with constant doubling dimension}.
    \item Our primal-dual radius algorithm runs in linear time in \textbf{$\delta$-hyperbolic graphs} with exact computation when the degree is bounded and constant additive approximation in general. 
    \item The centers of any \textbf{chordal graph} form a diameter certificate, and all its eccentricities can be computed in linear time if it has bounded degree. 
    \item \textbf{Helly graphs} and \textbf{bipartite Helly graphs} have radius certificates of constant size and their radius can be computed in near-linear time.
    \item Finally, the radius of \textbf{graphs with asteroidal number at most $k$} can be computed in $O(km^{3/2})$ time.
\end{enumerate}
Our results on power-law random graphs, negatively curved graphs, and chordal graphs revisit some previous results to obtain bounds on certificates. We are not aware of any specific diameter algorithm for graphs with low doubling dimension prior to this work. Our framework allows to improve the state of the art of radius computation in both Helly graphs and graphs with bounded asteroidal number.

\medskip

Overall, we believe that our notion of certificate allows to better understand when the quadratic barrier can be overcome for radius, diameter, and all eccentricity computations. Moreover, it could also be fruitful for investigating other known barriers in P.
It also provides new insights on the efficiency of practical algorithms for radius and diameter, and enables more robust practical algorithms with complexity guarantees.
We argue that it significantly enhances the state of the art for all-eccentricity computation. The new techniques proposed also enable new types of radius and diameter algorithms with parameterized complexity with respect to parameters related to the size of such certificates. We give one such example for graphs of bounded asteroidal number, for which we present the first subquadratic-time algorithm for computing the radius. Incidentally, the latter answers an open question from~\cite{Duc22}.

\lv{A large part of our study specifically focuses on one-to-all distance based algorithms. This paradigm is arguably very natural. Despite this, its theoretical analysis is not well understood (\cite{BCT17} is one of the few works in this direction). It is practical and it is agnostic: it can be applied to any graph while this is not the case for algorithms dedicated to specific graph classes.}

\subsec{Related work}


Our notion of certificate is related to the broader notion of certifying algorithm~\cite{McConnellMNS2011,AlkassarBMR2011} which is an algorithm that produces, in addition to its output, a certificate or witness, that is an easy-to-verify proof that the particular output is valid. Such an algorithm is said to be efficient if the complexity of computing the output and then verifying its validity using the certificate is similar to that of the best (non-certifying) algorithm. We are mainly interested in the case where verification can be done faster than the best algorithm. The practical algorithms we propose for diameter, radius and all eccentricities are indeed certifying algorithms. We are not aware of any non-trivial certifying algorithms for these prior to this work. 

As mentioned earlier, our notion of certificate is also closely related to nondeterministic algorithms\cite{CGIMPS16,Kunnemann2018}, since guessing a certificate and verifying it turns out to be a nondeterministic algorithm for solving the problem with the same complexity as the verification algorithm. However, we express this latter complexity as a function of both the size of the input and the size of the certificate. More specifically, we are interested in relating the size of a smallest certificate for a given input to the difficulty of solving the problem for that particular input. 

\lv{In the context of radius computation, our notion of radius certificate is related to the combinatorial dimension $d$ introduced in \cite{FunkePS2025}. In our terms, this parameter $d$,  can be defined as the maximum size of an inclusion-wise minimal radius certificate. Under the assumption that $d$ is subcubic, an adaptation of Clarkson's algorithm for solving LP problems~\cite{Clarkson1995} leads to subquadratic-time radius computation.}

In the context of approximating diameter~\cite{CGR16}, a $h$-dominating set $X$ is used for upper-bounding the diameter: given $X$, the minimum $h$ such that $X$ is $h$-dominating and the eccentricity of every node in $X$ can be computed with $|X|+1$ BFS traversals. The diameter can then be bounded by $h+\max_{x\in X}e(x)$. Our diameter certificate can be seen as a refinement of this approach where we bound independently the eccentricity of each node $u$ with respect to the best bound $d(u,x)+e(x)$ for $x\in X$. In general, this provides a better diameter bound.

\lv{The fine-grained complexity of approximating diameter is studied in a series of papers (see \cite{AingworthCIM1999,RV13,ChechikLRSTW2014,Li2020,Bonnet2021,DaLiWi}). For any integer $k$ and $\delta>0$, computing a $2-1/k-\delta$ approximation of the diameter requires $m^{1+1/(k-1)-o(1)}$ time under SETH~\cite{DaLiWi}. An algorithm achieving $\widetilde O(m^{1+1/k})$ time is known for $k=2$~\cite{ChechikLRSTW2014}. For $k>2$, an almost $2-1/2^k$ approximation (i.e. up to some additive term) can be computed in $\widetilde O(m^{1+1/(k+1)})$ time~\cite{CGR16}.}

Our work follows a long line of research around practical radius and diameter computations that dates back to the computation of a center in a tree~\cite{H73}. It consists in a two-sweeps approach where the last visited node in a first BFS traversal is used as the source of a second BFS traversal. It can also be seen as a heuristic~\cite{MLH09} providing a diameter estimate that appears to be often tight in practice~\cite{AADr}. It thus introduces the idea of using what we call antipodes as tentative diametral nodes. The two-sweep heuristic was also shown to provide good approximation (up to a small constant) for chordal graphs, $\delta$-hyperbolic graphs and various other graph classes~\cite{ChDrEsHaVa,ChDrHaVaAlR,CDK03}. Note that, if one insists in approximating the diameter of an arbitrary graph in near-linear time, it turns out~\cite{DaLiWi} that the simple linear-time algorithm, that outputs a 2-approximation to the diameter by performing one BFS traversal from an arbitrary vertex,  is optimal.
A four-sweeps heuristic is proposed in \cite{CGHLM13} and complemented with an exact diameter algorithm called iFub. The four-sweep heuristic performs twice the two-sweep method, using a mid-point of the longest path found in the first round as the starting point of the second one. The idea is that mid-points of longest paths make good candidates for central nodes or at least nodes with small eccentricity. The iFub method additionally inspects furthest nodes from the best candidate center found with the four-sweep heuristic until exact value of the diameter can be inferred.

The concept of certificate is somehow implicit in the method introduced in~\cite{TK11,TK13} that consists in maintaining lower and upper bounds on the eccentricity of each node. After each BFS traversal these bounds are improved based on distances from the source of the traversal. The sources used for the BFS traversals performed by the algorithm form what we call a certificate. Contrarily to this approach, we distinguish nodes used for improving lower bounds (the lower certificate) from those used for improving upper bounds (the upper certificate). 
This approach with bounds has been extended for diameter computation in weakly connected directed graphs  in parallel works~\cite{BCHKMT15,AIK15}.
The Exact SumSweep method~\cite{BCHKMT15} additionally computes the radius in addition to the diameter.
It integrates many techniques proposed in previous practical algorithms plus a heuristic based on sum of distances for discovering nodes with large eccentricity in an initial phase.

An impressive analysis of these algorithms (several heuristics, IFub and Exact SumSweep) within power-law random graphs is performed in~\cite{BCT17}. The analysis proposed for Exact SumSweep indeed implies the existence of certificates of size $n^{O(\eps)}$ for both radius and diameter for any $\eps>0$ and sufficiently large power-law random graphs (see Section~\ref{sec:power-law}).
We revisit their analysis to show a similar bound for a tight lower certificate.

Our primal-dual approach for practical algorithms has some similarities with the study of
packings and coverings of hyperbolic graphs with balls proposed in~\cite{CE07}, although slightly different problems are considered. It would be interesting to derive similar results in hyperbolic graphs for the collections of balls (or complementary of balls) we consider here.

\subsec{Structure of the paper}

\Cref{part:certificates} introduces basic graph and set-cover terminology (Section~\ref{sec:prelim}),
the notions of certificate for radius, diameter, and all eccentricities (Section~\ref{sec:cert}), and then provides a characterization of minimum tight upper certificates (Section~\ref{sec:min-tight-upper-cert}) as well as the link between radius certificates and one-to-all distance-based radius algorithms (Section~\ref{sec:lb}), \lv{and the link between extended diameter certificate and one-to-all distance-based diameter algorithms (\Cref{sec:extended-diameter})}.
%
\Cref{part:smallcert} provides partial answers to our main question about subquadratic-time computations of radius (Section~\ref{sec:small-rad-cert}), diameter (Section~\ref{sec:small-diam-cert}), and all eccentricities (Section~\ref{sec:approx-all-cert}). Results in terms of a stricter notion of radius certificate are given in Section~\ref{sec:approx-rad-cert}.
\lv{The notion of $c$-approximated diameter is introduced in \Cref{sec:approximate-diameter}.}
%
\Cref{part:antipodes} presents our practical algorithms for radius (Section~\ref{sec:rad}), diameter (Section~\ref{sec:diam}) and all eccentricities (Section~\ref{sec:all-ecc}). A new technique at the core of our radius and all-eccentricity algorithms is described in Section~\ref{sec:selection}. Some experimental results detail for some real-world graphs various parameters involved in our analyses (Section~\ref{sec:exp}).
%
\Cref{part:graphclasses} is dedicated to our results on power-law random graphs (Section~\ref{sec:power-law}), graphs with constant doubling dimension (Section~\ref{sec:doubling}), hyperbolic graphs
(Section~\ref{sec:hyperbolic}), chordal graphs (Section~\ref{sec:chordal}), Helly graphs (Section~\ref{Helly}) and graphs with bounded asteroidal number (Section~\ref{Asteroid}).

%% file: body/certificates-soda.tex
\part{Certificates for radius, diameter, all eccentricities}
\label{part:certificates}

\section{Preliminaries and definitions}
\label{sec:prelim}

In the sequel, we suppose that we are given a finite undirected unweighted connected graph $G$. We denote by $V$ its set of nodes and by $E$ its set of edges. We let $n=|V|$ and $m=|E|$ denote its number of nodes and edges respectively. Let $d(u,v)$ denote the distance between two nodes $u$ and $v$ in $G$, that is the length of a shortest path from $u$ to $v$. The \emph{eccentricity} $e(u)$ of a node $u$ is the maximum length of a shortest path from $u$, that is $e(u)=\max_{v\in V} d(u,v)$. The \emph{furthest nodes} of $u$ are the nodes $v$ at furthest distance from $u$, i.e. satisfying $d(u,v)=e(u)$. We let $F(u)=\set{v \in V:d(u,v)=e(u)}$ denote their set. Given a ranking $r$ of the nodes, the \emph{antipode} $\antipode_r(u)$ of a node $u$ for $r$ is its furthest node with highest rank. Formally, $\antipode_r(u) = \argmax_{v\in V} (d(u, v), r(v))$ where pairs are ordered lexicographically.
A node is called a furthest node (resp. an antipode) if it is a furthest node (resp. an antipode) of some other node. Given a set $W\subseteq V$, we let $\antipode_r(W)=\set{\antipode_r(u) : u\in W}$ denote the set of antipodes from nodes in $W$.
The \emph{diameter} $\diam(G)=\max_{u\in V}e(u)$ of $G$ is the maximum eccentricity in $G$ and the \emph{radius} $\rad(G)=\min_{u\in V}e(u)$ is the minimum eccentricity in $G$. A \emph{diametral node} $b$ is a node with maximum eccentricity ($e(b)=\diam(G)$). A \emph{central node} $c$ (or simply \emph{center}) 
is a node with minimum eccentricity ($e(c)=\rad(G)$). Denote by $C(G)=\{v\in V: e(v)=\rad(G)\}$ the set of all central nodes of $G$. 
We let $B[u,r]=\set{v\in V\mid d(u,v)\le r}$ (resp. $B(u,r)=\set{v\in V\mid d(u,v) < r}$)  denote the (closed) ball (resp. open ball) with radius $r$ centered at a node $u$. Similarly, we define its \emph{coball} of radius $r$ as $\ovB(u,r)=\set{v\in V\mid d(u,v)\ge r}$, that is the complementary of $B(u,r)$.

We restrict ourselves to algorithms based on one-to-all distance queries:
we suppose that an algorithm $\distfrom$ for one-to-all distances is given (typically BFS or Dijkstra in the weighted case). It takes a graph $G$ and a node $u$ as input and returns distances from $u$. More precisely, $\distfrom(G,u)$ returns a vector $D$ such that $D(v)=d(u,v)$ for all $v\in V$. In particular, $e(u)$ can be defined as the maximum value in the vector and the antipode of $u$ as the index with highest rank where this value appears in $D$. We may measure the complexity of an algorithm by the number of one-to-all distance queries it performs when its cost mainly comes from these operations. A one-to-all distance-based algorithm accesses the graph only through one-to-all distance queries and solely relies on distances known from queries, triangle inequality, and non-negativeness of distances for bounding unknown distances.

Given a collection $\S$ of subsets of $V$ such that $\cup_{S\in \S}S=V$, a \emph{covering} with $\S$ is a sub-collection $\C\subseteq \S$ of sets such that their union covers all of $V$: $V\subseteq \cup_{S\in \C}S$. (A set $S\in \S$ is said to cover elements in $S$). Recall that the set-cover problem consists in finding a covering of minimum size. We define a \emph{packing} for $\S$ as a subset $P\subseteq V$ such that any set of $\S$ contains at most one element in $P$. The denomination comes from the fact that elements of $P$ correspond to pairwise disjoint subsets of the dual collection $\S^*=\set{\set{S\in\S \mid u\in S} : u\in V}$. A \emph{hitting set} for $\S$ is a set $P$ that intersects all sets of $\S$. (Equivalently, a hitting set can be defined as a covering for $\S^*$ but it may be more convenient to consider a collection rather than its dual.)
We let $\pi(\S)$ denote the maximum size of a packing for $\S$, and $\kappa(\S)$ denote the minimum size of a covering for $\S$.
As a covering must cover each element of a packing with distinct sets, we obviously have $\pi(\S) \le \kappa(\S)$ (weak duality). We say that a collection $\R$ is \emph{restricted} compared to $\S$ if there exists a one-to-one mapping $f$ from $\R$ to $\S$ such that $S\subseteq f(S)$ for all sets $S\in \R$. Note that this mapping then turns any covering with $\R$ into a covering with $\S$ and we thus have $\kappa(\S)\le \kappa(\R)$.
Similarly, a packing for $\S$ is also a packing for $\R$ and we have $\pi(\S)\le \pi(\R)$. In other words, restricting the sets of a collection to smaller subsets increases maximum packing size and minimum covering size.

\section{Lower and upper certificates for eccentricities}
\label{sec:cert}


Our notion of certificate is based on the fact that knowing all distances from a given node $x$ provides some bounds on the eccentricities of other nodes:
\begin{equation}\label{eq:certif}
  \forall u\in V,\ d(u,x) \le e(u) \le d(u,x) + e(x).
\end{equation}
The first inequality derives directly from the eccentricity definition while the second one is a consequence of the triangle inequality.
A possibly tighter lower-bound of $\max\set{d(u,x), e(u) - d(u,x)}$ could be used as in \cite{TK11} but this optimization would complicate our definitions without reducing significantly certificate size  (see Theorem~\ref{th:lb} in Section~\ref{sec:lb}).

\medskip

We say that a set $L$ (resp. $U$) of nodes is a \emph{lower certificate} (resp. an \emph{upper certificate}) of $G$ when it is used to obtain lower bounds (resp. upper bounds) of eccentricities in $G$. Given the distances from a node $u$ to all nodes in $L\cup U$ and the eccentricities of nodes in $U$, we have the following lower and upper bounds for the eccentricity of any node $u$ (as a direct consequence of Inequation~\ref{eq:certif}):
\begin{equation*}
e_L(u) \le e(u)\le e^U(u),\ 
  \mbox{ where }\left\{\begin{array}{l}
                e^U(u) = \min_{x\in U}d(u,x) + e(x)\\
                e_L(u) = \max_{x\in L}d(u,x)
    \end{array}\right.
\end{equation*}
Note that the bounds would work also by taking the $\min$ and $\max$ over all nodes $x\in U \cup L$. However, we prefer to distinguish nodes that provide lower bounds from those that provide upper bounds. If needed, a node could be both in $L$ and $U$.

A lower (resp. upper) certificate $L$ (resp. $U$) is said to be \emph{tight} when $e_L(u)=e(u)$ (resp. $e^U(u)=e(u)$) for all $u\in V$. An \emph{all-eccentricty certificate} is defined as a couple $(L,U)$ of a tight lower certificate $L$ and a tight upper certificate $U$.


\medskip

Given a bound $D$ and a node $x$, we have $d(u,x)+e(x)\le D$ if and only if $u\in B[x,D-e(x)]$. Given an upper certificate $U$ we thus have
$e^U(u)\le D$ if and only if $u\in \cup_{x\in U} B[x,D-e(x)]$. A \emph{diameter certificate} is a set $U$ such that $e^U(u)\le \diam(G)$ for all $u\in V$. Equivalently it can be defined as a covering with $\set{B[x,\diam(G)-e(x)] : x\in U}$ using balls whose radius equals $\diam(G)$ minus eccentricity of the center (and identifying a ball with its center).

Similarly,
given a lower certificate $L$ and a bound $R$, we obviously have $e_L(u) \ge R$ for all nodes $u$ whose coball $\ovB(u,R)$ intersects $L$ (i.e., there exists a node in $L$ at distance $R$ at least from $u$). We thus define a \emph{radius certificate} as a set $L$ such that $e_L(u)\ge\rad(G)$ for all $u\in V$, or equivalently as a hitting set $L$ for the collection $\set{\ovB(u,\rad(G)) : u\in V}$ of coballs of radius $\rad(G)$.
As $x\in \ovB(u,\rad(G))$ if and only if $u\in \ovB(x,\rad(G))$, the collection of coballs of radius $\rad(G)$ is its own dual, and a radius certificate $L$ can equivalently be defined as a covering for this collection.

Note that a tight lower certificate can equivalently be defined as a hitting set for the collection $\set{\ovB(u,e(u)) : u\in V}$.
Similarly, a tight upper certificate can equivalently be defined as a covering for the collection $\set{\set{u\in V \mid d(u,x) \le e(u) - e(x)} : x \in V}$.

\paragraph{Examples.}
A path with $2k+1$ nodes has a radius certificate with two nodes (the two extremities) and a diameter certificate with one node (its mid-point). More generally, \guillaume{a graph $G$ has a one-node diameter certificate if and only if $\diam(G)=2\rad(G)$. Indeed, in one direction, if $\diam(G) = 2\rad(G)$ then any central node forms a diameter certificate (see also Proposition~\ref{prop:d=2r-any}). In the other direction, if some diameter certificate is reduced to some node $c$, then $\diam(G) \le 2(\diam(G)-e(c))$, which implies $e(c) \le \diam(G)/2$. Since $\diam(G) \le 2\rad(G)$ for any graph $G$, and moreover $e(c) \ge \rad(G)$, we obtain as desired $e(c) = \rad(G) = \diam(G)/2$.}

It can be shown that any tree has a radius certificate of two nodes (two well chosen leaves) while its centers (at most two nodes) form a diameter certificate. \guillaume{More generally, a graph $G$ has a radius certificate of size two (that is best possible) if and only if $\diam(G) \ge 2\rad(G)-1$. Indeed, in one direction, if $\diam(G) \ge 2\rad(G)-1$ then any diametral pair forms a radius certificate (see also Proposition~\ref{prop:d>2r-2}). In the other direction, if some radius certificate is reduced to two nodes $x,y$ then let $c$ be on a shortest path from $x$ to $y$ such that $d(c,x) = \left\lceil d(x,y)/2\right\rceil$. Then $d(c,x) \ge \rad(G)$, and so, $2\rad(G) \le 2d(c,x) = 2\left\lceil d(x,y)/2\right\rceil \le d(x,y) +1 \le \diam(G) +1$.} A square grid has a radius certificate with four nodes (the corners) while its centers (at most four nodes) form a diameter certificate. 

As an extreme example of graph with large certificates, consider a cycle $C$. Its only radius and diameter certificates are both the whole set of its nodes. More generally, \guillaume{a graph is \emph{self-centered} (its radius equals its diameter) \textit{if and only if} the whole set of nodes is its unique diameter certificate. 
Indeed, since any diametral node $u$ in a graph $G$ satisfies $\diam(G)-e(u) = 0$, it follows that if $U$ is a diameter certificate of $G$ such that $u \in U$ then necessarily $V \setminus \{u\} \subseteq \set{B[x,\diam(G)-e(x)] : x\in U \setminus \{u\}}$.
In particular if $G$ is self-centered, such a certificate $U$ must coincide with the whole set of its nodes. Conversely, if $G$ contains a node $v$ such that $e(v) < diam(G)$ then $V \setminus N(v)$ is a diameter certificate, and it excludes at least one node.
The \emph{antipodal graphs} are a subclass of self-centered graphs, which contain cycles and hypercubes, and it can be proved that their unique radius certificate is also the whole set of nodes.
More precisely, a graph $G$ is called antipodal in~\cite{goodman2000tight} if for every node $x$ there exists a node $y$ such that all other nodes must be on some shortest path between $x$ and $y$ (equivalently, $d(x,y) = d(x,z) + d(z,y)$ for every other node $z$).
From the definition it follows that for every node $x$, the associated node $y$ must be unique, and in particular $y$ is the unique furthest node from $x$.
Conversely, since $G$ must be self-centered~\cite{goodman2000tight}, we also have that $x$ is the unique furthest node from $y$.
Hence, every $x$ is the unique node at distance $\rad(G) = \diam(G)$ from some other node $y$, which implies, as claimed earlier, the whole set of nodes is the unique radius certificate of $G$.
The characterization of all graphs such that their only radius certificate equals the whole set of nodes remains open.}

\paragraph{Hardness of approximation.}
Similarly to~\cite{CE07}, we note that set cover can easily be encoded with ball cover: given a collection $\S$ of subsets of $V$ of an instance of the set-cover problem, consider the split graph where the sets $S$ of $\S$ form a clique and the elements $x\in V$ form a stable set so that the nodes $x$ and $S$ are adjacent if and only if $x\in S$. Without loss of generality, we may assume that no subset in $\S$ equals $V$ (otherwise, the problem is trivial), no set is empty (otherwise, we can remove it) and that there exist two elements such that no set contains both of them (if needed, we add a singleton $\set{z}$ to $\S$ where $z$ is a new dummy element added to $V$). In this graph, sets and elements have eccentricity 2 and 3 respectively. Any minimum diameter certificate is a covering with balls of radius 1 or 0 (if centered on a set or an element). One can easily transform it into a covering with balls of radius 1 centered at nodes corresponding to sets (only). It then corresponds to an optimal solution of the original set-cover problem.
Now consider the complementary graph which is  also a split graph where elements form a clique while sets form a stable set and where $x$ and $S$ are adjacent if and only if $x\notin S$. Similarly, a minimum radius certificate for this complementary graph corresponds to a covering with coballs of radius 2 centered at sets and is also an optimal solution to the original set-cover problem. 
The hardness of set-cover approximation~\cite{DS14} thus implies that computing a minimum diameter (resp. radius) certificate is NP-hard and no polynomial-time algorithm can approximate it with a factor $(1-o(1))\log n$ unless $P=NP$. Assuming that some set in $\S$ does not intersect any other set (it suffices again to add a singleton $\set{z}$ to $\S$ where $z$ is a new dummy element added to $V$), a tight lower certificate can easily be obtained from a radius certificate and we similarly obtain that computing a tight lower certificate is NP-hard and cannot be approximated within a factor $(1-o(1))\log n$ unless $P=NP$.
Surprisingly, the situation is different for the problem of finding a minimum tight upper certificate, which can be solved in polynomial time as shown now.

\section{Characterization of the unique minimum tight upper certificate}
\label{sec:min-tight-upper-cert}

Recall that a set $U$ is a tight upper certificate if for any  vertex $u$, there is a node $x\in U$ such that $d(u,x)+e(x) \le e(u)$, or equivalently $d(u,x)+e(x)=e(u)$ as $d(u,x)+e(x)\ge e(u)$ is always satisfied by triangle inequality. We call such a node $x$ a \emph{tight upper vertex-certificate} for $u$ as we have $e^{\set{x}}(u)=e(u)$. This notion yields to the following characterization.

\begin{proposition}\label{prop:upperopt}
  Given a graph $G$, being a tight upper vertex-certificate defines a binary relation $\preceq$ which is a partial order ($u\preceq x$ stands for  $e(u)=d(u,x)+e(x)$). Moreover, the set $U^{\preceq}$ of all maximum elements of this partial order is the unique tight upper certificate of $G$ with minimum size.
\end{proposition}



\begin{proof}
  We first prove that the relation $\preceq$ is a partial order.
  It is obviously reflexive as the distance from a node to itself is zero, implying $e(u)=d(u,u)+ e(u)$.
  It is antisymmetric: if $x$ and $y$ are both tight upper vertex-certificates one for the other, we then have $e(x)=d(x,y)+ e(y)$ and $e(y)=d(y,x)+ e(x)$, and thus $d(x,y)=0$.
  We finally show transitivity. Suppose that $y$ is a tight upper vertex-certificate for $x$ and that $z$ is a  tight upper vertex-certificate for $y$, {\it i.e.} $e(x) = d(x, y) +  e(y)$ and $e(y) = d(y, z) +  e(z)$. We thus have $e(x) = d(x,y) + d(y,z) + e(z)$. As triangle inequality implies $e(x) \le d(x,z) +  e(z)$,  we obtain $d(x,y) + d(y,z)\le d(x,z)$, and thus  $d(x,y)+d(y,z)=d(x,z)$ by triangle inequality again. We finally get $e(x) = d(x,z)+ e(z)$ and $z$ is a tight upper vertex-certificate for $x$.

We now show that the set $U^{\preceq}$ of maximal elements for $\preceq$ is the unique optimal tight upper certificate of $G$.
For any non-maximal element $u$, we can build a chain $u\preceq x_1\preceq x_2\preceq\cdots$ where $u$ has a tight upper vertex-certificate $x_1$, if $x_1$ is not in $U^{\preceq}$, it has a tight upper vertex-certificate $x_2$, and so on. As the partial order is finite, the chain must be finite and $x_k$ must be in $U^{\preceq}$ for some $k$. The transitivity of $\preceq$ implies that $x_k$ is a tight upper vertex-certificate for $u$ implying $e^{U^{\preceq}}(u)\le d(u,x_k)+e(x_k)=e(u)$. This shows that $U^{\preceq}$ is a tight upper certificate. 
As each element $x$ of $U^{\preceq}$ is the only tight upper vertex-certificate for itself (as a maximal element), $U^{\preceq}$ is included in any tight upper certificate of $G$. In particular, any minimum tight upper certificate must indeed equal $U^{\preceq}$.
\end{proof}

Note that $U^{\preceq}$ includes in particular all centers of the graph: a center $c$ cannot have a tight upper vertex-certificate $x\not= c$ (otherwise we have $e(x)=e(c)-d(c,x)<e(c)$ in contradiction with the minimality of $e(c)$).

The tight upper vertex-certificate relation and then $U^{\preceq}$ can easily be computed in $O(nm)$ time by computing all distances.
A more practical algorithm is proposed in Section~\ref{sec:all-ecc}.

\section{Lower bound for radius computation}
\label{sec:lb}

We now show that the notion of radius certificate is related to the minimum number of queries a one-to-all distance-based algorithm must perform.

\begin{theorem}\label{th:lb}
Given a graph $G$, if an execution of a one-to-all distance-based randomized Monte Carlo algorithm successfully computes its radius by  querying a set $L$ of nodes, then $L\cup\antipode_r(L)$ is a radius certificate for any ranking $r$. Such a radius algorithm thus requires at least $\frac{1}{2}|L_{OPT}|$ one-to-all distance queries in a successful execution where $|L_{OPT}|=\kappa(\set{\ovB(u,\rad(G)) : u \in V})$ is the minimum size of a radius certificate.
\end{theorem}

\begin{proof}
  Consider any Monte Carlo algorithm for computing the radius with positive probability. There must exist an execution with input $G$ that succeeds. Let $L$ denote the set of nodes queried for one-to-all distances in such an execution. A proof of correctness of the algorithm allows to conclude that all nodes have eccentricity $\rad(G)$ at least based on triangle inequality and the distances known to the algorithm. That is for each node $u$, there is a node $v$ such that we can prove $d(u,v)\ge \rad(G)$ based on triangle inequality and distances from nodes in $L$. Consider a node $v$ such that the proof uses a minimum number of triangle inequalities. If neither $u$ nor $v$ is in $L$, the proof must use a triangle inequality $d(u,v)\ge |d(u,x_1) - d(v,x_1)|$ for some node $x_1$ and a proof of $|d(u,x_1) - d(v,x_1)|\ge \rad(G)$. In the case $|d(u,x_1) - d(v,x_1)|=d(u,x_1)-d(v,x_1)$, we would have a shorter proof $d(u,x_1)\ge \rad(G)$ in contradiction with the choice of $v$. We thus continue with the case $|d(u,x_1) - d(v,x_1)|=d(v,x_1) - d(u,x_1)$. We then have a proof of $d(v,x_1)\ge \rad(G) + d(u,x_1)$. Either $x_1\in L$ or the proof uses a node $x_2$ such that $|d(v,x_2) - d(x_2,x_1)|\ge \rad(G) + d(u,x_1)$. The choice of $v$ again implies $|d(v,x_2) - d(x_2,x_1)|=d(v,x_2) - d(x_2,x_1)$ (otherwise $x_2$ would provide a shorter proof). By repeating this argument, we deduce that a shortest proof of $d(u,v)\ge \rad(G)$ uses a sequence $x_1,\ldots,x_p$ of nodes such that $d(v,x_i)\ge \rad(G) + d(u,x_1) + d(x_1,x_2) + \cdots + d(x_{i-1},x_i)$ for $i=1,\ldots,p$ and $x_p\in L$. Consider the antipode $a=\antipode_r(x_p)$. We then have $d(a,x_p)\ge d(v,x_p)\ge \rad(G) + d(u,x_1) + \cdots + d(x_{p-1},x_p)$. By triangle inequality, we have $d(u,a)\ge d(a,x_p) - d(u,x_p)$ and $d(u,x_p)\le d(u,x_1)+\cdots+d(x_{p-1},x_p)$. We thus have $d(u,a)\ge \rad(G)$. In all cases, $L\cup \antipode_r(L)$ must contain a node at distance $\rad(G)$ or more from $u$, it is thus a radius certificate.
\end{proof}

\lvbis{Note that a similar result does not hold for our diameter certificate definition. A one-to-all distance based algorithm for diameter could query a set $U$ of nodes such that for any pair $u,v$ there is $x\in U$ satisfying $d(u,x)+d(x,v)\le \diam(G)$ which implies  $d(u,v)\le\diam(G)$ by triangle inequality.
This is the subject of the next section in which  we study an interesting variation of diameter certificate.}


\section{Extended diameter certificates}
\label{sec:extended-diameter}

An \emph{extended diameter certificate} is defined as a set $X$ of nodes satisfying:
$$
\forall u\in V, \forall v\in V, \exists x\in X, d(u,x)+d(x,v)\le \diam(G).
$$
In contrast, $X$ is a diameter certificate if and only if $\forall u\in V, \exists x\in X, \forall v\in V, d(u,x)+d(x,v)\le \diam(G)$. Note that a diameter certificate is indeed an extended diameter certificate but the converse might not be true.

\lv{
We first show that the notion of extended diameter certificate is related to the minimum number of queries a one-to-all distance-based algorithm must perform.

\begin{theorem}\label{th:lbdiam}
    Given a graph $G$, if an execution of a one-to-all distance-based randomized Monte Carlo algorithm successfully computes its diameter by querying a set $X$ of nodes, then $X$ is an extended diameter certificate. Such a diameter algorithm thus requires at least $|X_{OPT}|$ one-to-all distance queries where $|X_{OPT}|$ is the minimum size of an extended diameter certificate.
\end{theorem}

\begin{proof}
    Consider any Monte Carlo algorithm for computing the diameter with positive probability. There must exist an execution with input $G$ that succeeds. Let $X$ denote the set of nodes queried for one-to-all distances in such an execution. A proof of correctness of the algorithm allows to conclude that, for any pair $u,v$ of nodes, we have $d(u,v)\le \diam(G)$ based on triangle inequality and the distances known to the algorithm, that is distances from nodes in $X$. We now show that $X$ is an extended diameter certificate. Given a pair $u,v$ of nodes, consider a proof $d(u,v)\le \diam(G)$ according to the correctness of the algorithm. If neither $u$ nor $v$ is in $X$, the proof must use a triangle inequality $d(u,v)\le d(u,x_1)+d(x_1,v)$ for some node $x_1$ and a proof of $d(u,x_1)+d(x_1,v)\le \diam(G)$. Either $x_1\in X$ or the proof obtains bounds of $d(u,x_1)$ and $d(x_1,v)$ based on other distances: using the existence of some nodes $x_2$ and $x_3$ and bounds of distances from $x_2$ and $x_3$, the triangle inequalities $d(u,x_1)\le d(u,x_2)+d(x_2,x_1)$ and $d(x_1,v)\le d(x_1,x_3)+d(x_3,v)$ lead to $d(u,x_1)+d(x_1,v)\le \diam(G)$. That is, it uses a proof of $d(u,x_2)+d(x_2,x_1) + d(x_1,x_3)+d(x_3,v)\le \diam(G)$. If neither $x_2\in X$ nor $x_3\in X$, we can repeat this argument until we have a sequence $u=y_0,y_1,\ldots y_k=v$ such that we have a proof of $d(y_0,y_1)+\cdots+d(y_{k-1},y_k)\le \diam(G)$ where some node $y_i$ with $i\in \{1,\ldots,k-1\}$ is in $X$. As triangle inequality implies $d(y_0,y_i)\le d(y_0,y_1)+\cdots+d(y_{i-1},y_i)$ and $d(y_i,y_k)\le d(y_i,y_{i+1})+\cdots+d(y_{k-1},y_k)$, the proof of correctness implies $d(y_0,y_i)+d(y_i,y_k)\le \diam(G)$, that is $d(u,y_i)+d(y_i,v)\le \diam(G)$ with $y_i\in X$.
\end{proof}
}

\begin{proposition}\label{prop:extended-diameter}
    Given a graph $G=(V,E)$ and an extended diameter certificate $X$ of $G$, we can compute its diameter in $O\left(|X|m +  2^{|X|}\binom{|X|+\ceil{\log n}}{|X|} n\log n\right)$ time.
\end{proposition}

Note that the above bound is subquadratic when $|X|=o(\log n)$ since $\binom{|X|+\ceil{\log n}}{|X|}=n^\eps 2^{O(|X|)}$ for any $\eps>0$ (see \cite{BringmannHM2020}).
\lvbis{Note also that checking that a set $X$ is an extended diameter certificate requires quadratic time in general (under SETH) when $X$ has size $\Omega(\log n)$ (see the reduction from SAT to diameter computation in \cite{RV13}). In contrast, the stronger requirement of our diameter certificate definition 
enables subquadratic-time verification that a set is indeed a diameter certificate as soon as its size is $o(n)$.}

\begin{proof}
    Perform a BFS traversal from each vertex $x\in X$ to obtain all distances $d(x,v)$ for $(x,v)\in X\times V$ in $O(|X|m)$ time. Let $x_1,\ldots,x_k$ denote the nodes of $X$ in an arbitrary order. We can then associate a point $p_v=(d(x_1,v),\ldots,d(x_k,v))$ in $\N^k$ to each vertex $v\in V$. Compute a data-structure enabling efficient orthogonal range queries on these points.
    Given a value $D$, we can test if $D$ is smaller or equal to the diameter of $G$ as follows. For each node $u$, we test whether the orthogonal range between points $(D-d(x_1,u),\ldots,D-d(x_k,u))$ and $(\infty,\ldots,\infty)$ is empty. If it is not, then there is a node $v$ such that $d(x,v)\ge D - d(x,u)$, or equivalently $d(u,x)+d(x,v)\ge D$, for all $x\in X$.
    However, since $X$ is an extended diameter certificate, there exists $x\in X$ such that $d(u,x)+d(x,v)\le diam(G)$, which thus implies $D\le\diam(G)$. Conversely, if $D\le \diam(G)$, then we can show that at least one such range is nonempty. Indeed, consider a diametral pair $a,b$ (i.e., satisfying $d(a,b)=\diam(G)$). By triangle inequality, we have $d(a,x)+d(x,b)\ge d(a,b)=\diam(G)\ge D$ for all $x\in X$. This implies $d(x,b)\ge D-d(x,a)$ for all $x\in X$, and $p_b$ is thus in the orthogonal range between $(D-d(x_1,a),\ldots,D-d(x_k,a))$ and $(\infty,\ldots,\infty)$. Thus, we can find the diameter of G using a binary search with $O(n\log n)$ range queries in total. A range tree on $n$ points in dimension $k$ can be computed in $O(n k B(n,k))$ time  where $B(n,k)=\binom{k+\ceil{\log n}}{k}$. It allows to answer if an orthogonal range query is empty in $O(2^k B(n,k))$ time. We use here the careful analysis of \cite{BringmannHM2020,Monier1980}, which holds even when $k$ is not constant. 
\end{proof}

\guillaume{Unfortunately, the ratio between the respective sizes of a smallest diameter certificate and a smallest extended diameter certificate can be arbitrarily large:
\begin{proposition}
    For a graph with $n$ nodes, where $n$ is an even number greater than two, the worst ratio between the respective sizes of optimal diameter and extended diameter certificates equals $n/2$.
\end{proposition}
\begin{proof}
    We first prove this ratio cannot exceed $n/2$.
    If an extended diameter certificate is reduced to one node, then it must be also a diameter certificate.
    Therefore, we may assume there are at least two nodes in a smallest extended diameter certificate.
    As in the worst case, a diameter certificate contains the whole set of nodes, its size is at most $n$, thus proving the ratio upper bound $n/2$. 
    In order to show this is sharp, let $G$ be a hyperoctahedron of size $n$ (complete graph minus a perfect matching).
    Since $G$ is self-centered (all its nodes have eccentricity equal to $2$), its only diameter certificate is the whole set of nodes.
    However, any two nonadjacent nodes form an extended diameter certificate.
\end{proof} 
}

\lv{As another example where both notions strikingly differ, consider a cycle $G$ of length $n=4k$ for some integer $k$. As already mentioned, the only diameter certificate of $G$ is the whole set of nodes. However, it has an extended diameter certificate of size four: if $G$ is the cycle $v_0,\ldots, v_{n-1}$, consider $X=\{v_0,v_{k},v_{2k},v_{3k}\}$. One can easily check that $X$ is an extended diameter certificate. For any pair $u,v$ of nodes, if there exists $x\in X$ on a shortest path from $u$ to $v$, we obviously have $d(u,v)= d(u,x)+d(x,v)\le \diam(G)=2k$. Otherwise, $u$ and $v$ must lie in-between two nodes of $X$, say $v_0$ and $v_k$ w.l.o.g, and we then have $d(u,v)\le d(u,v_0)+d(v_0,v)\le 2k=\diam(G)$.
}

%% file: body/smallcert-soda.tex
\part{Graphs with small certificates}
\label{part:smallcert}

\section{Computing the radius in graphs with small radius certificate}
\label{sec:small-rad-cert}

The algorithmic results presented in Sections~\ref{sec:rad} and~\ref{sec:diam} are based on the existence of certificates for radius and diameter of size the number of antipodes at most. However, the size of these certificates may be far from optimum. In this section, we further discuss the consequences of having small radius or diameter certificates.
As a starter, let us observe that for any graph $G = (V,E)$, for every positive integer $r$, there exists a {\em non-deterministic} linear-time algorithm that decides whether $\rad(G) \le r$: specifically, we choose non-deterministically a vertex $c$ of the graph, then we perform a BFS traversal with starting vertex $c$ in order to compute $e(c)$, and finally, we accept if and only if $e(c) \le r$. 
If $\rad(G) \le r$, then there always exists a choice of vertex $c$ such that the algorithm accepts ({\it e.g.}, for a central vertex $c$); conversely, if $\rad(G) > r$, then the algorithm must reject for any choice of vertex $c$.
Conversely, if $G$ has a radius certificate with $\ell(n)$ vertices at most, then we can non-deterministically decide in $O(\ell(n)m)$ time whether $\rad(G) > r$: specifically, we select non-deterministically a subset $L$ of $\ell(n)$ vertices at most, then we perform a BFS traversal for every vertex of $L$, and finally we accept if and only if $e_L(v) > r$ for every $v \in V$.
If $\rad(G) > r$, then as before there always exists a choice of subset $L$ such that the algorithm accepts ({\it e.g.}, if $L$ is a radius certificate of minimum size); conversely, if $\rad(G) \le r$, then for any subset $L$, for any central vertex $v$, we get $e_L(v) \le r$, and so the algorithm always rejects.
Note that similar results can be derived for diameter computation on graphs with a diameter certificate of size $\ell(n)$ at most.
Now, according to~\cite{CGIMPS16}, if a problem can be solved both non-deterministically and co-nondeterministically in $O(T(N))$ time, where $N$ denotes the input size, then a conditional lower bound in $\Omega(T(N)^{1+\gamma})$, for any $\gamma > 0$, cannot be proved based on the strong exponential-time hypothesis, unless some nondeterministic version of the latter can be falsified.
As a result, for any class of graphs ${\cal G}$, {\em the existence of sublinear certificates for radius (for diameter, resp.) is a barrier against quadratic-time lower bounds for radius computation (for diameter computation, resp.) under the strong exponential time hypothesis}. This result holds even if there is no known efficient algorithm for computing such certificates. However, we note that it does not immediately imply the existence of a truly subquadratic-time algorithm for computing the radius (the diameter, resp.) when restricted to graphs in ${\cal G}$.

Our main result in this section is as follows:
\begin{theorem}\label{thm:rad-cert-random}
    For every class of graphs ${\cal G}$ such that every graph with $n$ vertices has a radius certificate of size $\ell(n)$ at most, there exists a {\em randomized} (Monte Carlo) algorithm for computing the radius (and a radius certificate) within ${\cal G}$ in $O(\ell(n)m\log^3{n})$ time with high probability.
\end{theorem}

Our approach for proving Theorem~\ref{thm:rad-cert-random} is based on the existence of an efficient approximation algorithm for computing a hitting set in any family of coballs in a graph, where coballs $\ovB(v,r)$ are implicitly represented by the vertex-value pairs $(v,r)$.
%
This approach is indeed developed in \cite{SM10} where a set-cover of a universe $U$ with $N$ elements is computed even though the collection ${\cal C}$ of $M$ sets is not explicitly given. Instead their algorithm processes through queries on elements and sets: a {\em containment query} on $u\in U$ consists in listing the identifiers of all sets of ${\cal C}$ that contain it; a {\em subset query} on $S \in {\cal C}$ consists in listing the elements in $S$. 

\begin{theorem}[\cite{SM10}]\label{thm:covert-set-cover}
    A set-cover of size at most $O(OPT \cdot \log{(N+M)})$ can be computed with high probability using at most $O(OPT \cdot \log^2{(N+M)})$ containment and subset queries through a Monte Carlo algorithm.
\end{theorem}

We believe that their technique can be turned into a Las Vegas algorithm with few efforts, but leave it for future work. We nevertheless make a step in that direction in Section~\ref{sec:approx-rad-cert}.
Our proof, in what follows, is based on the observation that queries for coballs can be simulated using BFS traversals, namely:

\begin{proof}[of Theorem~\ref{thm:rad-cert-random}]
It suffices to present an $O(\ell(n)m\log^2{n})$-time algorithm that for any graph $G=(V,E)$ that is in ${\cal G}$, for any positive integer $r$, decides whether $\rad(G) \ge r$.
Indeed, we can compute $\rad(G)$ with $O(\log{n})$ calls to this algorithm, using binary search.

We observe that $\rad(G) \ge r$ if and only if there exists a hitting set of size $\ell(n)$ at most for the family of coballs $\{\ovB(v,r) : v \in V\}$.
In one direction, if $\rad(G) \ge r$ then any radius certificate can be selected as the desired hitting set. In particular, there is one with $\ell(n)$ vertices at most.
In the other direction, if $\rad(G) < r$ then $\ovB(c,r) = \emptyset$ for any central vertex $c$, which implies the nonexistence of a hitting set.
As a result, we are left presenting an algorithm that either computes a hitting set of size $O(\ell(n)\log{n})$ for $\{\ovB(v,r) : v \in V\}$, or correctly asserts that no such hitting set with $\ell(n)$ vertices at most can exist. We next explain how such an algorithm can be derived from Theorem~\ref{thm:covert-set-cover}.

More specifically, we set $U = V$, ${\cal C} = \{\ovB(v,r) : v \in V\}$.
Note that $N = M = n$.
We apply Theorem~\ref{thm:covert-set-cover}, but we add the restriction that at most $O(\ell(n)\log^2{n})$ queries can be performed (else, we abort).
By doing so, we either compute a set-cover ${\cal S} \subseteq {\cal C}$ of size $O(\ell(n)\log{n})$ at most, or we correctly assert that no such set-cover of size $\ell(n)$ at most can exist.
A containment query for $u \in V$ can be simulated using a BFS traversal with starting vertex $u$.
Similarly, a subset query for $\ovB(v,r)$ can be also simulated using a BFS traversal with starting vertex $v$.
Hence, the above procedure requires $O(\ell(n)m\log^2{n})$ time.
Furthermore, by duality, a hitting set $L$ can be transformed in the set-cover $\{\ovB(x,r) : x \in L\}$.
Therefore, if there is no set-cover of size $\ell(n)$ at most, then there is no hitting set of size $\ell(n)$ at most either.
Let us now assume the existence of a set-cover ${\cal S}$ of size $O(\ell(n)\log{n})$ at most.
Then, again by duality, $L = \{ x : \ovB(x,r) \in {\cal S} \}$ is a hitting set.
\end{proof}

\section{Computing the diameter in graphs with small diameter certificate}
\label{sec:small-diam-cert}

The existence of a similar result for diameter certificates remains open.
However, the following weaker result can be proved:
\begin{theorem}
    \label{thm:diam-cert-random}
    For every class of graphs ${\cal G}$ such that every graph with $n$ vertices has a diameter certificate of size $\ell(n)$ at most, there exists a {\em randomized} (Monte Carlo) algorithm for computing the diameter within ${\cal G}$ in $O(m\sqrt{\ell(n)n}\log^{3/2}{n})$ time with high probability.
\end{theorem}

In order to prove this above result, we need the following variant of~\cite[Lemma 2]{DuDr21-netw}:
\begin{lemma}\label{lem:approx-ecc-sampling}
    For every graph $G=(V,E)$, for any $\varepsilon \in (0,1)$, there is an algorithm that computes, for every vertex $v$, some value $r(v) \le e(v)$ such that $\left|B[v,r(v)]\right| \ge (1-\varepsilon)n$ with high probability. The algorithm runs in $O(\varepsilon^{-1}m\log{n})$ time with high probability.
\end{lemma}
\begin{proof}
If $\varepsilon \le \frac{2\log n}{n}$ we can perform a BFS traversal from each vertex in time $O(nm)=O(\varepsilon^{-1}m\log n)$ and get $e(v)$ for all $v\in V$, allowing to set $r(v)=e(v)$ with the desired complexity. Now assume $p=\frac{2\log n}{\varepsilon n} < 1$.
    Every vertex $v$ is added independently at random with probability $p$ in a vertex subset $L$.
    By Chernoff bounds, the size of $L$ is  $O(\varepsilon^{-1}\log{n})$ with high probability.
    Then, for every vertex $v$ of $G$, we set $r(v) = e_L(v)$.
    For that, it suffices to perform a BFS traversal for every vertex of $L$, which can be done in $O(m|L|) = O(\varepsilon^{-1}m\log{n})$ time with high probability.

    Let $v \in V$ be arbitrary. By construction we have $r(v) \le e(v)$.
    Conversely, for every $v \in V$, let $\rho(v)$ be the largest integer $k$ such that $\left|B[v,k]\right| < (1-\varepsilon)n$ (this value is independent from the random subset $L$).
    The probability of the event $\left\{ L \setminus B[v,\rho[v]] = \emptyset \right\}$ is at most $(1-p)^{\varepsilon n} = (1-p)^{\frac 1 p \cdot 2 \log{n}} \leq n^{-2}$.
    Therefore, by a union bound over all vertices, the property that $\left|B[v,r(v)]\right| \ge (1-\varepsilon)n$ for all vertices $v$ simultaneously must hold with probability at least $1 - n^{-1}$.
\end{proof}

We are now ready to prove Theorem~\ref{thm:diam-cert-random}.
Roughly, the proof follows from a combination of Theorem~\ref{thm:covert-set-cover} and Lemma~\ref{lem:approx-ecc-sampling}.

\begin{proof}[of Theorem~\ref{thm:diam-cert-random}]
Let $\varepsilon \in (0,1)$ to be fixed later in the proof.
First we apply Lemma~\ref{lem:approx-ecc-sampling} for $G,\varepsilon$. 
This can be done in $O(\varepsilon^{-1}m\log{n})$ time with high probability.
Furthermore, with high probability it results in the computation, for every vertex $v$ of $G$, of some value $r(v) \le e(v)$ such that $\left|B[v,r(v)]\right| \ge (1-\varepsilon)n$. 
Then in what follows, we present an $O(\varepsilon n \ell(n) m \log{n})$-time algorithm that for any graph $G=(V,E)$ that is in ${\cal G}$, for any positive integer $D$, decides whether $\diam(G) \le D$.
We can compute $\diam(G)$ with $O(\log{n})$ calls to this algorithm, using binary search.

Our algorithm works as follows.
We set  ${\cal C} = \{B[v,D-r(v)] : v \in V\}$.
Then, we apply Theorem~\ref{thm:covert-set-cover} to $V,{\cal C}$, but we add the restriction that at most $O(\ell(n)\log^2{n})$ queries can be performed (else, we abort).
By doing so, we either compute a set-cover ${\cal S} \subseteq {\cal C}$ of size $O(\ell(n)\log{n})$ at most, or we correctly assert that no such set-cover of size $\ell(n)$ at most can exist.
As it was argued before in the proof of Theorem~\ref{thm:rad-cert-random}, both containment queries and subset queries can be simulated using BFS traversals.
Hence, the above procedure requires $O(\ell(n)m\log^2{n})$ time.

Assume first there is no set-cover of size $\ell(n)$ at most.
We claim that $\diam(G) > D$.
Indeed, suppose by contradiction $\diam(G) \le D$, and let us consider a diameter certificate $U = \{u_1,u_2,\ldots,u_q\}$ of size $q \le \ell(n)$.
We have $V = \bigcup_{i=1}^q B[u_i,\diam(G)-e(u_i)]$ and for every $1 \le i \le q$ we get $D - r(u_i) \ge \diam(G)-e(u_i)$.
But then, $\{B[u_i,D-r(u_i)] : 1 \le i \le q\}$ should be a set-cover of size $\ell(n)$ at most, a contradiction. 

From now on, we assume there is a set-cover ${\cal S} = \{B[v_1,D-r(v_1)],B[v_2,D-r(v_2)],\ldots,B[v_q,D-r(v_q)]\}$ of size $q = O(\ell(n)\log{n})$.
For every $1 \le i \le q$, we compute $L_i = V \setminus B[v_i,r(v_i)]$.
This can be done in $O(mq) = O(m\ell(n)\log{n})$ time using BFS traversals.
Let $L = \bigcup_{i=1}^q L_i$. We compute $e_L(v)$, for every vertex $v$ of $G$, that can be done in $O(m|L|)$ time.
Recall that with high probability, for every $1 \le i \le q$ we have $|L_i| < \varepsilon n$, and so we have $|L| = O(\varepsilon \ell(n) n \log{n})$ with high probability.
Furthermore, we claim that $\diam(G) \le D$ if and only if we have $\max\{e_L(v) : v \in V\} \le D$.
In one direction, if $D \ge \diam(G)$, then $D \ge \diam(G) = \max\{e(v) : v \in V\} \ge \max\{e_L(v) : v \in V\}$.
In the other direction, let us assume that $\max\{e_L(v) : v \in V\} \le D$, and suppose, for the sake of contradiction, $\diam(G) > D$. 
Let $x,y \in V$ be such that $d(x,y) > D$.
Since ${\cal S}$ is a set-cover, we may assume without loss of generality that $x \in B[v_1,D-r(v_1)]$.
We must have $y \notin L_1$ because $e_{L_1}(x) \le e_L(x) \le D < d(x,y)$.
However, this implies that $y \in B[v_1,r(v_1)]$, and so,  by the triangle inequality, $d(x,y) \le d(x,v_1)+d(v_1,y) \le D - r(v_1) + r(v_1) = D < d(x,y)$, giving a contradiction.

The total running time of the algorithm is in $O(\varepsilon^{-1}m\log{n} + \log{n} \times \varepsilon n \ell(n) m \log{n})$ with high probability.
This is minimized for $\varepsilon = \Theta\left(\frac 1 {\sqrt{n\ell(n)\log{n}}}\right)$.
\end{proof}

Before concluding this section, the last two results are simple conditions ensuring the existence of constant-size certificates for radius or diameter.

\begin{proposition}\label{prop:d>2r-2}   For an arbitrary graph $G$ with $\diam(G)\ge 2\rad(G)-1$, any diametral pair of vertices $x,y$ forms a minimum radius certificate of $G$. Futhermore, $\rad(G)=\lfloor\frac{d(x,y)+1}{2}\rfloor$. 
\end{proposition}

\begin{proof} Let $x,y$ be an arbitrary diametral pair of $G$, i.e., $d(x,y)=\diam(G)$. If there is a vertex $z\in V$ with $\max\{d(z,x),d(z,y)\}\leq \rad(G)-1$ then $\diam(G)=d(x,y)\le d(z,x)+d(z,y)\leq 2\rad(G)-2$, and a contradiction arises. The equality $\rad(G)=\lfloor\frac{d(x,y)+1}{2}\rfloor$ comes from the hypothesis $\frac{\diam(G)+1}{2}\ge \rad(G)$ and from $\diam(G)\le 2\rad(G)$ which implies $\frac{\diam(G)+1}{2}\le \rad(G)+\frac{1}{2}$.
\end{proof}

By Theorem~\ref{thm:rad-cert-random}, the above condition implies the existence of an $O(m\log^3{n})$-time algorithm for computing the radius with high probability in that case.
In particular, let us briefly introduce a class of graphs such that the condition of Proposition~\ref{prop:d>2r-2} always holds. Namely, an {\em eccentricity-preserving spanning tree} of a graph $G=(V,E)$ is a spanning tree $T$ such that $e_T(v) = e_G(v)$ for every vertex $v$. Note that being given an efficient algorithm for computing an eccentricity-preserving spanning tree of $G$ when it exists, we could compute all eccentricities of $G$ in additional $O(n)$ time. But the existence of such an algorithm is open. However, it was observed in~\cite{NaP90} that graphs $G$ that admit an eccentricity-preserving spanning tree must satisfy $\diam(G) \ge 2\rad(G)-1$. Therefore, to our best knowledge, Theorem~\ref{thm:rad-cert-random}, in combination with Proposition~\ref{prop:d>2r-2}, implies the first almost-linear time algorithm for radius computation in this class of graphs. Other such examples are discussed in Section~\ref{Helly}.

\begin{proposition}\label{prop:d=2r-any}   For an arbitrary graph $G$ with $\diam(G)= 2\rad(G)$, any central vertex $c$ forms a smallest  diameter certificate of $G$. 
\end{proposition}

\begin{proof} In every graph $G$ with $\diam(G)=2\rad(G)$, for an arbitrary central vertex $c\in C(G)$ and every vertex $v\in V$,  we have $d(v,c)+e(c)\le e(c)+e(c)= 2\rad(G)=\diam(G)$.
\end{proof}

Note that on graphs $G$ such that $\diam(G) = 2\rad(G)$, their radius, and so their diameter as well, can be computed in $O(m\log^3{n})$ time with high probability.

\section{Approximating radius certificate}
\label{sec:approx-rad-cert}

We define a \emph{strict radius certificate} of $G=(V,E)$ as a set $L\subseteq V$ such that $e_L(v)\ge \rad(G)$ for all $v\in V$, and $e_L(v) > \rad(G)$ for all $v\in V\setminus C(G)$. In other words, a strict radius certificate is a radius certificate such that the centers are the only nodes whose lower bound matches the radius. The following result shows that it is possible to compute a radius certificate having size within a logarithmic factor from a minimum strict radius certificate.

\begin{theorem}
\label{th:strict-rad-cert}
Given an arbitrary graph $G$, its radius, a center, and a radius certificate of size $O(\ell^*\log n)$ can be computed using $O(\ell^*\log^2 n)$ one-to-all distance queries in $O(\ell^* m\log^2 n)$ time with high probability where $\ell^*$ is the minimum size of a \emph{strict} radius certificate.
\end{theorem}

The main idea  is to find a constant approximation of the size of a minimum strict radius certificate by binary search without relying on a guess of the radius $R$ of the graph. 
For that, we need to dive into the technique of \cite{SM10} and use a key property of the respective set-cover instances we consider for all possible values of $R$ to obtain a Las Vegas algorithm. 

\begin{proof}[of Theorem~\ref{th:strict-rad-cert}]
    Consider a graph $G=(V,E)$ and a strict radius certificate $L^*$ of $G$ with minimum size. Starting from an estimate $\ell$ of $|L^*|$, we try to construct a radius certificate $L$ as described below. If the construction fails, then we double our estimate $\ell$ and start again until the construction succeeds. Initially, we set $\ell=2$.

    Given the current estimate $\ell$, we grow a lower certificate $L$ which is initially empty. We use random sampling to find an appropriate vertex $x$ to add to $L$. First, choose a constant $c$ according to multiplicative Chernoff bounds so that the sum $S$ of $c\ell'\log n$ independent Bernoulli random variables taking value 1 with probability $\frac{1}{\ell'}$ satisfies $\frac{c\log n}{2} < S < 2c\log n$ with probability greater than $1-1/n^2$ for any integer $\ell'>1$. We randomly, uniformly and independently sample a set $Y$ of $2c\ell\log n$ nodes in $V$.
    We can then perform a BFS traversal from each sampled node $y\in Y$ and obtain distances $d(x,y)$ for all $(x,y)\in V\times Y$. In particular, we obtain the eccentricity $e(y)$ of all vertices $y\in Y$.

    As $L^*$ is a radius certificate, there must exist a node $x\in L^*$ hitting at least $2c\frac{\ell}{|L^*|}\log n$ coballs $\ovB(y,\rad(G))$ centered at vertices $y\in Y$ (we say that $x$ hits $\ovB(y,r)$ when $d(x,y)\ge r$). Moreover, for $\ell\ge |L^*|$, the choice of $c$ implies that $x$ hits $c\log n$ such coballs with probability greater than $1-1/n^2$. For a given set $U\subseteq V$ of vertices, a node $v\in V$ and a distance $r$, we define the number $h_U(v,r)$ of coballs of radius $r$ with center in $U$ that are hit by $v$. A key observation is that $h_U(v,r)$ is non-increasing in $r$. Using the selection algorithm of \cite{BlumFPRT1973} for finding the element with rank $c\log n$ in the vector of distances $d(v,y)$ for fixed $v$ and $y\in Y$, we can thus find the largest value $r_v$ such that $h_Y(v,r_v)\ge c\log n$. We then select a node $x$ such that $r_x=\max_{v\in V}r_v$, and add $x$ to $L$. Note that the choice of $c$ also implies that $x$ hits at least a fraction $\frac{1}{2\ell}$ of coballs of radius $r_x$ with probability greater than $1-1/n^2$. We now have an estimate $R=r_x$ of the radius of $G$. If $\ell\ge |L^*|$, recall that it satisfies $R\ge \rad(G)$ with probability greater than $1-1/n^2$.

    We iterate this operation with the set $W=\set{v\in V : e_L(v) < R}$ of vertices with lower bound less than our radius estimate $R$ by sampling nodes in $W$ and proceeding similarly as above. Each time we add a node $x$ to $L$, we update $L$ and $R=\min_{x\in L}r_x$. As a fraction at least $\frac{1}{2\ell}$ of coballs of radius $R$ and center in $W$ are hit by $x$ with probability greater than $1-1/n^2$, the size of $W$ shrinks by a factor at most $1-\frac{1}{2\ell}$ in each iteration. Note also that the nodes $v$ removed from $W$ by the update of $e_L(v)$ or the update of $R$ satisfy $e(v)\ge e_L(v)\ge r_x$.  After $2\ell\log n$ iterations, the size of $W$ is less than $\ell\log n$ with high probability. If it is not the case, we start again from $L=\emptyset$. (As it is a rare event, the complexity is not impacted.) Otherwise, we then perform a BFS from each remaining node  and add a furthest node $f$ of $v$ to $L$ for each $v\in W$. For such nodes $f$, we define $r_f=e(v)$ which ensures $r_f\ge \rad(G)$ and $e(v)=e_L(v)\ge r_f$. At this point, we check that we have $R= \min_{v\in V}e_L(v)$. If our estimate $\ell$ is at least $|L^*|$, we have $r_x\ge \rad(G)$ for all $x\in L$ with probability greater than $1-\frac{2\ell\log n}{n^2}$ by union bound, and we can then deduce $R\ge \rad(G)$. As each node $v$ gets removed from $W$ at some point, it satisfies $e(v)\ge e_L(v)\ge r_x$ for some $x\in L$, which implies $R\le \rad(G)$. If ever we observe $R\not=\min_{v\in V}e_L(v)$, we conclude that our estimate $\ell$ was too small, we double it, and start again from $L=\emptyset$.

    To confirm that the construction has succeeded, we try to find a center $c$ with a similar procedure using the fact that $L^*$ is a strict radius certificate. Consider the set $D=\set{v\in V : e_L(v)=R}$ of vertices with minimum lower-bound. We randomly, uniformly and independently sample a set $Y$ of $2c\ell\log n$ nodes in $D$, perform a BFS traversal from each $y\in Y$ to obtain distances $d(x,y)$ for all $(x,y)\in V\times Y$. If some vertex $y\in Y$ has eccentricity $e(y)=R$, we have succeeded: $y$ is a center, $R=\rad(G)$ and $L$ is a radius certificate. Otherwise, we have $e(y)>R$ for all $y\in Y$, and $Y$ is a uniform random sample of $D\setminus C(G)$. We can thus proceed similarly as above and find a node $x$ such that $c\log n$ vertices $y\in Y$ at least satisfy $d(x,y)>R$. If we do not find such a node, we again conclude that $\ell$ was too small and start from scratch with new estimate $2\ell$. Otherwise, the choice of $x$ implies that a fraction at least $\frac{1}{2\ell}$ of vertices $v\in D\setminus C(G)$ satisfy $d(x,v)>R$ with probability greater than $1-1/n^2$. These vertices are then removed from $D$ and we iterate until we either find a center, conclude that $\ell$ is too small or $D$ has size less than $\ell \log n$ after at most $2\ell \log n$ iterations with high probability. (If none of the three cases occurs, we start again from $L=\emptyset$.)
    In the latter case we can perform a BFS traversal for each remaining node and either find a center or conclude that all nodes have eccentricity larger than $R$, and conclude again that $\ell$ was too small.

    Overall, we stop with high probability as soon as $\ell\ge |L^*|$. We may (with very low probability) double $\ell$ several times, but the algorithms stops in any case if ever $\ell\ge n$ as we then sample all nodes. The algorithm thus always computes a radius certificate $L$ and a center $c$. With high probability, we have $|L|=O(|L^*|\log n)$ and $O(|L^*|\log^2 n)$ BFS traversals are performed in total. The rest of the computation is dominated by selection algorithm of \cite{BlumFPRT1973} which costs $O(n\ell\log n)=O(m\ell\log n)$ per sampling. The overall computation time is thus $O(m|L^*|\log^2 n)$ with high probability.
\end{proof}

\guillaume{
Roughly, in order to check whether the output $L$ of our construction is correct, we derive a test that either asserts that $L$ is indeed a radius certificate, \textit{or} that any \emph{strict} radius certificate must have a size larger than $\ell$ (though, in the latter case, it could still be the case that $L$ is a radius certificate). 
More generally, we can prove that deciding whether $L$ is a radius certificate cannot be done in truly subquadratic time, assuming the so-called Hitting Set Conjecture of~\cite{AWV16}.
Recall that the Hitting Set Conjecture posits that ''there is no $\varepsilon > 0$ such that for all $c \ge 1$, there is an algorithm that given two lists $A,B$ of $n$ subsets of a universe $U$ of size at most $c \log n$, can decide
in $O(n^{2-\varepsilon})$ time if there is a set in the first list that intersects every set in the second list''.
For a triple $A,B,U$ as above, define the graph $H\langle A,B,U \rangle$ with node set $A \cup B \cup U$, such that there is an edge between every set of $A \cup B$ and every element of $U$ that is contained in it.
The authors in~\cite{AWV16} embed $H\langle A,B,U \rangle$ in a graph $G$ with $3$ additional nodes, such that deciding in $O(n^{2-\varepsilon})$ time whether $\rad(G)$ is either $2$ or $3$ would falsify the Hitting Set Conjecture.
Now, let $L$ be constructed as follows: we add in $L$ an arbitrary node of $A$ and an arbitrary node of $B$; for the three additional nodes of $V(G) \setminus (A \cup U \cup B)$, and for every node of $U$, we compute a furthest node, which we also add in $L$.
By construction, $L$ has size at most $|U|+5 = O(\log{n})$, and it satisfies $e_L(v) \ge 2$ for every node $v$.
Hence, deciding whether $L$ is a radius certificate boils down to deciding whether the radius of $G$ is either $2$ or $3$.
}

\medskip

\guillaume{In light of this above hardness result, it is natural to ask what can be the worst ratio between the size of a smallest strict radius certificate and that of an optimal radius certificate.
Unfortunately, we prove this ratio to be unbounded, namely:
\begin{proposition}
    For a graph with $n$ nodes, where $n$ is an odd number greater than three, the ratio between the size of a smallest strict radius certificate and that of a smallest radius certificate can be as large as $(n-1)/2$.
\end{proposition}
\begin{proof}
    Let $G'$ be a hyperoctahedron (complete graph minus a perfect matching) with $n-1$ nodes.
    The graph $G$ is obtained from $G'$ by adding a universal node $u$.
    By construction, $\rad(G) = 1$.
    In particular, the union of $u$ with any other node forms a radius certificate of size two.
    However, $u$ is the unique center of $G$. This implies that in any strict radius certificate $L$, for every node $x$ of $G'$, there must be a furthest node from $x$ in $L$.
    Since conversely, every node of $G'$ is the unique furthest node of some node in $G'$ (namely, the one vertex to which it is nonadjacent), we obtain that every node of $G'$ must be in $L$.
\end{proof}
}

\section{Approximating all-eccentricity certificate}
\label{sec:approx-all-cert}

\begin{theorem}    
    \label{th:approx-all-cert}
    Given an arbitrary $n$-node graph $G$, all eccentricities, a tight lower certificate $L$ of size $O(\ell^* \log n)$ and the minimum tight upper certificate $U^\preceq$ can be computed in $O(m \ell^* \card{U^\preceq} \log^2 n)$ time with high probability where $\ell^*$ denotes the size of a minimum tight lower certificate.
\end{theorem}

We provide a Las Vegas algorithm by combining the greedy set-cover approach of Theorem~\ref{thm:rad-cert-random} with properties of the minimum tight upper certificate developed in Section~\ref{sec:min-tight-upper-cert}. More precisely, we rely on the following lemmas.

\begin{lemma}
    \label{lem:vertex-cert-prefix}
    Given a graph $G$, a vertex $u$ and a furthest node $f$ of $u$ (i.e. $e(u)=d(u,f)$), and a shortest path $u=v_1,v_2,\ldots,v_k=f$, all tight upper vertex-certificates for $u$ on the path form a prefix $v_1,\ldots,v_i$ for some $i\in [k]$.
\end{lemma}

\begin{proof}
    Suppose that for some $j\in [k]$, $v_j$ is a tight upper vertex-certificate for $u$, that is $e(u)=d(u,f)=d(u,v_j)+e(v_j)$. We thus have $e(v_j)= d(u,f) - d(u,v_j)=d(v_j,f)$ as $v_1,v_2,\ldots,v_k=f$ is a shortest path. Now consider $v_i$ with $i\le j$. By triangle inequality, we have $e(v_i)\le d(v_i,v_j)+e(v_j) = d(v_i,v_j) + d(v_j,f)=d(v_i,f)$. As $v_i,\ldots,v_k=f$ is a shortest path, we also have $e(v_i)\ge d(v_i,f)$ and conclude $e(v_i)=d(v_i,f)$. Finally, as $v_i$ is on a shortest path from $u$ to $f$, we obtain $d(u,f)=d(u,v_i)+d(v_i,f)=d(u,v_i)+e(v_i)$. This allows to conclude.
\end{proof}

\begin{lemma}
    \label{lem:vertex-cert}
    Given a graph $G$, a vertex $u$ and a furthest node $f$ of $u$ (i.e. $e(u)=d(u,f)$), we have $u\in U^\preceq$ if and only if $e(z)>d(z,f)$
    for all vertices $z\in I(u,f)\setminus\{u\}$ where $I(u,f)=\set{v:d(u,v)+d(v,f)=d(u,f)}$ is the interval between $u$ and $f$.
\end{lemma}

\begin{proof}
    If $u\notin U^\preceq$, then there exists a tight upper vertex-certificate $z\in U^\preceq$ for $u$ by Proposition~\ref{prop:upperopt}. It thus satisfies $e(u)=d(u,z)+e(z)=d(u,f)$ which implies $d(u,f)\ge d(u,z)+d(z,f)$ and thus $d(u,f) = d(u,z)+d(z,f)$ and $d(z,f)=e(z)$ by triangle inequality. We therefore get 
    $z\in I(u,f)$.
    Moreover, it satisfies $e(z)\le d(z,f)$ and $z\not=u$ since $z\in U^\preceq$ and $u\notin U^\preceq$.

    Conversely, assume $u\in U^\preceq$ and consider a vertex $z\in I(u,f)\setminus\{u\}$.
    We thus have $d(u,f)=d(u,z)+d(z,f)$. By definition of $U^\preceq$, $u$ is maximal for $\preceq$, and we cannot have $u\preceq z$, that is $e(z)\not= e(u)-d(u,z)$. As $e(u)-d(u,z)=d(z,f)\le e(z)$, we deduce $e(z)>d(z,f)$.
\end{proof}

\begin{lemma}
  \label{lem:vertex-cert-comput}
  Given a graph $G=(V,E)$ with minimum tight lower certificate $L^*$, and a vertex $u\in V$, a tight upper vertex-certificate $z$ for $u$ can be computed in $O(m|L^*|\log^2n)$ time 
  with probability greater than $1-O(1/n^2)$.
\end{lemma}

\begin{proof}
  Consider a furthest node $f$ of $u$ (i.e. $e(u)=d(u,f)$) and the interval $I(u,f)=\set{v:d(u,v)+d(v,f)=d(u,f)}$.
  This interval can be identified by performing two BFS traversals, one from $u$, and one from $f$.
  We search for a tight upper vertex-certificate $z\in I(u,f)$ for $u$ such that $r=d(z,f)$ is minimal. Note that a vertex $z\in I(u,f)$ satisfies $d(u,f)=d(u,z)+d(z,f)$ and it is therefore a tight upper vertex-certificate for $u$ if and only if it has eccentricity $e(z)=d(z,f)$. We thus search for the minimum distance $r$ such that some vertex $z\in I(u,f)$ at distance $r$ from $f$ satisfies $e(z)=d(z,f)=r$.  We use binary search to find $r$ by building a tight lower certificate $L'$ as follows. Starting with an estimate $\ell=2$ of $|L^*|$, $L'=\emptyset$ and a distance $r_1\in [0,d(u,f)]$ we consider the set $W\subseteq I(u,f)$ of vertices $v$ that are at distance $r_1$ from $f$ and such that $e_{L'}(v)\le r_1$ (we have initially $W=\set{v\in I(u,f) : d(v,f)=r_1}$). We randomly, uniformly and independently sample a set $Y$ of $2c\ell\log n$ nodes in $W$ 
  where $c$ is chosen according to multiplicative Chernoff bounds so that the sum $S$ of $c\ell'\log n$ independent Bernoulli random variables taking value 1 with probability $\frac{1}{\ell'}$ satisfies $\frac{c\log n}{2} < S < 2c\log n$ with probability greater than $1-1/n^4$ for any integer $\ell'>1$. If $|W|<2c\ell\log n$, we use $Y=W$. We perform a BFS traversal from each sampled node $y\in Y$ and obtain distances $d(x,y)$ for all $(x,y)\in V\times Y$. In particular, we obtain the eccentricity $e(y)$ of all vertices $y\in Y$.

  If some vertex $y\in Y$ has eccentricity $e(y)=d(y,f)=r_1$, then $y$ is a tight upper vertex-certificate for $u$ and we know $r\le r_1$. Otherwise, we have $e(y) > r_1$ for all $y\in Y$. We then count for each node $x\in V$ the number of vertices $y\in Y$ for which it is a furthest node, that is when $e(y)=d(x,y)$. The definition of $L^*$ implies that some node $x\in L^*$ is a furthest node of at least $2c\frac{\ell}{|L^*|}\log n$ vertices of $Y$. If $\ell\geq |L^*|$ it is the case for at least $c\log n$ vertices  with probability greater than $1-1/n^4$ by the choice of $c$. If we cannot find such a vertex $x$, we conclude that our estimate $\ell$ is smaller than $|L^*|$ and restart from scratch with estimate $2\ell$. If we do find a vertex $x$ that is a furthest node of at least $c\log n$ vertices of $Y$, we add it to $L'$. Note that the choice of $c$ also implies that $x$ is a furthest node of a fraction at least $\frac{1}{2\ell}$ of $W$ with probability greater than $1-1/n^4$. Note that those vertices $v$ may be removed from $W$ after adding $x$ to $L'$ and updating lower-bounds $e_{L'}(v)$ for $v\in I(u,f)$ if $e_{L'}(v)>r_1$. In particular, all vertices from $Y$ are removed. We iterate this sampling process until we either find $y$ with eccentricity $e(y)=d(y,f)=r_1$, or restart with estimate $2\ell$, or $W$ gets size less than $2c\ell\log n$. This latter case occurs after $2\ell \log n$ iterations at most and we then compute the eccentricities of all remaining vertices in $W$. Finally, if no vertex $y$ with eccentricity $e(y)=d(y,f)=r_1$ is found, we have $e(y)>d(y,f)$ for all vertices of $I(u,f)$ at distance at most $r_1$ from $f$ by Lemma~\ref{lem:vertex-cert-prefix}, and we know $r>r_1$.

  We then repeat this process for appropriate distances $r_2,\ldots,r_k$.
  For $\ell\geq |L^*|$, we find the value of $r$, and a vertex $y\in I(u,f)$ satisfying $e(y)=r=d(y,f)$, after $k=O(\log \diam(G))$ probes for finding $r$ using $O(c\ell\log^2n)$ BFS traversals with probability greater than $1-O(\frac{\log n}{n^3})$. Moreover, the node $y$ found is in $U^\preceq$ by Lemma~\ref{lem:vertex-cert}. The overall number of traversals for all trials with estimates $\ell=2,4,\ldots$ is $O(|L^*|\log^2n)$ with probability greater than $1-O(1/n^2)$. 
\end{proof}

\begin{proof}[of Theorem~\ref{th:approx-all-cert}]
    Let $L^*$ be a minimum tight lower certificate, and let $\ell^*=|L^*|$ be its size. We use binary search to find a value $\ell\ge \ell^*$. Starting with $\ell=2$, we grow a lower certificate $L$ and an upper certificate $U$ as follows (both are initially empty). 
    Let $W$ denote the set of nodes $v$ with non-matching lower and upper bounds, i.e. satisfying $e_L(v)<e^U(v)$ ($W=V$ initially as we consider that $e^U(v)=\infty$ when $U=\emptyset$). Randomly, uniformly and independently sample a set $Y$ of $2c\ell\log n$ nodes in $W$ where $c$ is chosen similarly as in the proof of Lemma~\ref{lem:vertex-cert-comput}.
    We perform a BFS traversal from each sampled node $y\in Y$ and obtain distances $d(x,y)$ for all $(x,y)\in V\times Y$. In particular, we obtain the eccentricity $e(y)$ of the $2c\ell\log n$ vertices $y\in Y$.

    As long as we find $2c\ell\log n$ vertices $y$ such that $e_L(y)<e(y)$, we proceed as in Lemma~\ref{lem:vertex-cert-comput} by adding to $L$ a vertex $x$ that is a tight lower vertex-certificate for at least $c\log n$ vertices in $Y$. With probability greater than $1-1/n^2$, it is also a tight lower vertex-certificate for a fraction at least $\frac{1}{2\ell}$ of vertices in $W$ with untight lower-bound. Again, if no such vertex is found, we double our estimate $\ell$ of $\ell^*$ and restart with empty certificates.

    As soon as we sample a node $y$ such that $e(y)=e_L(y)$, we apply Lemma~\ref{lem:vertex-cert-comput} with $u=y$ to find a tight upper vertex-certificate $z$ for $y$ using at most $O(\ell^*\log^2 n)$ BFS traversals with probability greater than $1-O(1/n^2)$, and add $z$ to $U$. We then update $W$ by removing vertices satisfying $e_L(v)=e^U(v)$. Note that we remove $y$ in particular. We then continue sampling until we find $2c\ell\log n$ vertices with untight lower bound. If ever the size of $W$ is less than $2c\ell\log n$, we finish by performing a BFS traversal from each remaining vertex $v$ and adding a furthest vertex of $v$ to $L$ for each $v\in W$.

    Similarly as in the proof of Lemma~\ref{lem:vertex-cert-comput}, we overall perform $O(\ell^*\log^2n)$ BFS traversals with high probability for the computation of $L$. However, each time we sample a vertex $u$ with tight lower-bound we perform $O(\ell^*\log^2n)$ BFS traversals with  probability greater than $1-O(1/n^2)$ to find a tight upper vertex-certificate $z\in U^\preceq$ for $u$. Note that $z$ is added only once to $U$. The reason is that if there exists $u'\in W$  such that $z$ is also a tight upper vertex-certificate for $u'$, we have $e^U(u')=e(u')$ after adding $z$ to $U$. Although $u'$ may still be in $W$ when $e_L(u')<e(u')$, it will be removed as soon as $L$ is modified so that $e_L(u')=e(u')$. The computation of $U$ thus costs $O(\ell^*\card{U^\preceq}\log^2n)$ BFS traversals overall with high probability.
\end{proof}






%% file: body/approx-certificate.tex
\section{Certificates for diameter approximation}
\label{sec:approximate-diameter}

Given a value $c\ge 1$, a \emph{$c$-approximate diameter certificate} of a graph $G=(V,E)$ is defined as a set $X$ of nodes satisfying:
$$
\forall u\in V, e^X(v)\le c\cdot\diam(G).
$$
In particular, our concept of diameter certificate corresponds to that of $1$-approximate diameter certificate.
Note that any singleton $\set{x}$ with $x\in X$ is a 2-approximate diameter certificate since $e^{\set{x}}(v)= d(x,v)+e(x)\le 2\diam(G)$ for all $v\in V$.
Let us first link this notion of approximate certificate with diameter approximation.

\begin{proposition}\label{prop:diam-approx}
    Given a $c$-approximate diameter certificate $X$ of a graph $G=(V,E)$, we can compute in deterministic $O(m|X|)$ time a $c$-approximation of its diameter, that is a value $D$ such that $\diam(G)\le D\le c\cdot\diam(G)$.
\end{proposition}

\begin{proof}
    Perform BFS traversals from nodes of $X$ to obtain $e^X(v)$ for all $v\in V$. We then set $D=\max\set{e^X(v) : v\in V}$. By definition, we have $e^X(v)\le c\cdot \diam(G)$ for all $v\in V$, implying $D\le c\cdot \diam(G)$. Moreover, if $a$ is any diametral node, we have $\diam(G)=e(a)\le e^X(a)\le D$.
\end{proof}

We show that any graph of even diameter has a sublinear $\frac{3}{2}$-approximate diameter certificate.

\begin{theorem}\label{th:approx-cert-even}
    Any graph $G=(V,E)$ with even diameter has a $\frac{3}{2}$-approximate diameter certificate $X$ of size $O(\sqrt{n}\log n)$. Moreover, such a set $X$ and a $3/2$-approximation of the diameter can be computed in randomized $O(m\sqrt{n}\log n)$ time by using a Monte Carlo algorithm. 
\end{theorem}

The following proof is a variation of the diameter approximation approach in \cite{AingworthCIM1999,RV13,ChechikLRSTW2014}.

\begin{proof}
    For $p\le n$, define the \emph{$p$-nearest set} $N_p(v)$ as the $p$ closest vertices to $v$, breaking ties arbitrarily. More precisely, $|N_p (v)| = p$ and for all pairs of vertices $x \in N_p (v)$ and $y \notin N_p (v)$, we have $d(v, x) \le d(v,y)$. In the sequel, we use $p=\ceil{\sqrt{n}}$.

    Pick a random sample $S$ of the vertices of size $\Theta(n/p \log n)=\Theta(\sqrt{n}\log n)$. With high probability, $S$ hits all the $p$-nearest sets $N_p(v)$, that is $S\cap N_p(v)\not=\emptyset$ for all $v\in V$. Let $w$ be a vertex maximizing $d(w,S)=\min_{x\in S} d(w,x)$. The proof follows from the two following claims combined with \Cref{prop:diam-approx}.

    \begin{claim}
        If $d(w,S)\le \diam(G)/2$, then $S$ is a $\frac{3}{2}$-approximate diameter certificate.
    \end{claim}

    As $w$ maximizes the distance from $S$, for any node $v$, there exists $x\in S$ such that $d(v,x)\le \diam(G)/2$. We thus have $e^S(v)\le d(v,x) + e(x)\le \diam(G)/2 + \diam(G) = \frac{3}{2} \diam(G)$. 

    \begin{claim}\label{claim:large-dist}
        If $d(w,S) > \diam(G)/2$, then $N_p(w)$ is a $\frac{3}{2}$-approximate diameter certificate.
    \end{claim}

    Since $S$ must hit all the $p$-nearest sets, there exists $x\in S\cap N_p(w)$. We thus have $d(w,x)\ge d(w,S)>\diam(G)/2$, and so $N_p(w)$ contains $B(w,\diam(G)/2)$. Now consider a vertex $v$ and a shortest path $P$ from $w$ to $v$. Let $\ell=d(w,v)$ denote the length of $P$. If $\ell\le \diam(G)/2$, we have $e^{N_p(w)}(v)\le d(v,w)+e(w)\le \diam(G)/2+\diam(G)=\frac{3}{2}\diam(G)$. Otherwise, let $x$ be the vertex of $P$ at distance $\diam(G)/2$ from $w$ so that $d(v,x)=\ell-\diam(G)/2\le \diam(G)/2$. We then have $e^{N_p(w)}(v)\le d(v,x)+e(x)\le \diam(G)/2+\diam(G)=\frac{3}{2}\diam(G)$.

\end{proof}

Similarly to \cite{ChechikLRSTW2014}, we can extend the above result to show that any graph has a $\frac{3}{2}$-approximate diameter certificate of size $O(\sqrt{m}\log n)$ that can be computed in $O(m^{3/2})$ time. If the graph has constant degree, it suffices to include neighbors of $N_p(w)$ in the certificate of Claim~\ref{claim:large-dist}. Otherwise, the idea is to work in the weighted setting and lower the maximum degree of the graph $G$ to at most 3 by processing each vertex $u$ of degree $k>3$ one after another as follows. Replace $u$ by a cycle of length $k$. Each cycle edge has weight 0. Associate each neighbor $v$ of $u$ to a distinct vertex $c_v$ of the cycle and replace each edge $uv$ with $c_vv$.
Distances in the resulting graph $G'$ are preserved in the following sense: if $c$ and $c'$ are cycle nodes resulting from the replacement of $v$ and $v'$ respectively, we have $d_{G'}(c,c')=d_G(v,v')$. A $\frac{3}{2}$-approximate diameter certificate for $G'$ thus yields one for $G$ by replacing each cycle node by the original vertex it was created from. One can easily check that $G'$ has $\Theta(m)$ vertices and $\Theta(m)$ edges, inflating the certificate size from $O(\sqrt{n}\log n)$ to $O(\sqrt{m}\log n)$ and the running time from $O(m\sqrt{n}\log n)$ to $O(m^{3/2}\log n)$. This yields the following result.

\begin{theorem}\label{th:approx-cert}
    Any graph has a $\frac{3}{2}$-approximate diameter certificate $X$ of size $O(\sqrt{m}\log n)$. Moreover, such a set $X$ and a $3/2$-approximation of the diameter can be computed in randomized $O(m^{3/2}\log n)$ time by using a Monte Carlo algorithm. 
\end{theorem}

We ask whether this can be generalized to higher approximation ratios. 

\begin{question}
    For integral $k>0$, does any graph has a $(2-\frac{1}{k})$-approximate diameter certificate of size $\widetilde O(m^{1/k})$?
\end{question} 

Note that the lower bound of \cite{DaLiWi} rules out the existence of certificates of size $n^{1/k-\eps}$ for any $\eps>0$ under SETH and NSETH as \cref{prop:diam-approx} would then yield a non-deterministic algorithm running in $O(mn^{1/k-\eps})$ time.

\guillaume{Related to the concept of $c$-approximate diameter certificate is that of distance $d$ domination. Indeed, for any graph $G$, any distance $(1 - \frac 1 k)\diam(G)$ dominating set is also a $(2-\frac{1}{k})$-approximate diameter certificate. The existence of a distance $\diam(G)/2$ dominating set of size $O(\sqrt{n\log{n}})$ has long been known~\cite{gavoille2001small}. It would be interesting to prove sublinear bounds on the size of $(1 - \frac 1 k)\diam(G)$ dominating sets for larger values of $k \ge 3$.}

%% file: body/antipodes-journal.tex
\part{Practical algorithms}
\label{part:antipodes}

\section{Radius computation and certification}
\label{sec:rad}

We now propose a radius algorithm with complexity parameterized by the number of antipodes in the input graph. Similarly to previous algorithms~\cite{TK11,BCHKMT15}, it maintains lower bounds on eccentricities of all nodes and performs one-to-all distance queries from nodes with minimal lower bound as a first ingredient. Similarly to the two-sweeps and four-sweeps heuristics~\cite{H73,MLH09,CGHLM13}, it performs one-to-all distance queries from antipodes of previous query source as a second ingredient. However, contrarily to these heuristics, it iterates until an exact solution is obtained (together with a radius certificate).

\smallskip

The idea of the algorithm is to maintain both a set $K$ of nodes with distinct antipodes and a lower certificate $L$ (initally empty). We iteratively select a node $u$ with minimal lower-bound $e_L(u)$ and perform a one-to-all distance query from $u$. As long as this bound is not tight (i.e., $e_L(u) < e(u)$), we add $\antipode_r(u)$ to $L$ and $u$ to $K$ while eccentricity lower-bounds are improved accordingly. (The fact that the bound is not tight implies that no antipode of $u$ could previously be in $L$.) As soon as the bound is tight (i.e., $e_L(u)=e(u)$), we then claim that $u$ is a center (i.e., its eccentricity is minimal) and return $e(u)$ as the radius and $L$ as radius certificate. Algorithm~\ref{alg:rad} formally describes the whole method.

Note the primal-dual flavor of this algorithm as the set $K$ (which has same size as $L$) is a packing for $\set{\antipode_r^{-1}(u) : u\in V}$ which is a restricted collection of $\set{\ovB(u,\rad(G)) : u\in V}$ for which the computed certificate $L$ is a covering.

\medskip
\begin{algorithm2e}[H]
  \Input{A graph $G$ and a ranking $r$ of its node set $V$.}
  \Output{The radius $\rad(G)$ of $G$, a center $c$ and a radius certificate $L$.}
  $L := \emptyset$\ \Comment{Lower certificate (tentative covering with $\set{\ovB(u,\rad(G)) : u\in V}$).\hspace{-1mm}}
  Maintain $e_L(v)=\max_{x\in L}d(v,x)$ (initially 0) for all $v\in V$.\\
  $K := \emptyset$\ \Comment{Packing for $\set{\antipode_r^{-1}(u) : u\in V}$.}
  
  \Do{$\min_{u\in V}e_L(u) < \min_{u\in K}e(u)$.}{
    Select $u\in V$ such that $e_L(u)$ is minimal.\\
    $D_u := \distfrom(G, u)$ \quad\Comment{Distances from $u$.}
    $e(u):=\max_{v\in V}D_u(v)$ \quad\Comment{Eccentricity of $u$.}
    \eIf{$e(u) = e_L(u)$}{
      \return{$e(u)$, $u$, and $L$}
    }{
      $a := \argmax_{v\in V}(D_u(v), r(v))$\quad\Comment{Antipode of $u$ for $r$.}
      $D_a:=\distfrom(G,a)$ \quad\Comment{Distances from $a$.}
      $K := K\cup\set{u}$\\
      $L := L\cup\set{a}$\\
      \lFor{$v\in V$}{$e_L(v):=\max(e_L(v), D_a(v))$}\\
    }
  }
  
  \Return{$e(c)$, $c$ and $L$ where $c=\argmin_{u\in K}e(u)$.}
  \caption{Computing the radius, a center and a radius certificate.}
  \label{alg:rad}
\end{algorithm2e}

\begin{theorem}
  \label{th:rad}
  Given a graph $G$ and a ranking $r$ on its node set $V$, Algorithm~\ref{alg:rad} computes its radius $\rad(G)$, a center $c$ and a radius certificate $L\subseteq \antipode_r(V)$ with $2\card{L}+1 = O(\card{\antipode_r(V)})$ one-to-all distance queries in $O(m\card{\antipode_r(V)})$ time. 
\end{theorem}

\begin{proof} 
  We first prove the termination of Algorithm~\ref{alg:rad}. Consider an iteration where we add the antipode $a$ of the selected node $u$ to $L$. We cannot have $a\in L$ as we would then have $e_L(u)=e(u)$ which is the termination case. In other words, nodes added to $K$ have distinct antipodes and $K$ is a packing for $\set{\antipode_r^{-1}(u) : u\in V}$. As long as the do-while loop runs, each iteration adds a new node to $L$. If ever we reach the point where $L=\antipode_r(V)$, then the lower-bound of each node $u\in V$ is tight: $e_L(u)=e(u)$. The next iteration must then terminate.
The complexity is straightforward: at most $2\card{L}+1$ one-to-all distance queries are performed and $\card{L}\le \card{\antipode_r(V)}$ as $L\subseteq \antipode_r(V)$.

  We now prove the correctness of Algorithm~\ref{alg:rad}. Consider an iteration of the do-while loop. By the choice of $u$, we then have $e_L(v)\ge e_L(u)$ for all $v\in V$. If the termination case $e_L(u)=e(u)$ occurs, we have $e(v)\ge e_L(v)\ge e_L(u) = e(u)$ for all $v\in V$. This ensures that $u$ has minimum eccentricity (it is a center). We thus have $\rad(G)=e(u)$ and $L$ is a radius certificate as $e_L(v)\ge \rad(G)$ for all $v\in V$.
  Finally, if ever the condition for continuing the do-while loop is false, we have $\min_{u\in V}e_L(u) \ge \min_{u\in K}e(u)$. For $c=\argmin_{u\in K}e(u)$, we thus have $\rad(G)=\min_{u\in V}e(u)\ge \min_{u\in V}e_L(u) \ge e(c)\ge \rad(G)$. That is, $L$ is a radius certificate and $c$ is a center.
\end{proof}

In practice, we observe very fast convergence compared to $\card{\antipode_r(V)}$ (see Section~\ref{sec:exp}).  
We can give the following argument for that. The node $u$ selected at each iteration satisfies $e_L(u)=\min_{v\in V} e_L(v)\le \min_{v\in V} e(v)\le \rad(G)$. We thus have $\max_{x\in L} d(u,x)\le \rad(G)$, that is $u\in \cap_{x\in L} B[x,\rad(G)]$. It appears that the eccentricity of antipodes is generally large compared to radius in practical graphs, and the set $\cap_{x\in L} B[x,\rad(G)]$ tends to quickly shrink toward the set of centers as we add antipodes to $L$.

\section{Minimum eccentricity selection}
\label{sec:selection}

The core of the above radius algorithm is a general technique depending on a user-defined function $f$ that we call \emph{minimum eccentricity selection} (minES) for $f$. It is a procedure that returns a node with minimum eccentricity with respect to $f$. Its amortized complexity is low in graphs with few antipodes. More precisely, for a given graph $G$ and a function $f$ that maps a node $v$ and an estimation $\ell$ of $e(v)$ to a value, it provides a function $\argminecc$ returning a node $u$ such that $f(u,e(u))$ is minimum as long as $f$ is non-decreasing, i.e., $f(v,\ell)\le f(v,\ell')$ for $\ell\le \ell'$ for all $v$.  A similar function $\minecc$ returns the value of $f(u,e(u))$ for such node $u$. The challenge here is to avoid the computation of all eccentricities. 

\medskip
We implement such a selection by maintaining lower bounds of all eccentricities as in Algorithm~\ref{alg:rad} and by using these lower bounds as estimates for true eccentricities. When the selection procedure is called, a node $u$ which is minimum according to lower bounds is considered. Such a node is found by evaluating $f(v,e_L(v))$ for all $v\in V$ where $e_L(v)$ denotes the lower bound stored for a node $v$. A one-to-all distance query from $u$ is then performed. If its eccentricity happens to be equal to its lower-bound $e_L(u)$ we claim that $f(u,e(u))$ is minimum and return that node. Otherwise, the antipode of $u$ is used to improve lower bounds before trying again. Algorithm~\ref{alg:selection} formally describes this.

\begin{algorithm2e}[ht]
  $L := \emptyset$\quad\Comment{Lower certificate.}
  Maintain $e_L(v)=\max_{x\in L}d(v,x)$ (initially 0) for all $v\in V$.\\
  \Function{$\argminecc(G,r,L,e_L,f)$}{
    \Repeat{}{
      $u := \argmin_{v\in V} f(v,e_L(v))$\\
      $D_u := \distfrom(G, u)$ \quad\Comment{Distances from $u$.}
      $e(u):=\max_{v\in V}D_u(v)$ \quad\Comment{Eccentricity of $u$.}
      \eIf{$e_L(u) = e(u)$}{
        \return{$u$}
      }{
        $a := \argmax_{v\in V}(D_u(v), r(v))$\quad\Comment{Antipode of $u$ for $r$.}
        $D_a:=\distfrom(G,a)$ \quad\Comment{Distances from $a$.}
        $L:=L\cup\set{a}$\\
        \lFor{$v\in V$}{$e_L(v):=\max(e_L(v), D_a(v))$}\\
      }
    }
  }
  \Function{$\minecc(G,r,L,e_L,f)$}{
    $u := \argminecc(G,r,L,e_L,f)$\\
    \Return{$f(u,e_L(u))$}
  }
  \caption{Minimum eccentricity selection with respect to function $f$.}
  \label{alg:selection}
\end{algorithm2e}

\begin{proposition}\label{prop:selection}
  Given a graph $G$ and a ranking $r$ of its node set $V$, we consider a function $f$ such that $f(v,\ell)$ can be evaluated for any $v\in V$ and $\ell\le e(v)$. If $f(v,.)$ is non-decreasing for all $v\in V$, i.e., $f(v,\ell)\le f(v,\ell')$ for $\ell\le \ell'$, function $\argminecc$ of Algorithm~\ref{alg:selection}  returns a node $u$ such that $f(u,e(u))$ is minimal and updates the lower certificate $L$ such that $e_L(u)=e(u)$ and $f(v,e_L(v))\ge f(u,e(u))$ for all $v\in V$. Moreover it can perform $k$ computations of $\argmin f(u,e(u))$ using $k + 2\card{L'}$ one-to-all distance queries and $(k+2\card{L'})n$ calls to $f$ where $L'\subseteq \antipode_r(V)$ denotes the set of nodes added to $L$. It runs in $O(m\card{\antipode_r(V)})$ time when $f(v,\ell)$ can be evaluated in $O(\deg(v))$ time.
\end{proposition}


\begin{proof}
The correctness of the selection comes from the fact that $f(v,.)$ is non-decreasing: if $e_L(u) = e(u)$, we then have
$f(u,e(u))=f(u,e_L(u))\le \min_{v\in V}f(v,e_L(v))\le \min_{v\in V}f(v,e(v))$.
The case $e_L(u) < e(u)$ can only occur if the antipode of $u$ was not in $L$
and happens at most $\card{\antipode_r(V)}$ times in total. In particular, each call to $\argminecc$ terminates.
If an algorithm makes $x$ calls to the $\argminecc$,
the number of successful iterations where $e_L(u) = e(u)$ is precisely $x$
while the number of unsuccessful iterations is at most the number of nodes added to $L$.
For each such iteration we perform 2 one-to-all distance queries instead of 1.
The total number of queries is thus $k+2\card{L'}$.
In all cases, we perform $n$ calls to $f$ per iteration: one for each node $v$. Assuming that $f(v,\ell)$ can be computed in $O(\deg(v))$ time, this costs $O(m)$ per iteration. The overall time required is thus in $O(m\card{\antipode_r(V)})$.
\end{proof}

\smallskip

As an example of usage, Algorithm~\ref{alg:rad} for radius is equivalent to the following algorithm using our minimum eccentricity selection for the basic function $v,\ell \mapsto \ell$.

\smallskip
\begin{algorithm2e}[H]
  $L:=\emptyset$; $e_L(v):=0$ for all $v\in V$.\\
  \lFunction{$\ecc(v,\ell)$:}{
    \return{$\ell$}
  }\\
  $c := \argminecc(G, r, L, e_L, \ecc)$\\
  \return{$e_L(c),c,L$}
  \caption{Radius computation through minimum eccentricity selection with respect to $f(v,\ell)=\ell$ (equivalent to Algorithm~\ref{alg:rad}).}
  \label{alg:radbis}
\end{algorithm2e}
 \medskip
 
As another example, the function $f$ can be used to select a node with minimum eccentricity in a set $W$ of nodes when $f(v,\ell)$ returns $\ell$ if $v\in W$ and $\infty$ otherwise. One can easily check that $f(v,.)$ is non-decreasing for all $v\in V$. Using this function in a variant of Algorithm~\ref{alg:radbis} then allows to find a node with minimum eccentricity in a given subset $W$ of nodes.
We use our minimum eccentricity selection as an optimization for diameter computation and as a core tool for computing all eccentricities in the next sections.

\section{Diameter computation and certification}
\label{sec:diam}

We now analyze a simple diameter algorithm. The main ingredient of the algorithm consists in maintaining upper bounds of all eccentricities and performing one-to-all distance queries from nodes with maximum upper bound. It thus follows the main line of previous practical algorithms~\cite{TK11,BCHKMT15}. However, we present the algorithm with a more general primal-dual approach which was not noticed before. Moreover, we introduce a new technique called \emph{delegate certificate}: after selecting a node $u$ with maximal upper bound, it consists in performing a one-to-all distance query from any tight upper vertex-certificate for $u$, that is a node $x$ such that $d(u,x)+e(x)=e(u)$ (see Section~\ref{sec:min-tight-upper-cert}). 
A possible choice for $x$ is $u$ itself in which case the algorithm becomes equivalent to the variant of \cite{TK11} where the selection procedure for BFS sources always selects a node with maximum upper-bound. However, we observe that choosing a node $x$ with minimal eccentricity offers much better performances in practice (see Section~\ref{sec:exp}). Our complexity analysis is independent of the choice of $x$, we thus present the algorithm in the most general manner.


The algorithm grows both a packing $K$ and an upper certificate $U$ until the upper bound $e^U(u)$ on the eccentricity 
of any node $u$ is at most the maximum eccentricity of nodes in $K$. As long as this condition is not satisfied, a node $u$ with maximal upper bound is selected and added to $K$. We then choose a tight upper vertex-certificate $x$ for $u$ and add it to $U$. Note that we now have $e^U(u)=e(u)\le \max_{v\in K}e(v)$ and $u$ cannot be selected again. This ensures that the termination condition is reached at some point when $U$ is a certificate that all nodes have eccentricity at most that of a maximum-eccentricity node in $K$, which must thus be equal to diameter.
See Algorithm~\ref{alg:diam} for a formal description.

\smallskip

We claim that the set $K$ is a packing for the collection $\D_{1/3}=\set{B(u,\frac{1}{3}(\diam(G)-e(u))) : u\in V}$  of open balls. As it has same size as the certificate $U$ returned by the algorithm in the end, this allows to state the following theorem.

\begin{theorem}\label{th:diam}
  Given a graph $G$, Algorithm~\ref{alg:diam} computes the diameter of $G$, a diametral node $b$ and a diameter certificate $U$ of size $\pi_{1/3}$ at most, with $O(\pi_{1/3})$ one-to-all distance queries in $O(m\pi_{1/3})$ time, where $\pi_{\alpha}$ is the maximum packing size for the collection of open balls $\D_\alpha=\set{B(u,\alpha(\diam(G)-e(u))) : u\in V}$  for $\alpha > 0$. It approximates minimum diameter certificate within a factor $\frac{\pi_{1/3}}{\pi_{[1]}}$ where $\pi_{[1]}$ is the maximum packing size for the collection $\D_{[1]}=\set{B[u,\diam(G)-e(u)] : u\in V}$.
\end{theorem}

\medskip
\begin{algorithm2e}[ht]
  \Input{A graph $G$ and a ranking $r$ of its node set $V$.}
  \Output{The diameter $\diam(G)$ of $G$, a diametral node $b$ and a diameter certificate $U$.}
  $U := \emptyset$\quad\Comment{Upper certificate (tentative covering with $\set{B[u,\diam(G)-e(u)] : u\in V}$).}
  Maintain $e^U(u) = \min_{x\in U} d(u,x) + e(x)$ (initially $\infty$) for all $v\in V$.\\
  $K := \emptyset$\quad\Comment{Packing for $\set{B(u,\frac{1}{3}(\diam(G)-e(u))) : u\in V}$.}
  \Do{$\max_{u\in K}e(u) < \max_{u\in V}e^U(u)$}{
    Select $u$ such that $e^U(u)$ is maximal.\\
    $D_u := \distfrom(G, u)$ \quad\Comment{Distances from $u$.}
    $e(u):=\max_{v\in V}D_u(v)$ \quad\Comment{Eccentricity of $u$.}
    $K:=K\cup\set{u}$\\
    \nl\label{line:delegate}%
    Select $x$ such that $d(u,x)+e(x) = e(u)$.\quad\Comment{Delegate certificate for $u$.}
    \nl\label{line:x-sweep}%
    $D_x := \distfrom(G,x)$ \quad\Comment{Distances from $x$.}
    $e(x):=\max_{v\in V}D_x(v)$ \quad\Comment{Eccentricity of $x$.}
    $U:=U\cup\set{x}$\\
    \lFor{$v\in V$}{$e^U(v):=\min(e^U(v), D_x(v)+e(x))$}\\
  }

  \medskip
  \Return{$e(b)$, $b$ and $U$ where $b\in K$ satisfies $e(b)=\max_{u\in K}e(u)$.}
  
  \caption{Computing diameter and a diameter certificate. The basic version of the algorithm consists in selecting $x:=u$ in Line~\ref{line:delegate}. (The redundant one-to-all distance query in Line~\ref{line:x-sweep} can then be omitted.)}
  \label{alg:diam}
\end{algorithm2e}

\begin{proof} 
  We already argued the termination and the correctness of the algorithm above.
  We thus show the packing property of $K$. Suppose for the sake of contradiction that $K$ is not a packing for $\D_{1/3}$. Consider the first iteration where a node $v$ is added to $K$ while some open ball $B(y,\frac{1}{3}(\diam(G)- e(y)))$ in $\D_{1/3}$ contains both $v$ and some node $u\in K$ added previously.
  Let $x$ be the tight upper vertex-certificate for $u$ that was added to $U$. By triangle inequality, we have $d(x,v)\le d(x,u) + d(u,y) + d(y,v)$. The choice of $x$ implies $d(x,u) =  e(u) -  e(x) \le d(u,y) +  e(y) -  e(x)$. Combining the two inequalities, we obtain $d(x,v)\le 2 d(u,y) + d(y,v) +  e(y) -  e(x)$. As $u$ and $v$ are in $B(y,\frac{1}{3}(\diam(G)- e(y)))$, we have $2 d(u,y) + d(y,v) < \diam(G) -  e(y)$ and finally get $d(x,v) < \diam(G) -  e(x)$ which implies $e^U(v) < \diam(G)$. However, it is required that $v$ has maximal upper bound $e^U(v)=\max_{w\in V}e^U(w)$ when it is selected for being added to $K$, in contradiction with $\max_{w\in V}e^U(w) \ge \max_{w\in V} e(w)=\diam(G)$. We conclude that $K$ must be a packing for $\D_{1/3}$. Both sizes of $K$ and $U$ are thus bounded by $\pi_{1/3}$. As any diameter certificate is a covering for $\D_{[1]}$ and has size $\pi_{[1]}$ at least, this guarantees  that the size of $U$ is within a factor $\frac{\pi_{1/3}}{\pi_{[1]}}$ at most from optimum.
\end{proof}

This analysis can be complemented when we start Algorithm~\ref{alg:diam} with $K:=\set{c}$ and $U:=\set{c}$ initially where $c$ is a center of the graph computed with Algorithm~\ref{alg:rad}.
We reference this combination as Algorithm~1+4 in the sequel.
A similar proof then allows to show that $K$ is a packing for $\D_{1/3}^c(G)=\set{B(u,\beta_u(\diam(G)- e(u))) : u\in V}$ where $\beta_u=1/3$ for $u\not= c$ and $\beta_c=1$. We obtain the following corollary from Theorems~\ref{th:rad} and~\ref{th:diam}.

\begin{corollary}\label{cor:diam}
  Given a graph $G$ and a ranking $r$ of its node set $V$, Algorithm~1+4 computes the diameter of $G$, a diametral node $b$ and a diameter certificate $U$ of size $\pi_{1/3}^c$ at most with $\card{U} + 2\card{\antipode_r(V)}+1$ one-to-all distance queries at most in $O(m(\pi_{1/3}^c+\card{\antipode_r(V)}))$ time where $c$ is a center of $G$ returned by Algorithm~\ref{alg:rad} and $\pi_{1/3}^c=\pi(\D_{1/3}^c)$ is the maximum packing size for the collection $\D_{1/3}^c=\set{B(u,\beta_u(\diam(G)- e(u))) : u\in V}$ of open balls with radii factors $\beta_u=\frac{1}{3}$ for $u\not= c$ and $\beta_c=1$ (for $u=c$).
\end{corollary}

\begin{proof} 
  In addition to the proof of Theorem~\ref{th:diam}, we just have to consider the case when a node $v$ would be added to $K$ while having $v\in B(c,\diam(G)- e(c))$. As $c\in U$, we then have $e^U(v)\le d(c,v)+e(c)<\diam(G)$. This would raise a contradiction as the choice of $v$ relies on $e^U(v)=\max_{w\in V}e^U(w)\ge \max_{w\in V}e(w)\ge\diam(G)$.
\end{proof}

This explains efficiency of practical algorithms as we observe that coverings of small size often exist for $\D_{1/3}^c(G)$ in practical graphs (see Section~\ref{sec:exp}).
As mentioned before, a further optimization consists in selecting a tight upper vertex-certificate $x$ for $u$ with minimal eccentricity. Using a function $f$ such that $f(v,\ell)$ returns $\ell$ when $D_u(v)+\ell \le e(u)$ and returns $\infty$ otherwise, it can be obtained through our minimum eccentricity selection procedure by replacing Line~\ref{line:delegate} with $x := \argminecc(G,r,L,e_L,f)$. The algorithm is referenced as Algorithm~1+4' in the sequel. This optimization through the delegate certificate technique provides performances similar to previous practical algorithms (see Section~\ref{sec:exp}) while providing the complexity guarantee of Corollary~\ref{cor:diam}.

\section{All eccentricities}
\label{sec:all-ecc}

We now present a novel algorithm for all eccentricities. It relies on minimum eccentricity selection and the characterization of the minimum tight upper certificates presented in Section~\ref{sec:min-tight-upper-cert}.
We propose to compute all eccentricities of a graph as follows (see Algorithm~\ref{alg:all-ecc} for a formal description). We maintain both a lower certificate $L$ and an upper certificate $U$. As long as some node has untight upper bound, we select a node $u$ with untight upper bound and minimal eccentricity using our minimum eccentricity selection procedure which then additionally ensures $e_L(u)=e(u)$. (We use for that purpose a function returning $\infty$ when the eccentricity value equals the upper bound.) We claim that $u$ is in $U^{\preceq}$ (see Lemma~\ref{lem:cmax} below). We thus add $u$ to the upper certificate $U$ and update upper bounds accordingly. When our minimum eccentricity selection procedure detects that all nodes have tight upper bounds, lower bounds must be tight also. The algorithm then terminates with the following guarantees.

\medskip
\begin{algorithm2e}[t]
  \Input{A graph $G$ and a ranking $r$ of $V$.}
  \Output{All eccentricities, a tight lower certificate $L$ of $G$ and a tight upper certificate $U$ of $G$.}
  $L := \emptyset$\quad\Comment{Lower certificate (tentative hitting set for $\set{\ovB(u,e(u)) : u\in V}$)}
  Maintain $e_L(v)=\max_{x\in L}d(v,x)$ (initially 0) for all $v\in V$.\\
  $U := \emptyset$\quad\Comment{Upper certificate (maximal nodes for $\preceq$)}
  Maintain $e^U(u) = \min_{x\in U} d(u,x) + e(x)$ (initially $\infty$) for all $v\in V$.\\

  \Function{$\eccuntight(v,\ell)$}{
    \lIf{$\ell<e^U(v)$}{\return{$\ell$}} \lElse{\return{$\infty$}}
  }
  \While{$\minecc(L,e_L,\eccuntight) < \infty$}{
    $u := \argminecc(G,r,L,e_L,\eccuntight)$\\
    $D_u := \distfrom(G, u)$ \quad\Comment{Distances from $u$.}
    $e(u):=\max_{v\in V}D_u(v)$ \quad\Comment{Eccentricity of $u$.}
    $U:=U\cup\set{u}$\\
    \lFor{$v\in V$}{$e^U(v):=\min(e^U(v), D_u(v)+e(u))$}\\
  }
  \Return{$e^U,L,U$}
  
  \caption{Computing all eccentricities and tight lower/upper certificates.}
  \label{alg:all-ecc}
\end{algorithm2e}

\begin{theorem}
  \label{th:all-ecc}
  Given a graph $G$ and a ranking $r$ of its node set $V$, Algorithm~\ref{alg:all-ecc} computes all eccentricities, a tight lower certificate $L\subseteq \antipode_r(V)$ and the optimal tight upper certificate $U^{\preceq}$ with $\card{U^{\preceq}} + 2\card{L}$ one-to-all distance queries in $O(m(\card{U^{\preceq}} + 2\card{L}))$ time.
\end{theorem}

In practical graphs, we observe that the size of $U^{\preceq}$ is much larger than that of the computed lower certificate $L$ and the algorithm roughly costs $\card{U^{\preceq}}$ BFS traversals (see Section~\ref{sec:exp}). This seems almost optimal for an algorithm performing a BFS traversal from each node of $U^{\preceq}$. In particular, we expect a speed-up factor of 2 at least compared to~\cite{TK13} which suggest to alternate its selection of BFS sources between nodes of minimum lower-bound and nodes of maximum upper-bound (in addition there is no guarantee that nodes with minimum lower-bound are in $U^{\preceq}$).




The correctness of Algorithm~\ref{alg:all-ecc} mainly rely on the following lemma.

\begin{lemma}
  \label{lem:cmax}
  Consider an upper certificate $U$ and the set $S_U$ of nodes that do not have a tight upper vertex-certificate in $U$. Any node $v\in S_U$ with minimum eccentricity (having $e(v)=\min_{u\in S_U} e(u)$) is its unique tight upper vertex-certificate (i.e., $v\in U^{\preceq}$).
\end{lemma}

\begin{proof}
  For the sake of contradiction, suppose that $v\in S_U$ has a tight upper vertex-certificate $x\not=v$. As $e(v)=d(v,x)+ e(x)$, we have $e(x)< e(v)$ and $x$ cannot be in $S_U$ by the minimality of $e(v)$. It thus has a tight upper vertex-certificate $y\in U$. But the transitivity of being a tight upper vertex-certificate (Proposition~\ref{prop:upperopt}) implies that $y$ is also a tight upper certificate for $v$ which is in contradiction with $v\in S_U$.
\end{proof}

\begin{proof}[of Theorem~\ref{th:all-ecc}]
  We first prove $U\subseteq U^{\preceq}$.
  As in Lemma~\ref{lem:cmax}, let $S_U$ denote the set of nodes that do not have a tight upper vertex-certificate in $U$ ($S_U=\set{v\in V \mid e(v)<e^U(v)}$).
  Now consider the node $u$ selected at some iteration of the while loop. We prove $u\in U^{\preceq}\setminus U$. The correctness of our minimum eccentricity selection (Proposition~\ref{prop:selection}) implies that $u$ has minimum eccentricity in $S_U$ and is thus in $U^{\preceq}$ by Lemma~\ref{lem:cmax}. (Note that $\eccuntight(v,.)$ is non-decreasing for all $v\in V$ as $\eccuntight(v,\ell)=\ell$ for $\ell < e(v)$ and $\eccuntight(v,\ell)=\infty$ for $\ell\ge e(v)$.) Additionally, $u\in S_U$ implies that $u$ has untight upper bound and is not in $U$ until we add it at that iteration.
  
  The termination of the algorithm is guaranteed by the fact that $U$ grows at each iteration. The algorithm ends when minimum eccentricity selection returns a node $u$ such that $\eccuntight(u,e(u))=\infty$. Proposition~\ref{prop:selection} then ensures $\eccuntight(v,e_L(v))=\infty$ for all $v$. That is $e_L(v)=e^U(v)$ for all $v$ and both bounds must equal $e(v)$. This implies that $L$ (resp. $U$) is then a tight lower (resp. upper) certificate of $G$. Moreover, $U\subseteq U^{\preceq}$ then implies $U=U^{\preceq}$ by Proposition~\ref{prop:upperopt}.
\end{proof}

\section{Experiments}
\label{sec:exp}

We test social networks (Epinions, Hollywood, Slashdot, Twitter, dblp), computer networks (Gnutella, Skitter), web graphs (BerkStan, IndoChina, NotreDame), road networks (CAL-t, CAL-d, CAL-u, FLA-t, europe-t), a 3D triangular mesh (buddha), and grid like graphs from VLSI applications (alue7065) and from computer games (FrozenSea). The data is available from \url{snap.stanford.edu}, \url{webgraph.di.unimi.it}, \url{www.dis.uniroma1.it/challenge9}, \url{graphics.stanford.edu}, \url{steinlib.zib.de} and \url{movingai.com}. We also test synthetic inputs: bowtie500 is the graph $BT_{500,500}$ represented in Figure~\ref{fig:bowtie}, grid500-10 is a $501\times 501$ square grid with random deletion of 10\% of the edges, grid1500-wd is a weighted directed graph obtained from a $1501\times 1501$ square grid where each edge is oriented randomly (with probability $1/2$ for each direction) and assigned a random weight uniformly in $\set{0,1,\ldots,9}$, pwlaw2.5 is a random graph generated according to the configuration model with a degree sequence following a power law with exponent 2.5, udg10 is a random unit disk graph where field size is parameterized to obtain average degree 10 roughly.
Each graph is restricted to its largest (strongly) connected component. Our code is available at \url{https://github.com/lviennot/weighted-diameter}.


\begin{table}[t]
\centering
\resizebox{\linewidth}{!}{%
  \setlength{\tabcolsep}{4pt}
\begin{tabular}{lrrrrrrrrrrrrrrr}
type & name & $n$ & $m/n$ & d & w & $\frac{\diam}{\rad}$ & $D$ & $\pi^c_{0.8}$ & $\pi^c_{1/3}$ & $n_c/n$ & $R$ & $A_{ID}$ & $F$ & $L_{all}$ & $U_{all}$\\[1mm]
\hline
comm & Gnutella & 14149 & 3.60 & $\bullet$ & $\circ$ & 1.58 & 19 & 47 & 1912 & 0.27 & 4 & 10 & 23 & 7 & 2457\\
comm & skitter & 1694616 & 13.09 & $\circ$ & $\circ$ & 1.94 & 3 & 5 & 7 & 0.99 & 3 & 6 & 6 & 3 & 33535\\
game & FrozenSea & 753343 & 7.70 & $\circ$ & $\bullet$ & 1.80 & 15 & 35 & 191 & 0.91 & 7 & 384 & 388 & 381 & 54568\\
geom & buddha & 543652 & 6.00 & $\circ$ & $\bullet$ & 1.87 & 27 & 63 & 385 & 0.95 & 14 & 897 & 897 & 897 & 71863\\
road & CAL-d & 1890815 & 2.45 & $\circ$ & $\bullet$ & 1.89 & 3 & 17 & 90 & 0.99 & 3 & 11 & 11 & 11 & 725\\
road & CAL-t & 1890815 & 2.45 & $\circ$ & $\bullet$ & 1.83 & 7 & 17 & 105 & 0.98 & 5 & 13 & 13 & 13 & 2810\\
road & CAL-u & 1890815 & 2.45 & $\circ$ & $\circ$ & 1.99 & 2 & 4 & 6 & 0.99 & 3 & 7 & 11 & 7 & 1075\\
road & FLA-t & 1070376 & 2.51 & $\circ$ & $\bullet$ & 1.99 & 2 & 4 & 4 & 0.99 & 2 & 2 & 2 & 2 & 174\\
road & europe-t & 18010173 & 2.34 & $\bullet$ & $\bullet$ & 1.99 & 2 & 4 & 5 & 1 & 2 & - & - & 2 & 711\\
soc & Epinions & 32223 & 13.76 & $\bullet$ & $\circ$ & 2 & 2 & 4 & 7 & 0.99 & 2 & 7 & 20 & 3 & 294\\
soc & Hollywood & 1069126 & 106.33 & $\circ$ & $\circ$ & 1.71 & 29 & 603 & 4183 & 0.99 & 3 & 34 & 335 & 10 & 120626\\
soc & Slashdot & 71307 & 12.80 & $\bullet$ & $\circ$ & 1.86 & 4 & 40 & 65 & 0.99 & 4 & 12 & 135 & 9 & 8204\\
soc & Twitter & 68413 & 24.63 & $\bullet$ & $\circ$ & 2.50 & 2 & 6 & 8 & 0.99 & 4 & 31 & 4753 & 9 & 3702\\
soc & dblp & 226413 & 6.33 & $\circ$ & $\circ$ & 2 & 1 & 11 & 20 & 1 & 5 & 6 & 43 & 6 & 13776\\
synth & bowtie500 & 505002 & 2.00 & $\circ$ & $\circ$ & 1.99 & 3 & 2001 & 4001 & 0.99 & 5 & 5 & 1507 & 5 & 9\\
synth & grid1500-wd & 296680 & 1.66 & $\bullet$ & $\bullet$ & 2.38 & 3 & 12 & 69 & 0.04 & 4 & 5 & 6 & 5 & 82\\
synth & grid500-10 & 250976 & 3.59 & $\circ$ & $\circ$ & 2 & 1 & 5 & 5 & 1 & 3 & 4 & 4 & 4 & 104\\
synth & pwlaw2.5 & 1000000 & 3.85 & $\circ$ & $\circ$ & 1.90 & 4 & 24 & 24 & 0.99 & 2 & 19 & 47 & 6 & 18458\\
synth & udg10 & 999888 & 9.99 & $\circ$ & $\circ$ & 1.99 & 5 & 17 & 86 & 0.99 & 4 & 5 & 10 & 4 & 1273\\
vlsi & alue7065 & 34046 & 3.22 & $\circ$ & $\circ$ & 2 & 1 & 5 & 5 & 1 & 3 & 4 & 4 & 4 & 54\\
web & BerkStan & 334857 & 13.51 & $\bullet$ & $\circ$ & 2.73 & 2 & 7 & 28 & 2e-03 & 3 & 16 & 17 & 16 & 20\\
web & Indochina & 3806327 & 25.96 & $\bullet$ & $\circ$ & 6.91 & 2 & 4 & 6 & 0.99 & 3 & - & - & 5 & 24571\\
web & NotreDame & 53968 & 5.65 & $\bullet$ & $\circ$ & 2.11 & 2 & 5 & 22 & 0.01 & 2 & 2 & 2 & 2 & 45\\
\end{tabular}

  }
  \caption{Diameter and radius certificate sizes ($D,R$) and  sizes of all eccentricity certificates $(L_{all},U_{all})$ for various graphs, and related parameters. The $U_{all}$ column in this table differs from previous versions of this paper, in which it was reported incorrectly.}
  \label{tab:all}
\end{table}

Table~\ref{tab:all} summarizes our main practical observations. For each instance $G$, we show its type, the number $n$ of nodes in the largest (strongly) connected component, the average out-degree $m/n$, whether it is directed (d) and weighted (w), and the diameter to radius ratio $\frac{\diam(G)}{\rad(G)}$. We then show the size $D$ of the diameter certificate computed by Algorithm~1+4', bounds on maximum packing sizes $\pi^c_{0.8}$ and $\pi^c_{1/3}$ (defined in Section~\ref{sec:diam}), the proportion $n_c/n$ of nodes in the (in-)ball of radius $\diam(G)-\rad(G)$ centered at a center $c$, the size $R$ of the radius certificate computed by Algorithm~\ref{alg:rad}, the number $A_{ID}$ of antipodes for ID ranking and the number $F$ of furthest nodes. The two latter numbers were obtained by performing a traversal per node of the graph (in quadratic time). A dash indicates a value that could not be obtained in less than few days of computation. 

\medskip

The first observation is that diameter and radius certificates are extremely small for all instances (less than 30 nodes for all of them). 
Several observations allow to explain this phenomenon. First, all graphs have high diameter to radius ratio (over 1.5 for all of them). Note that this ratio is at most 2 for an undirected graph (it is unbounded in general directed graphs). Undirected graphs with ratio 2 have a one node diameter certificate: a center. This concerns two practical graphs while several ones have ratio very close to 2. Coherently, the concentration of nodes around the center is also high with respect to the diameter minus radius difference. This is measured by the ratio $n_c/n$ where $c$ is a center computed by our radius algorithm ($e(c)=\rad(G)$) and $n_c$ denotes the number of nodes in $B[c,\diam(G)-\rad(G)]$ (in directed graphs we count the number of nodes $u$ such that $d(u,c)\le \diam(G)-\rad(G)$). It counts the proportion of nodes $u$ such that $e^{\set{c}}(u)\le \diam(G)$ which appears to be very close to 1 for most of the graphs. Notable exceptions are Gnutella, BerkStan and NotreDame. This may be explained by the low (compared to others) diameter to radius ratio (1.58) of the first one and probably to the highly asymmetric nature of the two others. Seeing diameter certification as a covering problem with balls $B[x,\diam(G)-e(x)]$, there are thus few nodes that are not covered by a center $c$. Additionally, other nodes can be covered using few balls with reduced radii: the columns $\pi^c_{0.8}$ and $\pi^c_{1/3}$ indicate the size of coverings we could find using balls with radii reduced by a factor .8 and $1/3$ respectively. These numbers upper bound maximum packing sizes $\pi^c_{0.8}$ and $\pi^c_{1/3}$ of the associated collections of balls with reduced radii. Our theoretical upper-bound of $\pi^c_{1/3}$ thus explains fast diameter computation for most of the graphs. A notable exception is $BT_{500,500}$ alias bowtie500 which was tailored for making former diameter algorithms slow (including Algorithm~1+4) and thus have large $\pi^c_{1/3}$ value. (Note that Algorithm~1+4' performs roughly $2D+2R$ distances queries and at most $2D+2A_{ID}$ distance queries, that is 16 for bowtie500.) Other exceptions are Hollywood and Gnutella for which the diameter to radius ratios are not so high either (compared to other graphs). However the parameter is still much smaller than the number of nodes.

Concerning radius computation, we observe that most graphs have very few antipodes as indicated by the $A_{ID}$ column although the number $F$ of furthest nodes can be significantly larger as observed among several social networks. A notable exception is the buddha graph which is a triangulated 3D surface and thus has more or less a sphere like topology (the arms form handles) that may explain why all furthest nodes are antipodes. However the number of antipodes remains much smaller than the number of nodes.
FrozenSea also has a relatively large number of furthest nodes that are almost all antipodes. This might come from the design of the graph as a map where players of a video game evolve and should find dead ends.


%% file: body/graphclasses.tex
\part{Graph classes with specific certificates}
\label{part:graphclasses}

\section{Power law random graphs}
\label{sec:power-law}

A $\beta$-power-law random graph can be defined according to the configuration model~\cite{Bollobas1980} starting from a degree distribution following a power law with exponent $\beta$ (self-loops and multiple edges are allowed). We now translate some results of \cite{BCT17} in terms of certificates.

\begin{proposition}[\cite{BCT17}]
    Given $\eps>0$ and $\beta>1$, any $\beta$-power-law random graph $G$ has asymptotically almost surely:
    \begin{itemize}
        \item  a tight lower certificate (and thus a radius certificate) of size $n^{O(\eps)}$ when $\beta > 2$,
        \item  a tight lower certificate (and thus a radius certificate) of size $O\left(n^{1-\frac{2-\beta}{\beta-1}\left(\lfloor\frac{\beta-1}{2-\beta} - \frac{3}{2}\rfloor - \frac{1}{2}\right)}\right)$ when $1<\beta<2$;
        \item a diameter certificate of size
        $n^{O(\eps)}$ when $1<\beta<3$,
        \item a diameter certificate of size
        $n^{\frac{1}{1+\Omega(\frac{\beta-1}{\beta-3})}+O(\eps)}$ when $\beta > 3$.
    \end{itemize} 
\end{proposition}

We now detail the analyses of~\cite{BCT17} from which these bounds are derived.
The following variant of SumSweep heuristic~\cite{BCHKMT15} is proposed: given a graph $G=(V,E)$, select independently and uniformly at random $s_1,\ldots,s_k\in V$, then iteratively choose $t_1,\ldots, t_k$ such that each $t_j$ is selected in $V\setminus\set{t_1,\ldots, t_{j-1}}$ so as to maximize $\sum_{i\le j} d(s_i,t_j)$. The algorithm performs a BFS from each of these nodes and maintains corresponding lower bounds for all vertices. The analysis of this variant (see Section~9 of~\cite{BCT17}) shows that for any $\eps>0$ and for sufficiently large graph $G$, all lower bounds are tight for $k=n^{3\eps}+n^{O(\eps)}$ when $\beta>2$, and for $k=n^{3\eps}+n^{f(\beta)}$ for $1<\beta<2$ where $f(\beta)=1-\frac{2-\beta}{\beta-1}\left(\lfloor\frac{\beta-1}{2-\beta} - \frac{3}{2}\rfloor - \frac{1}{2}\right)$. We note that we have $f(\beta)<1$ for $\beta>1.72$.
In our terminology, $L=\set{s_1,\ldots,s_k,t_1,\ldots, t_{j-1}}$ is a tight lower certificate and the analysis shows that for any $\eps>0$ a $\beta$-power-law random graph has asymptotically almost surely (a.a.s. for short) a tight lower certificate of size $n^{O(\eps)}$ when $\beta>2$. When $1.72<\beta<2$, it has a.a.s. a tight lower certificate of sublinear size $O(n^{f(\beta)})$.

The analysis of the Exact SumSweep algorithm~\cite{BCHKMT15} then shows that after finding a center $c$, the number of nodes with upper bound greater than the diameter is proven to be $n^{O(\eps)}$ a.a.s. for $1<\beta<3$ by bounding the number of nodes that are not in $B[c,diam(G)-e(c)]$ (see Section~11.1 of \cite{BCT17}). In our terminology, $U=\set{c}\cup (V\setminus B[c,diam(G)-e(c)])$ is a diameter certificate of size $n^{O(\eps)}$.
For $\beta>3$, the analysis is more involved. It proves the existence of a diameter certificate of size $n^{g(\beta)+O(\eps)}$ with $g(\beta)=\frac{1}{1+\Omega(\frac{\beta-1}{\beta-3})}$ (see Section~11.2 of \cite{BCT17}). Note that $n^{g(\beta)}$ is always sublinear.

\medskip

Concerning the lower certificate bound, the tools provided in~\cite{BCT17} allow to more generally bound the overall number of furthest nodes.
We show that the set $F=\set{t\in V: \exists s\in V, d(s,t)=e(s)}$ of furthest nodes has size $n^{O(\eps)}$. 
We rely on the following results of~\cite{BCT17} where $\tau_s(p)=\min\set{\ell\in \mathbb{N} : \card{B[s,\ell]\setminus B[s,\ell-1]} > p}$ denotes the minimum distance $\ell$ such that $s$ has more than $p$ nodes at distance exactly $\ell$ and $T(1\rightarrow p)$ is the average of $\tau_s(p)$ over vertices $s$ of degree one: more precisely, it satisfies $T(1\rightarrow p)=2$ for $1<\beta<2$, $T(1\rightarrow p)=(1+o(1))\log_{1/(\beta-2)}\log p$ for $1<\beta<2$, and $T(1\rightarrow p)=O(\log p)$ for $3<\beta$.

\begin{lemma}[\cite{BCT17}] \label{lem:bct17}
    Given $\eps>0$ and $2<\beta<3$, there exists a positive constant $c<1$ such that any $\beta$-power-law random graph $G$ satisfies a.a.s:
    \begin{itemize}
        \item[(i)] for any node $s$ and $t$ at furthest distance from $s$, $\tau_t(n^{1/2}) \geq \lceil (1-O(\eps))(T(1\rightarrow n^{1/2}) + \frac{\log n}{-\log c} - \frac{3}{2})\rceil$ (Lemma~7.1 in \cite{BCT17} with $x=1/2$);
        \item[(ii)] the number of vertices satisfying $\tau_s(n^{1/2}) \ge (1+\eps)(T(1\rightarrow n^{1/2}) + \alpha)$ is $O(nc^{\alpha-1/2})$   (Property~2.1 in \cite{BCT17} with $x=1/2$).
    \end{itemize}
\end{lemma}

We can now state the following.

\begin{proposition}
    Given $\eps>0$ and $\beta>1$, any $\beta$-power-law random graph $G$ has a.a.s. $n^{O(\eps)}$ furthest nodes.
\end{proposition}

Note that the set of furthest nodes is obviously a tight lower certificate.

\begin{proof}
    Consider a furthest node $t\in F$.
    According to Lemma~\ref{lem:bct17}(i),    we have $\tau_t(n^{1/2}) \geq (1-a\eps)(T(1\rightarrow n^{1/2}) + \frac{\log n}{-\log c} - \frac{3}{2})$ for some constant $a>0$. Set $\alpha$ to obtain $ (1-a\eps)(T(1\rightarrow n^{1/2}) + \frac{\log n}{-\log c} - \frac{3}{2}) = (1+\eps)(T(1\rightarrow n^{1/2}) + \alpha)$, that is    $\alpha = \frac{\log n}{-\log c} - \frac{3}{2} - O(\eps) (T(1\rightarrow n^{1/2}) + \frac{\log n}{-\log c} - \frac{3}{2})$ (we assume $\eps<1/2$ without loss of generality). In all three regimes for $\beta$, we have $\alpha-\frac{1}{2} = \frac{\log n}{-\log c} - 2 - O(\eps\log n)$.
    Now Lemma~\ref{lem:bct17}(ii) bounds the number of such $t$, yielding: $\card{F} = O(nc^{\alpha-1/2})=c^{-2-O(\eps\log n)} =n^{O(\eps)}$.
\end{proof}

\medskip

Note that the results of \cite{BCT17} also apply to Rank-1 Inhomogeneous Random Graph models (see Appendix~A in \cite{BCT17}).

\section{Graphs with low doubling dimension} 
\label{sec:doubling}


A graph is $\gamma$-doubling if every ball of positive radius is included in the
union of at most $\gamma$ balls with half radius.

\subsection{Exact diameter computation with few antipodes}

In the case of $\gamma$-doubling graphs, the following theorem shows that we can obtain a diameter certificate with almost linear size compared to the maximum packing size $\pi(\D_\alpha)$ for the collection $\D_\alpha$ (see Section~\ref{sec:diam}). This complements Theorem~\ref{th:diam} in the range $\frac{1}{3}\le\alpha<1$.

\begin{theorem}\label{th:doubling}
Given a $\gamma$-doubling graph $G$ and $\alpha < 1$, the diameter $\diam(G)$, a diametral node $p$ and a diameter certificate $U$ satisfying $\card{U}\le \pi_\alpha \gamma^{O(1)+\log\frac{1}{1-\alpha}}\log \diam(G)$ can be computed with $2\card{\antipode(V)}+\card{U}$ one-to-all distance queries, where $\pi_\alpha$ is the maximum size of a packing for the collection $\D_\alpha=\set{B(u,\alpha(\diam(G)- e(u))) : u\in V}$.
\end{theorem}

Note that this implies that minimum diameter certificate can be approximated within a factor $\gamma^{O(1)+\log(\diam(G)-\rad(G))}\log \diam(G)\frac{\pi_1}{\pi_{[1]}}$ when $G$ is $\gamma$-doubling as $\D_\alpha=\D_1$ for $\alpha > 1 - \frac{1}{r+1}$ where $r=\diam(G)-rad(G)$ is the maximum radius of a ball in $\D_1$ (recall that $\pi_{[1]}$ lower bounds the size of a minimum diameter certificate).

\smallskip

The above theorem is a consequence of Algorithm~\ref{alg:doubling} which follows a primal-dual approach by constructing both a packing $K$ for $\D_\alpha$ together with a covering $U$ with $\D_{[1]}$ such that $\card{U}\le \card{K} \gamma^{O(1)+\log\frac{1}{1-\alpha}}\log \diam(G)$. Given a node $u$, let $S_u=\set{v\in V \mid \exists B\in \D_\alpha \mbox{ s.t. } u,v\in B}$ denote the set of nodes $v$ that cannot be in a packing for $\D_\alpha$ containing $u$. The idea is to iteratively add a node $u$ to $K$ with highest eccentricity according to $e^U$ and then to add sufficiently many nodes to $U$ so that any node $v$ in $S_u$ gets an eccentricity upper bound $e^U(v)$ equal to $\diam(G)$ or less. This will guarantee that no such node is added later to $K$ and that $K$ is a packing for $\D_\alpha$.

\begin{algorithm2e}[t]
  \Input{A graph $G$ and a parameter $\alpha$ with $0<\alpha<1$.}
  \Output{The diameter $\diam(G)$ of $G$ and a diameter certificate $p,U$.}
  $K := \emptyset$\quad\Comment{Packing for $\set{B(u,\alpha(\diam(G)- e(u))) : u\in V}$.}
  $U := \emptyset$\quad\Comment{Upper certificate.}
  Maintain $e^U(u) = \min_{x\in U} d(u,x) + e(x)$ (initially $\infty$) for all $v\in V$.\\
  $L := \emptyset$\quad\Comment{Lower certificate.}
  Maintain $e_L(v)=\max_{x\in L}d(v,x)$ (initially 0) for all $v\in V$.\\
  \While{$\max_{p\in K} e(p) < \max_{u\in V}e^U(u)$}{
    Select $u$ such that $e^U(u)$ is maximal.\\
    $D_u:=\distfrom(u)$\\
    $e(u):=\max_{w\in V}D_u(w)$ \quad\Comment{Eccentricity of $u$.}
    $K:=K\cup\set{u}$\\
    $U:=U\cup\set{u}$\\
    \lFor{$w\in V$}{$e^U(w):=\min(e^U(w), D_u(w)+e(u))$}
    
    \Function{$\eccslack(v,\ell)$}{
      \lIf{$e^U(v)-\ell > \frac{1-\alpha}{2\alpha} D_u(v)$}{
        \return{$-D_u(v)$}
      }
      \lElse{\return{$\infty$}}
    }
    \While{$\minecc(G,r,L,e_L,\eccslack) < \infty$}{
      $v := \argminecc(G,r,L,e_L,\eccslack)$\\
      $D_v := \distfrom(v)$\\
      $e(v):=\max_{w\in V}D_v(w)$ \quad\Comment{Eccentricity of $v$.}
      $U:=U\cup\set{v}$\\
      \lFor{$w\in V$}{$e^U(w):=\min(e^U(w), D_v(w)+e(v))$}
    }
  }

  $p := \argmax_{p\in K} e(p)$\\
  \Return{$e(p)$ and $p,U$.}
    \medskip
  \caption{Computing diameter and a diameter certificate assuming doubling property.}
  \label{alg:doubling}
\end{algorithm2e}

 \medskip
 
Our selection rule for adding nodes to $U$ is based on comparing their eccentricity to their distance to $u$. The rough idea is to add a node $v$ as certificate in $U$ when $e^U(v) > \diam(G)$ and $\diam(G)- e(v)=\Omega(d(u,v))$. Note that any node $w\in B[v,\diam(G)- e(v)]$ then satisfies $e^U(w)\le \diam(G)$ and the doubling property will allow us to bound the number of nodes added to $U$.
As the eccentricity of $v$ is not known precisely until we perform a one-to-all distance query from $v$, we use our minimum eccentricity selection technique using a lower certificate $L$. As $\diam(G)$ is not known either, we select $v$ such that $e^U(v)-e(v)=\Omega(d(u,v))$. This ensures that a ball of radius $\Omega(d(u,v))$ will then be covered. We thus prefer $v$ such that $d(u,v)$ is additionally maximal.

\begin{proof}[of Theorem~\ref{th:doubling}]
The main arguments of the proof are the following.

The function $\ell\mapsto\eccslack(v,\ell)$ is non-decreasing as it returns $-d(u,v)$ for $\ell \le e^U(v) - \frac{1-\alpha}{2\alpha} d(u,v)$ and $\infty$ otherwise. Note that the minimum eccentricity selection for $\eccslack$ thus returns a node $v$ such that $e^U(v)-e(v)\ge \frac{1-\alpha}{2\alpha} d(u,v)$ and $d(u,v)$ is maximal.

We first prove that  the inner loop performs at most $\gamma^{O(1)+\log\frac{1}{1-\alpha}}\log \diam(G)$ iterations. This bounds the number of nodes added to $U$ when one node is added to $K$ and will thus ensure $\card{U}\le \card{K} \gamma^{O(1)+\log\frac{1}{1-\alpha}}\log \diam(G)$.
Consider an iteration of the main loop where $u$ is added to $K$. Just after adding $v$ to $U$ in the inner loop, consider $w$ s.t. $d(v,w)\le \frac{\rho}{2+\rho} d(u,v)$ where $\rho=\frac{1-\alpha}{2\alpha}$. We then have $e^U(w)\le d(v,w)+ e(v)$ as $v\in U$ and $e(w)\ge  e(v)-d(v,w)$ by triangle inequality. This gives $e^U(w) -  e(w)\le \frac{2\rho}{2+\rho} d(u,v)$.
As $d(u,w)\ge d(u,v)-d(v,w)\ge \frac{2}{2+\rho} d(u,v)$, we get $e^U(w) -  e(w)\le \rho d(u,w)$ and nodes in $B[v,\frac{\rho}{2+\rho} d(u,v)]$ do not satisfy the condition of the inner loop.
The doubling property implies that the number of iterations where we select $v$ such that $d(u,v)>\frac{ e(u)}{2}$ is at most $\gamma^{O(1)+\log\frac{1}{1-\alpha}}$.
Then we may select $v$ such that $d(u,v) > \frac{ e(u)}{4}$ during the same number of iterations at most, and so on until we eventually select $u$ itself (the only node $v$ such $d(u,v)=0$). The overall number of iterations of the inner loop is thus bounded by $\gamma^{O(1)+\log\frac{1}{1-\alpha}}\log  e(u)$.

For the sake of contradiction, suppose that $K$ is not a packing and consider $u,u'\in K$ and $x\in V$ such that both $u$ and $u'$ are in $B(x,\alpha(\diam(G)- e(x)))$. Assume without loss of generality that $u$ was added to $K$ before $u'$. After the inner loop for $u$, we have $e^U(x)- e(x) \le \frac{1-\alpha}{2\alpha} d(u,x)$. Therefore there exists $y\in U$ such that $d(x,y)+ e(y) \le  e(x) + \frac{1-\alpha}{2\alpha} d(u,x)$. We thus have $e^U(u')\le d(u',y)+ e(y) \le d(u',x)+ d(x,y)+ e(y)\le d(u',x) +  e(x) + \frac{1-\alpha}{2\alpha} d(u,x)$. As $u,u'\in B(x,\alpha(\diam(G)- e(x)))$, we get $e^U(u') < \frac{1+\alpha}{2} \diam(G) + \frac{1-\alpha}{2}  e(x)$. As $e(x)\le \diam(G)$, we obtain $e^U(u') < \diam(G)$. This is a contradiction since the choice of $u'$ implies $e^U(u')=\max_{v\in V} e^U(v)\ge \diam(G)$.
\end{proof}

\subsection{Approximating radius and diameter}

Interestingly, the following lemma links the gap between the bound provided by a lower/upper certificate for a node $u$ and the distance from the certificate to a tight lower/upper certificate for $u$. Recall that a tight upper certificate for $u$ is a node $x$ such that $e(u)=d(u,x)+e(x)$. We similarly define a tight lower certificate for $u$ as a node $x$ such that $e(u)=d(u,x)$ (equivalently, $x$ is a furthest node from $u$).

\begin{lemma}
  \label{lem:error}
  Given a lower certificate $L$ (resp. an upper certificate $U$) and a node $u$, we have $e(u)-e_L(u)\le d(x,L)$ (resp. $e^U(u) - e(u) \le 2 d(x, U)$) for any tight lower (resp. upper) certificate $x$ for $u$.
\end{lemma}

\begin{proof}
  Consider a tight lower certificate $x$ for a node $u$ ($e(u)=d(u,x)$). Let $y\in L$ be a closest node to $x$ in $L$ ($d(x,y)=d(x,L)$). By triangle inequality, we have $e(u)=d(u,x)\le d(u,y)+d(y,x)\le e_L(u)+d(x,L)$.

  Similarly, consider a tight upper certificate $x$ for a node $u$ ($e(u)=d(u,x)+e(x)$). Let $y\in U$ be a closest node to $x$ in $U$ ($d(x,y)=d(x,U)$). By triangle inequality, we have $e(u)=d(u,x)+e(x)\ge d(u,y)-d(x,y)+e(y)-d(x,y)\ge e^U(u)-2d(x,U)$.
\end{proof}

Now consider the choice of a node $u$ with minimal eccentricity lower bound in Algorithm~\ref{alg:rad}. This choice implies $e_L(u)\le \rad(G)$ and Lemma~\ref{lem:error} then implies $e(u) \le \rad(G) + d(a,L)$ where $a$ is the antipode of $u$ which is added to $L$ (if $u$ is not a center). As long as the selected node has eccentricity greater than $(1+\eps)\rad(G)$, the nodes in $L$ are $\eps\rad(G)$ far apart and form a packing for the collection of balls of radius $\eps\rad(G)/2$.
Similarly, the choice of a node $u$ with maximal eccentricity upper bound in Algorithm~\ref{alg:diam} implies $e^U(u)\ge\diam(G)$ and Lemma~\ref{lem:error} then implies $e(u)\ge \diam(G) - 2d(x,U)$ where $x$ is the tight upper certificate chosen for $u$ that is added to $U$. As long as the selected node has eccentricity less than $(1-\eps)\diam(G)$, the nodes in $U$ form a packing for the collection of balls of radius $\frac{\eps}{2}\diam(G)/2$.
As the doubling property implies that such packings have size bounded by $\gamma^{\ceil{\log \frac{2}{\eps}}+1}$, we obtain the following approximation results for radius and diameter.

\begin{proposition}
  \label{prop:doubling}
  Given a $\gamma$-doubling graph $G$ and $\eps > 0$, Algorithm~\ref{alg:rad} (resp. Algorithm~\ref{alg:diam}) provides a node $u$ with eccentricity $(1+\eps)\rad(G)$ at most (resp. $(1-\eps)\diam(G)$ at least) with $O(\gamma^{\ceil{\log \frac{2}{\eps}}+1})$ traversals.
\end{proposition}


\begin{proof}
  As discussed above, the set $L$ is a packing for balls of radius $\eps\rad(G)$ until a node whose eccentricity approximates the radius is found. We use the fact that packing size is bounded by covering size. By the $\gamma$-doubling property, the whole graph can be covered by $\gamma^i$ balls of radius $\frac{2\rad(G)}{2^i}$ as any ball of radius $2\rad(G)$ contains all nodes. For $i\ge\log \frac{2}{\eps}+1$ these balls have radius $\eps\rad(G)/2$ at most. Using that packing size is bounded by covering size, we can bound the number of iterations of Algorithm~\ref{alg:rad} where chosen nodes $u$ have eccentricity greater than $(1+\eps)\rad(G)$. If we stop the algorithm after $\gamma^{\ceil{\log \frac{2}{\eps}}+1}$ iterations, the node $u\in K$ with smallest eccentricity is guaranteed to have eccentricity $(1+\eps)\rad(G)$ at most.  The argument for diameter approximation is similar.
\end{proof}

In particular, for any $\eta \in (0,1)$, if $G$ is $\gamma$-doubling, has at most $n^\eta$ antipodes and $diam(G) < n^{\frac \eta {\log \gamma}}/8$, 
  then Proposition~\ref{prop:doubling} (applied for $\varepsilon = 8n^{-\frac \eta {\log \gamma}}$) 
  implies that we can compute both the diameter and the radius in $O(n^\eta m)$ time, which is truly subquadratic.

\section{Negatively curved graphs} 
\label{sec:hyperbolic}

In this section we analyze the behavior of our algorithms in the class of negatively curved graphs, alias, $\delta$-hyperbolic graphs.

Let $(X,d)$ be a metric space and $w\in X$. The {\it Gromov product} of $y,z\in X$ with respect to $w$ is defined to be
$$(y|z)_w=\frac{1}{2}(d(y,w)+d(z,w)-d(y,z)).$$
Let $\delta\ge 0$. A metric space $(X,d)$ is said to be $\delta$-{\it hyperbolic} \cite{Gr} if
$$(x|y)_w\ge \min \{ (x|z)_w, (y|z)_w\}-\delta$$
for all $w,x,y,z\in X$. Equivalently, $(X,d)$ is $\delta$-hyperbolic
if  for any four points $u,v,x,y$ of $X$, the two larger of the three distance sums
$d(u,v)+d(x,y)$, $d(u,x)+d(v,y)$, $d(u,y)+d(v,x)$ differ by at most
$2\delta \geq 0$. A graph $G=(V,E)$ is $\delta$-hyperbolic if the associated shortest-path metric space $(V,d)$ is $\delta$-hyperbolic.

The hyperbolicity can be viewed as a local measure of how close a graph is
metrically to a tree: the smaller the hyperbolicity is, the closer its metric is to
a tree metric (trees are 0-hyperbolic). Recent empirical studies showed that many 
real-world graphs (including Internet application networks, web networks, 
collaboration networks, social networks, biological networks, and others) have small 
hyperbolicity \cite{AADr,KSN16,Vien}. It is known \cite{ChDrEsHaVa,ChDrHaVaAlR} that 
if $G$ is a $\delta$-hyperbolic graph and $\{y,z\}$ is a pair returned after two BFS 
scans, then $d(y,z)\ge \diam(G) - 2\delta$, $\diam(G)\ge 2\rad(G) - 4\delta - 1$, 
$\diam(C(G)) \le 4\delta + 1$, and $C(G)$ is contained in a small ball centered at a 
middle vertex of any shortest $(y, z)$-path.
Recall that $C(G)$ denotes the set of centers of $G$ (see Section~\ref{sec:prelim}).
Consequently, there exist linear-time 
algorithms for the diameter
and radius problems with additive errors linearly depending on the input graph’s
hyperbolicity.  

In what follows,  we will analyze the behavior of our algorithms in the class of  $\delta$-hyperbolic graphs.  From the definition of a $\delta$-hyperbolic graph, we immediately get the following simple but very useful auxiliary lemma.

\begin{lemma}\label{lm:aux1}  Let $G=(V,E)$ be a $\delta$-hyperbolic graph. For every vertices $c,v,x,y\in V$, $d(x,v)-d(x,y)\geq d(c,v)-d(y,c) -2\delta$ or $d(y,v)-d(x,y)\geq d(c,v)-d(x,c) -2\delta$ holds.
\end{lemma}

\begin{proof} Assume, without loss of generality, that $d(x,c)+d(v,y)\leq d(y,c)+d(x,v)$. If also $d(c,v)+d(x,y)\geq d(y,c)+d(x,v)$ then, by $\delta$-hyperbolicity of $G$, $d(c,v)+d(x,y)- d(y,c)-d(x,v)\leq 2\delta$, i.e., $d(x,v)-d(x,y)\geq d(c,v)-d(y,c) -2\delta$. If $d(c,v)+d(x,y)\leq d(y,c)+d(x,v)$, then $d(x,v)-d(x,y)\geq d(c,v)-d(y,c)\geq d(c,v)-d(y,c) -2\delta$.
\end{proof}

Using this simple lemma, we get as easy corollaries two useful results from \cite{ChDrEsHaVa}. Denote by $F(s):=\{v\in V: d(s,v)=e(s)\}$ the set of vertices furthest from $s$. 

\begin{corollary} \label{cor:diam-rad}  For every $\delta$-hyperbolic graph $G$, $\diam(G)\geq 2\rad(G)-4\delta-1$.
\end{corollary}

\begin{proof} Let $x,y$ be vertices of $G$ such that $d(x,y)=\diam(G)$. Let $c$ be a middle vertex of any shortest path connecting $x$ with $y$. Applying  Lemma \ref{lm:aux1} to $c,v,x,y$, where $v\in F(c)$, assume without loss of generality that
$d(x,v)-d(x,y)\geq d(c,v)-d(y,c) -2\delta$ holds. Then, since $d(x,y)\geq d(x,v)$, $d(y,c)\geq d(c,v) -2\delta$. Hence,
$d(x,y)=d(x,c)+d(c,y)\geq 2d(c,y)-1\geq 2 d(c,v)-4\delta-1= 2e(c)-4\delta-1\geq 2\rad(G)-4\delta-1$.
\end{proof}

\begin{corollary} \label{cor:diam-2delta}  Let $G=(V,E)$ be a $\delta$-hyperbolic graph. For every vertices $c,v\in V$ such that $v\in F(c)$,
  $e(v)\geq \diam(G)-2\delta\geq 2\rad(G)-6\delta-1$.
\end{corollary}

\begin{proof} Apply Lemma \ref{lm:aux1} to $c,v$ and vertices $x,y$ such that $d(x,y)=\diam(G)$. Without loss of generality, assume that
$d(x,v)-d(x,y)\geq d(c,v)-d(y,c) -2\delta$ holds. Then, since $d(c,v)\geq d(c,y)$, $e(v)\geq d(x,v)\geq d(x,y)+d(c,v)-d(y,c) -2\delta\geq d(x,y)-2\delta=\diam(G)-2\delta$.
\end{proof}

We are ready to analyze Algorithm~\ref{alg:rad}. Let $u_i$ and $a_i\in F(u_i)$ be the vertices picked in iteration $i$ of the do-while loop. Let $K_i:=\{u_1,u_2,\dots, u_i\}$ and $L_i:=\{a_1,a_2,\dots, a_i\}$. 
According to the algorithm, $u_1$ is picked arbitrarily (as initially $L=\emptyset$), $a_1$ is a vertex furthest from $u_1$, $u_2=a_1$ (as $L_1=\{a_1\}$ is a singleton), $a_2$ is a vertex most distant from $u_2=a_1$, $u_3$ is a middle vertex of a shortest $(a_1,a_2)$-path. By Corollary \ref{cor:diam-2delta}, we already have $d(a_1,a_2)\geq \diam(G)-2\delta$. We can also show that $e(u_3)\leq \rad(G)+3\delta$.

\begin{proposition}\label{prop:middlevertex-ecc} If $G$ is a $\delta$-hyperbolic graph, then $d(a_1,a_2)\geq \diam(G)-2\delta$ and $e(u_3)\leq \rad(G)+3\delta$.
\end{proposition}

\begin{proof} We only need to estimate the eccentricity of the vertex $u_3$. As $u_3$ is a middle vertex of a shortest $(a_1,a_2)$-path, $\min\{d(u_3,a_1),d(u_3,a_2)\}\geq \lfloor\frac{d(a_1,a_2)}{2}\rfloor \geq \lfloor\frac{\diam(G)}{2}\rfloor -\delta$. Now, without loss of generality, assume (see Lemma \ref{lm:aux1}) that for vertices $u_3,a_3,a_2,a_1$, $d(a_2,a_3)-d(a_2,a_1)\geq d(u_3,a_3)-d(a_1,u_3) -2\delta$ holds. Then, $e(u_3)= d(u_3,a_3)\leq  d(a_2,a_3)-d(a_2,a_1)+d(a_1,u_3) +2\delta=d(a_2,a_3)- d(a_2,u_3) +2\delta\leq \diam(G)-\lfloor\frac{\diam(G)}{2}\rfloor +\delta +2\delta= \lceil\frac{\diam(G)}{2}\rceil +3\delta\leq \lceil\frac{2\rad(G)}{2}\rceil +3\delta=\rad(G)+3\delta$.
\end{proof}

Thus, in $\delta$-hyperbolic graphs, a vertex with eccentricity at most $\rad(G)+3\delta$ and a pair of vertices that are at least $\diam(G)-2\delta$ apart from each other are found by Algorithm \ref{alg:rad} in at most 3 iterations, i.e., in linear time. Note that a similar linear-time algorithm was already reported in  \cite{ChDrEsHaVa}: in $\delta$-hyperbolic graphs, a vertex with eccentricity at most $\rad(G)+5\delta$ and a pair of vertices that are at least $\diam(G)-2\delta$ apart from each other can be found in linear time.

Next, we show that a vertex with eccentricity at most $\rad(G)+2\delta$ is found by Algorithm \ref{alg:rad} in at most $2\delta+2$ iterations. Consider iteration $i\geq 3$ and let $p',p''$ be vertices of $L_{i-1}$ with the largest distance, i.e., $d(p',p'')=\max\{d(x,y): x,y\in L_{i-1}\}=diam(L_{i-1})$. If $e(u_{i})>\rad(G)+2\delta$, then applying Lemma \ref{lm:aux1} to $u_i,a_i,p',p''\in V$, we get $d(p',a_i)-d(p',p'')\geq d(u_i,a_i)-d(p'',u_i) -2\delta$ or $d(p'',a_i)-d(p',p'')\geq d(u_i,a_i)-d(p',u_i) -2\delta$. Hence, $\max\{d(p',a_i)-d(p',p''), d(p'',a_i)-d(p',p'') \}\geq e(u_i)- \rad(G)-2\delta>0$ (as $\max\{d(p'',u_i),d(p',u_i)\}\leq 
\min_{u\in V} e_{L_{i-1}}(u)\leq \rad(G)$). That is, if $e(u_{i})>\rad(G)+2\delta$ then $\diam(L_i)>\diam(L_{i-1})$. As $\diam(L_2)=d(a_1,a_2)\geq \diam(G)-2\delta$, in at most $2\delta+2$ iterations of the while-loop of Algorithm \ref{alg:rad} we will get $diam(L_i)=diam(L_{i-1})$ (with $i\leq 2\delta+2$) and hence  $e(u_{i})\leq \rad(G)+2\delta$ must hold. Thus, we proved the following proposition.

\begin{proposition}\label{prop:rad-2delta} If $G$ is a $\delta$-hyperbolic graph, then there is an index $i\leq 2\delta+2$ such that $e(u_i)\leq \rad(G)+2\delta$. Furthermore, $e(u_j)\leq \rad(G)+2\delta$ for all $j\geq 2\delta+2$.
\end{proposition}

The second part of Proposition \ref{prop:rad-2delta} says that all $u_i$ vertices generated by Algorithm \ref{alg:rad} after $2\delta+1$ iterations have eccentricity at most $\rad(G)+2\delta$. Hence, in $\delta$-hyperbolic graphs where the set $C^{2\delta}(G):=\{c\in V: e(c)\leq \rad(G)+2\delta\}$ has  cardinality  bounded by some function $g(\delta)$, depending only on $\delta$, our algorithm will produce a vertex with eccentricity $\rad(G)$ (i.e., a central vertex)   in at most $g(\delta)+2\delta+1$ iterations.

Next we show that the set $C^{2\delta}(G)$ of a $\delta$-hyperbolic graph has bounded diameter. Before our work, it was known that $diam(C(G))\leq 4\delta+1$ and there exists a vertex $c\in V$ such that $d(v,c)\leq 5\delta+1$ for every $v\in C(G)$ \cite{ChDrEsHaVa}.

\begin{proposition}\label{prop:bonded} If $G$ is a $\delta$-hyperbolic graph, then for every $x,y\in C^{2\delta}(G)$, $d(x,y)\leq 8\delta+1$. Furthermore, there is a vertex $c\in V$ such that $d(v,c)\leq 6\delta+1$ for every $v\in C^{2\delta}(G)$. 
\end{proposition}

\begin{proof} Let $c$ be a middle vertex of any shortest path connecting $x$ with $y$. Apply Lemma \ref{lm:aux1} to $c,v,x,y$, where $v\in F(c)$. Without loss of generality, assume that $d(x,v)-d(x,y)\geq d(c,v)-d(y,c) -2\delta$ holds. Then, $d(x,c)=d(x,y)-d(y,c)\leq d(x,v)-d(c,v) +2\delta\leq e(x)-e(c)+2\delta\leq \rad(G)+ 2\delta-\rad(G)+2\delta=4\delta$. Hence, $d(x,y)=d(x,c)+d(y,c)\leq 2d(x,c)+1\leq 8\delta+1$.

To prove the second assertion, consider a pair of vertices $x,y\in V$ with $d(x,y)=\diam(G)$ and a middle vertex $c$ of any shortest $(x,y)$-path.
Apply Lemma \ref{lm:aux1} to $c,v,x,y$, where $v$ is an arbitrary vertex from $C^{2\delta}(G)$. Without loss of generality, assume that $d(x,v)-d(x,y)\geq d(c,v)-d(y,c) -2\delta$ holds. We know also that $d(x,c)\geq \lfloor\frac{d(x,y)}{2}\rfloor = \lfloor\frac{\diam(G)}{2}\rfloor \geq \lfloor\frac{2\rad(G)-4\delta-1}{2}\rfloor\geq\rad(G)-2\delta-1$ (see Corollary \ref{cor:diam-rad}). Hence, $d(c,v)\leq d(x,v)-d(x,y)+d(y,c) +2\delta =d(x,v)-d(x,c)+2\delta \leq e(v)-\rad(G)+2\delta +2\delta\leq \rad(G)+2\delta+1-\rad(G)+4\delta=6\delta+1$.
\end{proof}

If the vertex degrees of a $\delta$-hyperbolic graph are bounded by a constant $\Delta$ then $C^{2\delta}(G)$ has at most $\Delta^{O(\delta)}$ vertices.
Summarizing, we have the following result.


\begin{theorem} \label{th:rad-diam-appr}
Let $G=(V,E)$ be a $\delta$-hyperbolic graph with $m$ edges. Algorithm \ref{alg:rad} finds \vspace*{-2mm}
\begin{enumerate}
   \item[1.] a vertex with eccentricity at most $\rad(G)+3\delta$ in at most $O(m)$ time,
   \item[2.] a vertex with eccentricity at most $\rad(G)+2\delta$ in at most $O(\delta m)$ time,
   \item[3.] a central vertex and a $O(1)$-size radius certificate in at most $O(m)$ time, if the vertex degrees and $\delta$ are bounded by constants.
\end{enumerate}
\end{theorem}

Another linear-time algorithm for finding a central vertex of a $\delta$-hyperbolic graph with $\delta$ and vertex degrees bounded by constants was proposed in \cite{ChDrEsHaVa}.

\section{Chordal graphs} 
\label{sec:chordal}

Recall that $F(s)=\{v\in V: d(v,s)=e(s)\}$ denotes the set of all vertices of $G$ that are furthest from $s$ and $C(G)=\{c\in V : e(c)=\rad(G)\}$ denotes the set of all central vertices of $G$. The {\em metric interval} $I(u,v)$ between vertices $u$ and $v$ is defined by $I(u,v)=\{x\in V: d(u,x)+ d(x,v)=d(u,v)\}$, i.e., it consists of all vertices of $G$ that lie on shortest paths between $u$ and $v$.

\bigskip

In this section we analyze the behavior of our algorithms in the class of chordal graphs. Recall that a graph $G$ is {\em chordal} if every induced cycle of length at least 4 has a chord.  Chordal graphs are interesting because a central vertex in them can be found in linear time \cite{ChDr94} but finding the diameter in truly subquadratic time would refute the Orthogonal Vectors Conjecture \cite{CDHP01,RV13}.


First we give an example of an $n$-vertex chordal graph $G$ on which Algorithm~\ref{alg:rad} will need $n/2$  iterations although  $G$ has a certificate for the radius consisting of only two vertices in $L$. Set $n=2k$ and consider two sets of vertices $X=\{x_1,\dots, x_k\}$ and $Y=\{y_1,\dots, y_k\}$. The vertex set of $G$ is $X\cup Y$. Make  $X$ a clique and $Y$ an independent set in $G$. Make every vertex $x_i$ adjacent to all vertices $y_j$ with $j\leq i$. Algorithm \ref{alg:rad} may place vertices $x_1,y_2,x_2,x_3, \dots, x_k$ (in this order) into $K$ and vertices $y_2,y_1,y_3,y_4,\dots, y_k$ (in this order) into $L$. The central vertex $x_k$ will be determined only when all $Y$-vertices are in $L$. On the other hand, $\{y_1,y_k\}$ is a radius certificate of $G$.

Note that the graph $G$ constructed has vertices of large degrees (up-to $n-1$). As every chordal graph $G$ has hyperbolicity at most 1,  it follows from Theorem \ref{th:rad-diam-appr} that our algorithm finds a central vertex  in linear time in every chordal graph  with vertex degrees bounded by a constant. It should be noted that there is a linear-time algorithm that finds a central vertex of an arbitrary chordal graph \cite{ChDr94}; it uses additional metric properties of chordal graphs.

To analyze possible radius and diameter certificates  in the class of chordal graphs, we will need the following important lemma.

\begin{lemma}[\cite{Ch86}]\label{lm:victor}
  Let $G$ be a chordal graph.  Let $x,y,v,u$ be
  vertices of $G$ such that $v \in I(x,y)$, $x\in I(v,u)$, and $x$ and
  $v$ are adjacent.  Then $d(u,y)\ge d(u,x) + d(v,y)$. Furthermore, $d(u,y)=d(u,x) + d(v,y)$ if and only
  if there exist a neighbor $x'$ of $x$ in $I(x,u)$, a neighbor
  $v'$ of $v$ in $I(v,y)$ and a vertex $w$ with $N(w)\supseteq\{x',x,v,v'\}$; in particular, $x'$,
  $v'$ and $w$ lie on a common shortest path of $G$ between $u$ and $y$.
\end{lemma}

Our analysis is based on the following propositions which are also of independent interest. Recall that $C^1(G):=\{v\in V: e(v)\le \rad(G)+1\}$.

\begin{proposition}\label{prop:chordal} Let $G=(V,E)$ be a chordal graph.
\begin{itemize}
  \item[(i)] If $\diam(G)<2\rad(G)$ then, for every vertices $s\in V$ and $t\in F(s)$, there is a vertex $w\in I(s,t)\cap C(G)$ such that $t\in F(w)$.
  \item[(ii)] If $\diam(G)=2\rad(G)$ then, for every vertices $s\in V$ and $t\in F(s)$, there is a vertex $w\in I(s,t)\cap C^1(G)$ such that $t\in F(w)$.
\end{itemize}
\end{proposition}

\begin{proof}
First we show that for every vertex $x$ of $G$ with $e(x)=k>\rad(G)$ there is a vertex $y$ such that $e(y)=k-1$ and $d(x,y)\leq 2$. This is true even in the case when $\diam(G)=2\rad(G)$. Consider a vertex $y$ in $G$ with $e(y)=k-1$ that is closest to $x$. Let $z$ be any neighbor of $y$ in $I(y,x)$. Necessarily, $e(z)=k$. Consider a vertex $u\in F(z)$. Since $d(y,u)\le e(y)=k-1=e(z)-1=d(z,u)-1\le d(y,u)$, we have $y\in I(z,u)$. Applying Lemma \ref{lm:victor} to $y\in I(z,u)$ and $z\in I(y,x)$, we get $d(x,u)\ge d(x,y)-1+d(y,u)=d(x,y)+k-2$. As $d(x,u)\leq e(x)=k$, we conclude $d(x,y)\le 2$.

Next we claim that if $\diam(G)<2\rad(G)$ then for every vertex $x$ of $G$ with $e(x)=k>\rad(G)$ there is in fact a vertex $z\in N(x)$ such that $e(z)=k-1$. Furthermore, if $\diam(G)=2\rad(G)$, such a neighbor $z$  with $e(z)=k-1$ exists for every vertex $x$ of $G$ with $e(x)=k>\rad(G)+1$.
Assume, by way of contradiction, that no neighbor of $x$ has eccentricity $k-1$ and let $y$ be an arbitrary vertex of $G$ with $d(x,y)=2$ and $e(y)=k-1$. Let also $z$ be a vertex from $N(x)\cap N(y)$ for which the set $S_x(z)=\{v\in F(x): z\in I(x,v)\}$ is largest. Necessarily, $e(z)=k$. As before, consider a vertex $u\in F(z)$. By Lemma \ref{lm:victor},  applied to $y\in I(z,u)$ and $z\in I(y,x)$, we get $d(x,u)\ge d(y,u)+d(x,z)=k-1+1=k$. As $d(x,u)\leq e(x)=k$, we conclude $d(x,u)=k$. Hence, by the second part of Lemma \ref{lm:victor}, there must exist a vertex $w$ adjacent to $y,z,x$ and at distance $k-1$ from $u$. As $u\in F(x)$, $w\in I(x,u), z\notin I(x,u)$, by the maximality of $|S_x(z)|$, there must exist a vertex $u'\in F(x)$ with $w\notin I(x,u'), z\in I(x,u')$. We have $d(z,u')=k-1$,  $d(w,u')=k$, and hence $z\in I(w,u')$ and $w\in I(z,u)$. By Lemma \ref{lm:victor},  $d(u,u')\ge d(u,w)+d(z,u')=k-1+k-1=2k-2$. Hence, $\diam(G)\ge d(u,u')> 2\rad(G)$, if $k> \rad(G)+1$, and $\diam(G)\ge d(u,u')\ge 2\rad(G)$, if $k= \rad(G)+1$.
These contradictions prove the claim.

Now we can conclude our proof. Consider arbitrary vertices  $s\in V$ and $t\in F(s)$ and proceed by induction on $k=e(s)$. If $k= \rad(G)$ then $w=s$ and we are done. If $k=\rad(G)+1$ and $\diam(G)=2\rad(G)$ then again $w=s$ and we are done. If $k>\rad(G)+1$ or $k=\rad(G)+1$ and $\diam(G)<2\rad(G)$ then a neighbor $z$ of $s$ with $e(z)=k-1$ satisfies $t\in F(z)$, and we can apply the induction hypothesis.
\end{proof}

A pair $x,y$ of vertices is called a {\em diametral pair} of a graph $G$ if $d(x,y)=\diam(G)$. 

\begin{proposition}\label{prop:chordal-cert-diam}  The center $C(G)$ of a chordal graph $G$ is a diameter certificate of $G$ (not necessarily a smallest one).
\end{proposition}

\begin{proof} It follows from Proposition \ref{prop:d=2r-any} and Proposition \ref{prop:chordal}{\it (i)}. 
\end{proof}

\begin{proposition}\label{prop:chordal-cert-ecc}   For every chordal graph $G$, the set $C^1(G)$ is a tight upper certificate of $G$ (not necessarily a smallest one).
\end{proposition}

\begin{proof} The statement follows from Proposition \ref{prop:chordal} and the definition of a tight upper certificate.
\end{proof}

So, it is interesting that if the center 
$C(G)$ is known for a chordal graph $G$ then its diameter can be computed in $O(|C(G)|m)$ time. 
However, there is no way to bound the cardinality of the set $C(G)$ in an arbitrary chordal graph $G$.  In fact $C(G)$ may contain $n-2$ vertices in some chordal graphs. To construct such a graph $G$, take a complete graph $K_{n-2}$ on $n-2$ vertices. Add two new vertices $u$ and $v$ adjacent to all vertices of $K_{n-2}$ but not to each other. It is easy to see that $C(G)$ contains exactly the vertices of $K_{n-2}$.

Nevertheless,  it is known that for every chordal graph $G$ and any two vertices $x,y$ from $C(G)$, $d(x,y)\leq 3$ holds \cite{Ch88}. This suggest the following approach for computing the diameter of a chordal graph $G=(V,E)$. \\
\indent\indent\indent - Use the linear-time algorithm from \cite{ChDr94} to find a central vertex $c$ of $G$.  \\
\indent\indent\indent - Set $C_3:=\{x\in V: d(x,c)\le 3\}$. /* $C(G)\subseteq C_3$ */ \\
\indent\indent\indent - Find a vertex $p$ such that $e^{C_3}(p)$ is maximum. \\
\indent\indent\indent - Report $e(p)$ as the diameter value. 

The complexity of this approach is $O(|C_3|m)$. As a consequence, we have that when the vertex degrees are bounded in a chordal graph by a constant $\Delta$ then its diameter can be computed in linear time (as $|C_3|$ is bounded by a constant  $\Delta^3$). We are not aware if such a result was known before. Note also that in general chordal graphs the cardinality of $C_3$ cannot be bounded by a constant since otherwise the diameter of an arbitrary chordal graph could be computer in linear time, refuting the Orthogonal Vectors Conjecture \cite{CDHP01,RV13}.


  From the proof of Proposition \ref{prop:chordal} it follows also that, for every vertex $v$ with $e(v)=\rad(G)+1$, $d(v,C(G))\leq 2$ holds.  This suggests the following approach for computing the eccentricities of all vertices of a chordal graph $G=(V,E)$. \\
\indent\indent\indent - Use the linear-time algorithm from \cite{ChDr94} to find a central vertex $c$ of $G$.  \\
\indent\indent\indent - Set $C_5:=\{x\in V: d(x,c)\le 5\}$. /* $C^1(G)\subseteq C_5$ */ \\
\indent\indent\indent - For every vertex $v\in V$ report $e(v)=\min_{c\in C_5} d(v,c)+e(c)$.

The complexity of this approach is $O(|C_5|m)$. As a consequence, we have that when the vertex degrees are bounded in a chordal graph by a constant $\Delta$ then the eccentricities of all its vertices can be computed in linear time (as $|C_5|$ is bounded by a constant  $\Delta^5$). We are not aware if such a result was known before.

Summarizing, we have the following result.


\begin{theorem} \label{th:all-ecc-chord}
Let $G=(V,E)$ be a chordal graph with $m$ edges and whose vertex degrees are bounded by a constant. Then, eccentricities of all vertices of $G$ can be computed in total $O(m)$ time. 
\end{theorem}

By Proposition \ref{prop:chordal-cert-diam}, the center $C(G)$ of a chordal graph $G$ is a diameter certificate of $G$. Next we will show that the set of all diametral vertices of a chordal graph $G$ forms a radius certificate of $G$. 

Let $D^k(G):=\{v\in V: e(v)\ge diam(G)-k\}$.
It is known that for every vertex $v$ of a chordal graph $G$ there is  a vertex $u\in F(v)$ with $e(u)\ge diam(G)-1$ \cite{DrNiBr97}. Hence, the set $D^1(G)$  contains the output set $L$ of Algorithm~\ref{alg:rad} and, therefore, it gives already a radius certificate for a chordal graph $G$. In fact, we can prove a stronger result.   

\begin{proposition}\label{prop:chordal--d=2r-2} For every chordal graph $G$, the set $D(G):=D^0(G)$ is a radius certificate of $G$ (not necessarily a smallest one).
\end{proposition}

\begin{proof} By Proposition \ref{prop:d>2r-2}  and the fact that in any chordal graph $G$, $diam(G)\ge 2rad(G)-2$ \cite{Ch88}, we need to consider only the case when $diam(G)= 2rad(G)-2$. 

Assume that there is a vertex $u$ in $G$ such that $d(u,t)\le rad(G)-1$ for every vertex $t\in D(G)$. Denote by $S$ the set of all such vertices $u$. Denote by $S'$ those vertices from $S$ that have the minimum eccentricity. Finally, denote by $S''$ those vertices $u$ from $S'$ that have 
the smallest number of vertices in $F(u)$. 

Consider a vertex $u\in S''$, a vertex $v\in F(u)$ and a neighbor $w$ of $u$ on a shortest path from $u$ to $v$.
Consider also an arbitrary vertex $x\in D(G)$ and an arbitrary vertex $y\in F(x)\subset D(G)$. Since $d(x,y)=2rad(G)-2$, we have $d(x,u)=d(y,u)=rad(G)-1$ and hence $u\in I(x,y)$. We claim that $d(w,x)\le rad(G)-1$ as well. Indeed, suppose that $d(w,x)>rad(G)-1$. Then, $u\in I(x,w)$ and $w\in I(u,v)$. By Lemma  \ref{lm:victor},  $d(x,v)\ge d(x,u)+d(w,v)\ge rad(G)-1+e(u)-1\ge 2rad(G)-2$, i.e., $d(x,v)=2rad(G)-2$ (hence $v$ must belong to $D(G)$) and $d(u,v)=e(u)=rad(G)$. The latter contradicts with the choice of $u$ (as $u\in S$). Thus, $d(w,x)\le rad(G)-1$ for every vertex $x\in D(G)$, i.e., $w\in S$. 

As $u\in S'$, $e(u)\leq e(w)$. First assume that $e(u)=e(w)$, i.e., $w\in S'$. 
Since $v\in F(u)$ and $v\notin F(w)$ (note that $d(w,v)=d(u,v)-1=e(u)-1\leq e(w)-1$), by the choice of $u$ (as $u\in S''$), there must exist a vertex $t\in V$ such that $t\in F(w)$ and $t\notin F(u)$.
We necessarily have $d(t,u)= e(w)-1\ge rad(G)-1$ and $d(v,w)= e(u)-1\ge rad(G)-1$. 
One can apply now Lemma \ref{lm:victor} to $u\in I(w,t)$ and $w\in I(u,v)$ and get $d(v,t)\ge d(v,w)+d(t,u)\ge 2rad(G)-2=diam(G)$. That is, both $v$ and $t$ must be in $D(G)$, contradicting again with the choice of $u$ (as $u\in S$ and $d(u,v)=e(u)\ge rad(G)$).

Assume now that $e(u)<e(w)$ and consider an arbitrary vertex $t\in F(w)$. Necessarily, $u\in I(w,t)$. Hence again one can apply Lemma \ref{lm:victor} to $u\in I(w,t)$ and $w\in I(u,v)$ and get $d(v,t)\ge d(v,w)+d(t,u)=e(u)-1+e(w)-1\ge 2rad(G)-1>diam(G)$, and a contradiction arises. 

Contradictions obtained prove that for every vertex $u\in V$ there is a vertex $t\in D(G)$ such that $d(u,t)\ge rad(G)$, i.e., $D(G)$ is a radius certificate of $G$.
\end{proof}

\section{Graph analogues of hyperconvex metric spaces}\label{Helly}
In this section, we analyze radius and diameter certificates in Helly graphs and in bipartite Helly graphs. A graph $G=(V,E)$ is {\em Helly} if its family of balls $\mathcal{B}(G):=\{B[v,r]: v\in V, r\in \mathbb{N}_{\ge 0}\}$ satisfies the Helly property (i.e., every collection  of pairwise intersecting balls has a nonempty common intersection). A bipartite graph $G=(X\cup Y, E)$  is {\em bipartite Helly} if its family of half-balls 
$\mathcal{B}_X(G)\bigcup \mathcal{B}_Y(G)$ satisfies the Helly property where $\mathcal{B}_X(G):=\{B[v,r]\cap X: v\in X\cup Y, r\in \mathbb{N}_{\ge 0}\}$ and $\mathcal{B}_Y(G):=\{B[v,r]\cap Y: v\in X\cup Y, r\in \mathbb{N}_{\ge 0}\}$.   
Helly graphs (and bipartite Helly graphs)  are discrete analogues  of {\em hyperconvex metric spaces} (sometimes called also {\em injective metric spaces}). It is well known that every (bipartite) graph  $G$ isometrically embeds into an unique smallest (bipartite) Helly graph $H(G)$ called the {\em injective hull} of $G$~\cite{Dress84,Isbell64}. The diameter of the injective hull $H(G)$ of a graph $G$ is equal to the diameter of $G$ (see, e.g., \cite{DrGu}). The latter may suggest that finding an efficient algorithm for the diameter of a (bipartite) Helly graph $H(G)$ is all what one needs to compute efficiently the diameter of an arbitrary graph $G$ (given $H(G)$). Unfortunately,  the injective hull $H(G)$ of a graph $G$ may have exponentially more vertices than $G$ has (see \cite{GuDrLe22} for several restricted graph classes where injective hulls are exponentially large).  

It was open until recently whether there are truly subquadratic-time algorithms for the radius and diameter problems on Helly graphs.  Such  algorithms were recently presented in~\cite{DrDuGu_Helly_hyp,DuDr21-netw}. One can compute in ${\mathcal O}(m\sqrt{n})$ time for a given $n$-vertex $m$-edge Helly graph  its radius,  diameter and all vertex eccentricities. The algorithms make use of the Helly property and of the unimodality of the eccentricity function in Helly graphs~\cite{Dr_thesis,Dr_Helly}: every vertex of locally minimum eccentricity is a central vertex. 

We will need the following result establishing relationships between the diameter  and the radius. Recall that, for an arbitrary graph $G$, $\rad(G)\le \diam(G)\le 2\rad(G)$ holds. 

\begin{proposition}\label{prop:rad-vs-diam}   For every Helly graph $G$, $\diam(G)\ge 2\rad(G)-1$~\cite{Dr_thesis,Dr_Helly}.  For every bipartite Helly graph $G$, $\diam(G)\ge 2\rad(G)-2$. 
\end{proposition}

\begin{proof} We give here a simple proof for both statements. The cases where $\rad(G)=1$ being obvious, we assume $\diam(G)\ge \rad(G)\ge 2$. 

Let  $G=(V,E)$ be a Helly graph and assume, by way of contradiction, 
$\diam(G)\le 2\rad(G)-2$. Consider in $G$ a system of balls $\mathcal{F}=\{B[v,\rad(G)-1]: v\in V\}$. For every two vertices $x,y$ of $G$, we have  $d(x,y)\le \diam(G)\le 2\rad(G)-2$. Hence, all balls of $\mathcal{F}$ pairwise intersect. By the Helly property, there must exist a vertex $c$ in $G$ which is at distance at most $\rad(G)-1$ from every $v\in V$.  Necessarily, $e(c)\leq \rad(G)-1$, contradicting the definition of the radius of $G$. 

Now consider a bipartite Helly graph $G=(X\cup Y,E)$ and assume, by way of contradiction, $\diam(G)\le 2\rad(G)-3$. 
First assume that $\rad(G)$ is even. Consider in $G$ the system of half-balls $\mathcal{F}_X(G):=\{B[v,\rad(G)-2]\cap X: v\in X\}\bigcup \{B[v,\rad(G)-1]\cap X: v\in Y\}$. 
We claim that all half-balls of $\mathcal{F}_X(G)$ pairwise intersect. 
Indeed, let $u,v \in X \cup Y$ be arbitrary. 
Let $\ell = \left\lfloor \frac{d(u,v)}2 \right\rfloor$, and let $a,b$ be consecutive vertices on a shortest $(u,v)$-path such that $d(u,a) = d(u,b)-1 = \ell$.
Since $G$ is bipartite, and vertices $a$ and $b$ are adjacent, one of the vertices $a$ or $b$ is in $X$. 
Therefore, $B[u,rad(G)-1] \cap B[v,rad(G)-1] \cap X \supseteq B[u,\ell+1] \cap B[v,\ell+1] \cap X \ne \emptyset$.
If furthermore $u \in X$, then $B[u,rad(G)-2] \cap X = B[u,rad(G)-1] \cap X$ because $rad(G)$ is even.
We proceed similarly if $v \in X$, thus proving the claim.
By the Helly property, there must exist a vertex $c$ in $X$ which is at distance at most $\rad(G)-1$ from every $v\in X\cup Y$.  Necessarily, $e(c)\leq \rad(G)-1$, contradicting the definition of the radius of $G$. 
If $\rad(G)$ is odd, then consider in $G$ the system of half-balls $\mathcal{F}_X(G):=\{B[v,\rad(G)-1]\cap X: v\in X\}\bigcup \{B[v,\rad(G)-2]\cap X: v\in Y\}$. Arguing the same way as above,
and noting that for $v \in Y$, we have $B[v,rad(G)-2] \cap X = B[v,rad(G)-1] \cap X$ when $rad(G)$ is odd,
we will get again a vertex $c$ in $X$ which is at distance at most $\rad(G)-1$ from every $v\in X\cup Y$ (i.e., with $e(c)\leq \rad(G)-1$), giving again a contradiction.  
\end{proof}

\begin{proposition}\label{prop:helly-cert}   Every Helly graph $G=(V,E)$ has: 
\begin{itemize}
  \item[$(1)$]   a radius certificate of size at most 2; 
  \item[$(2)$]  a diameter certificate of size at most $\omega(C(G))$ $($the size of a largest clique in the subgraph induced by $C(G))$. 
 \end{itemize}
 \end{proposition}

\begin{proof} By Proposition \ref{prop:d>2r-2}  and Proposition \ref{prop:rad-vs-diam}, we are done with $(1)$.  For  $(2)$, 
by Proposition \ref{prop:d=2r-any} and Proposition \ref{prop:rad-vs-diam}, it remains only to consider the case when $\diam(G)=2\rad(G)-1$. We will show that, in this case, there is a clique $K\subseteq C(G)$ such that, for every $v\in V$, $d(v,K)\le \rad(G)-1$ holds. Consider in $C(G)$ a maximal by inclusion clique $K\subseteq C(G)$ and assume that, for some vertex $u$, $d(u,K)> \rad(G)-1$. As $K\subseteq C(G)$, necessarily, $d(u,c)=\rad(G)$ holds for every $c\in K$.  
Consider in $G$ the system of balls $\{B[v,\rad(G)]: v\in V\}\bigcup\{B[c,1]: c\in K\}\bigcup \{B[u,\rad(G)-1]\}$. Since, for every $u,v\in V$ and $c\in K$, $d(u,v)\le \diam(G)\le 2\rad(G)-1$ and $d(v,c)\le \rad(G)$,  all these balls pairwise intersect. By the Helly property, there is a vertex $c'$ in $G$ which is adjacent to all $c\in K$, at distance $\rad(G)-1$ from $u$, and at distance at most $\rad(G)$ from each $v\in V$. Necessarily, $c'$ belongs to $C(G)$ and $K\cup\{c'\}$ is a clique from $C(G)$ with one more vertex than $K$ has. This contradicts the maximality of $K$, proving that  for every $v\in V$, $d(v,K)\le \rad(G)-1$ holds. Hence, $K$ is a diameter certificate of $G$ since, for every $v\in V$, there is a vertex $c\in K$ such that $d(v,c)+e(c)\le \rad(G)-1+\rad(G)=2\rad(G)-1=\diam(G)$. 
\end{proof}

The upper bound given in Proposition~\ref{prop:helly-cert}(2) is sharp.
For example, in a complete graph, all vertices must be included in the diameter certificate.

\begin{proposition}\label{prop:biphelly-cert}   Every bipartite Helly graph $G=(X\cup Y,E)$ has: 
\begin{itemize}
  \item[$(1)$]  a radius certificate of size at most 4; 
  \item[$(2)$] a diameter certificate of size at most $\omega_{bip}(C(G))$ $($the size of a largest biclique in the subgraph induced by $C(G))$. 
 \end{itemize}
 \end{proposition}

\begin{proof}
$(1)$ The case $diam(G) \ge 2rad(G)-1$ is covered by Proposition~\ref{prop:d>2r-2}.
Therefore, by Proposition~\ref{prop:rad-vs-diam}, we only need to consider the case when $diam(G) = 2rad(G)-2$.
Let $x,x' \in X$ be vertices of $X$ maximizing $d(x,x')$. In the same way, let $y,y' \in Y$ be vertices of $Y$ maximizing $d(y,y')$.
We prove, in what follows, that $L = \{x,x',y,y'\}$ forms a radius certificate.
For that, without loss of generality, assume $d(x,x') \ge d(y,y')$.
Since $G$ is bipartite and $diam(G)$ is even, we get $d(x,x') = diam(G) = 2rad(G)-2$.
In particular, for every vertex $v \in X \cup Y$, we have $\max\{d(v,x),d(v,x')\} \ge rad(G)-1$.
If furthermore $rad(G)$ is even and $v \in X$ ($rad(G)$ is odd and $v \in Y$, resp.), then we get $\max\{d(v,x),d(v,x')\} \ge rad(G)$.
By symmetry, we thus deduce that if $d(y,y') = d(x,x')$, then $L$ indeed forms a radius certificate.
Thus, from now on, suppose $d(y,y') < d(x,x')$, which implies $d(y,y') = diam(G)-2 = 2rad(G)-4$.
In order to complete the proof, it suffices to show that this case can never happen.
Let $H_Y$ be the graph with vertex set $Y$ such that there is an edge between every two vertices of $Y$ that are at distance two in $G$.
We denote by $d_Y$ the distance function in $H_Y$, while we keep denoting by $d$ (without a subscript) the distance function in $G$.
Since $G$ is bipartite, we have $d_Y(u,v) = d(u,v)/2$ for every $u,v \in Y$.
In particular, by the choice of $y,y'$ we get $diam(H_Y) = d_Y(y,y') = d(y,y')/2 = rad(G)-2$.
It was observed in~\cite[Lemma 1]{Duc21} that $H_Y$ is a Helly graph.
If $rad(G)$ is even, then $diam(H_Y)$ is even, and so by Proposition~\ref{prop:rad-vs-diam} we get $rad(H_Y) = diam(H_Y)/2 = rad(G)/2 - 1$.
But then, for any $c_Y \in C(H_Y)$, we would obtain $e_G(c_Y) \le 2rad(H_Y)+1 = rad(G) - 2 + 1 = rad(G) -1 < rad(G)$, which is a contradiction.
As a result, $rad(G)$ must be odd.
By Proposition~\ref{prop:rad-vs-diam}, it implies that $rad(H_Y) = (diam(H_Y)+1)/2 = (rad(G)-1)/2$.
Let $K \subseteq C(H_Y)$ be a clique of $H_Y$ such that $d_Y(v,K) \le rad(H_Y)-1$ for every $v \in Y$, whose existence follows from Proposition~\ref{prop:helly-cert}(2).
Since $K$ is a clique, the half-balls in $\{ B[z,1] \cap X: z \in K \}$ pairwise intersect in $G$.
By the Helly property, there exists a vertex $c_K \in X$ 
such that $K \subseteq N(c_K)$.
But then, $e_G(c_K) \le 1 + \max\{d(v,K) : v \in X \cup Y\} \le 2 + \max\{d(v,K) : v \in Y\} \le 2 + 2(rad(H_Y)-1) = 2rad(H_Y) = rad(G)-1 < rad(G)$, which is again a contradiction.


$(2)$ The case when $\diam(G)=2\rad(G)$ is covered by Proposition \ref{prop:d=2r-any}; any central vertex of $G$ is a diameter certificate.  
Assume $\diam(G)=2\rad(G)-1$.  We will show that, in this case, two adjacent central vertices form a diameter certificate. Note that since   $\diam(G)=2\rad(G)-1$, i.e., the diameter is odd, for any two vertices $u,w$ from $X$ (or from $Y$), $d(u,w)\le 2\rad(G)-2$ holds. Assume $\rad(G)$ is odd. Consider in $G$ the system of half-balls $\mathcal{F}_X(G):=\{B[v,\rad(G)-1]\cap X: v\in X\}$.  These half-balls  pairwise intersect. By the Helly property, there must exist a vertex $c_X$ in $X$  which is at distance at most $\rad(G)-1$ from every $v\in X$  and, hence, at distance at most $\rad(G)$ from every $v\in X\cup Y$. Hence, $c_X\in C(G).$ 
Consider now a system of half-balls $\mathcal{F}_Y(G):=\{B[v,\rad(G)-1]\cap Y: v\in Y\}$ and a half-ball $B[c_X,1]\cap Y$.  All these half-balls  pairwise intersect. By the Helly property, there must exist a vertex $c_Y$ in $Y$  which is adjacent to $c_X$, and which is at distance at most $\rad(G)-1$ from every $v\in Y$. Since  $c_Y$ is  at distance at most $\rad(G)$ from every $v\in X\cup Y$,  $c_Y\in C(G).$ So, we found two adjacent vertices $c_X, c_Y\in C(G)$ such that, for every vertex $v\in X$,  $d(v,c_X)\le \rad(G)-1$ and, therefore, $d(v,c_X)+e(c_X)\le \rad(G)-1+\rad(G)=\diam(G)$, and, for every vertex $v\in Y$,  $d(v,c_Y)\le \rad(G)-1$ and, therefore, $d(v,c_Y)+e(c_Y)\le \rad(G)-1+\rad(G)=\diam(G)$. Hence, $\{c_X, c_Y\}$ is a diameter certificate. The case when $\rad(G)$ is even is similar. We just need to consider first  the system of half-balls $\mathcal{F'}_X(G):=\{B[v,\rad(G)-1]\cap X: v\in Y\}$ which will produce, by the Helly property,  a central vertex $c_X\in X$ which is at distance at most $\rad(G)-1$ from every $v\in Y$. 
Then, we can  consider the system of half-balls $\mathcal{F'}_Y(G):=\{B[v,\rad(G)-1]\cap Y: v\in X\}$ plus a half-ball $B[c_X,1]\cap Y$ which together will produce a central vertex $c_Y\in Y$ which is at distance at most $\rad(G)-1$ from every $v\in X$. 
Again, $\{c_X, c_Y\}\subseteq C(G)$ forms a diameter certificate. 

Finally assume $\diam(G)=2\rad(G)-2$.  We will show that, in this case, a biclique consisting of central vertices of $G$ forms a diameter certificate. Assume  $\rad(G)$ is odd (the case, when $\rad(G)$ is even, is very similar as above). As in the case $\diam(G)=2\rad(G)-1$, we can get two adjacent vertices $c_X\in C(G)\cap X$ and $c_Y\in C(G)\cap Y$ such that, for every vertex $v\in X$,  $d(v,c_X)\le \rad(G)-1$ and, for every vertex $v\in Y$,  $d(v,c_Y)\le \rad(G)-1$.  Let $K_{p,q}$ be a maximal by inclusion biclique in $C(G)$ containing edge $c_Xc_Y$. We will show that $d(v,K_{p,q})\le \rad(G)-2$ holds for every $v\in X\cup Y$. Assume there is a vertex $v'$ such that $d(v',K_{p,q})> \rad(G)-2$. Without loss of generality, say $v'\in X$. Since $d(v',c_X)\le \rad(G)-1$, we have  $d(v',K_{p,q})= \rad(G)-1$. Consider a system of half-balls $\{B[v,\rad(G)-1]\cap Y: v\in Y\}\bigcup \{B[c,1]\cap Y: c\in K_{p,q}\cap X\}\bigcup \{B[v',\rad(G)-2]\cap Y\}$. Since the diameter $\diam(G)=2\rad(G)-2$ is even and, for every $u\in Y$, $d(v',u)$ is odd, necessarily, $d(v',u)\le 2\rad(G)-3$ and $B[u,\rad(G)-1]$ intersects $B[v',\rad(G)-2]$. Furthermore, since for every $c\in K_{p,q}\cap X$, we have $d(v',c)\le \rad(G)$, and $d(v',c)$ is  even while $\rad(G)$ is odd, we necessarily have $d(v',c)\le \rad(G)-1$, and $B[c,1]$ intersects $B[v',\rad(G)-2]$. Since we have $c_Y\in B[v,\rad(G)-1]\cap B[c,1]$ for all $v\in Y$ and $c\in K_{p,q}\cap X$,  we conclude that all half-balls pairwise intersect. By the Helly property, there must exist a vertex $c'$ in $Y$ which is adjacent to every $c\in K_{p,q}\cap X$, at distance at most $\rad(G)-2$ from $v'$, and at distance at most $\rad(G)-1$ from every $y\in Y$. That is, $c'\in C(G)\cap Y$.
Moreover, $K_{p,q}\cup\{c'\}$ is a biclique in $C(G)$ containing one more vertex than $K_{p,q}$. This contradicts the maximality of $K_{p,q}$, proving that $d(v,K_{p,q})\le \rad(G)-2$ holds for every $v\in X\cup Y$. Therefore, the biclique  $K_{p,q}$ is a diameter certificate as for every $v\in X\cup Y$ there is a vertex $c$ in  $K_{p,q}$ such that $d(v,c)+e(c)\le \rad(G)-2+\rad(G) =\diam(G).$  
\end{proof}

Again, we can observe that the upper bound stated in Proposition~\ref{prop:biphelly-cert}(2) is sharp.
For instance, in a complete bipartite graph, all vertices must be included in the diameter certificate.

\begin{corollary}
    For every (bipartite) Helly graph $G$, we can compute its radius in $O(m\log^3{n})$ time with high probability.
\end{corollary}
\begin{proof}
    It directly follows from Theorem~\ref{thm:rad-cert-random}, in combination with Propositions~\ref{prop:helly-cert} and~\ref{prop:biphelly-cert}.
\end{proof}

\guillaume{
We end up observing that our result for Helly graphs can be generalized to an upper bound on the size of a smallest radius certificate.
More precisely, the \emph{Helly number} of a graph $G$ is the least integer $k \ge 2$ such that for every family $\mathcal{F}$ of balls in $G$, if every $k$ balls in $\mathcal{F}$ have a nonempty common intersection, then some node must be contained in every ball in $\mathcal{F}$. Helly graphs are exactly the graphs of Helly number two.
\begin{proposition}
    For any graph $G$, there exists a radius certificate of size at most its Helly number.
\end{proposition}
\begin{proof}
    Let $r = \rad(G)$.
    We consider a smallest subset $L$ of nodes such that $\bigcap\set{B[x,r-1] : x\in L} = \emptyset$.
    If $G$ has Helly number $k$, then $L$ has size at most $k$.
    Furthermore, $e_L(v) \ge r$ for every node $v$, and so, $L$ is a radius certificate.
\end{proof}
\begin{corollary}
    For every graph $G$ with Helly number at most $k$, we can compute its radius in $O(km\log^3{n})$ time with high probability.
\end{corollary}
This improves on~\cite{ducoffe2023distance}, where an $O(m\sqrt{kn\log{n}})$-time algorithm is presented.
}

\section{Graphs of bounded asteroidal number}\label{Asteroid}

We finally give 
an application of this paper's framework to parameterized complexity.
Namely, an {\em asteroidal set} in a graph $G=(V,E)$ is an independent set $A \subseteq V$ such that, for every $a \in A$, all vertices of $A \setminus \{a\}$ must be in a same connected component of $G \setminus B[a,1]$.
The {\em asteroidal number} of $G$ is the largest cardinality of its asteroidal sets.
In particular, the graphs of asteroidal number one are exactly the complete graphs.
The graphs of asteroidal number at most two are called AT-free graphs, and they contain cocomparability graphs, interval graphs and permutation graphs amongst their interesting subclasses.
In~\cite{Duc21JGT}, a deterministic $O(m^{3/2})$-time algorithm for computing all eccentricities in an AT-free graph is presented.
In~\cite{Duc22}, a deterministic algorithm for computing the diameter of graphs with asteroidal number $k$ in $O(k^3m^{3/2})$ time is given.
However, the complexity of computing the radius in graphs of asteroidal number $k \ge 3$ was open until this work.
We prove the following result in this section.
\begin{theorem}\label{thm:rad-asteroidal-number}
    For every graph $G=(V,E)$ with asteroidal  number at most $k$, its radius and a central vertex can be computed in $O(km^{3/2})$ time. 
\end{theorem}

Roughly, the result follows from replacing in Algorithm~\ref{alg:rad} the notion of antipode with that of {\em extremity}.
A vertex $v$ in a graph $G$ is called an extremity if and only if $G \setminus B[v,1]$ is connected.
There are graphs with no extremities, such as complete bipartite graphs.
Therefore, in what follows, we need an additional assumption on the graphs considered.
Namely, a module in a graph $G=(V,E)$ is a vertex subset $M \subseteq V$ such that every vertex of $V \setminus M$ is either adjacent to all of $M$, or nonadjacent to all of it.
A graph $G=(V,E)$ is prime if and only if it has no other modules but $\emptyset, V$ and $\{v\}$, for every $v \in V$.

\begin{lemma}[see Theorem 14 in~\cite{CDP19}]\label{lem:rad-prime}
    Computing the radius of any graph $G$ can be reduced in linear time to computing the radius of one of its prime induced subgraphs $G'$.
\end{lemma}


Hence, we only need to consider prime graphs of bounded asteroidal number.
We need the following additional results:

\begin{lemma}[see Lemma $19$ in~\cite{Duc22}]\label{lem:furthest-extrem}
    If $G=(V,E)$ is prime and $n \ge 3$, then for every vertex $v$, there exists an extremity $u$ such that $d(u,v) = e(v)$, which can be computed in linear time. 
\end{lemma}

\begin{lemma}[see Lemma $7$ in~\cite{Duc22}]\label{lem:num-extrem}
    If $G=(V,E)$ has asteroidal number at most $k$, then it contains at most $O(k\sqrt{m})$ extremities.
\end{lemma}


We are now ready to prove the main result of this section, namely:

\begin{proof}[of Theorem~\ref{thm:rad-asteroidal-number}]
    We may assume that $G$ is prime (by Lemma~\ref{lem:rad-prime}) and that $n \ge 3$.
    We modify Algorithm~\ref{alg:rad} as follows: at every iteration, we add in $L$ an extremity $a$ such that, for the vertex $u$ considered, $d(u,a) = e(u)$. 
    By Lemma~\ref{lem:furthest-extrem}, such an extremity $a$ always exists, and it can be computed in linear time.
    Doing so, we still have the invariant $e_L(u) = e(u)$ for every $u \in K$.
    In particular, we can prove as before (see the proof of Theorem~\ref{th:rad}) that this algorithm correctly computes the radius and a central vertex.
    By Lemma~\ref{lem:num-extrem}, the number of iterations is in $O(k\sqrt{m})$ at most.
    Therefore, the runtime of the algorithm is in $O(km^{3/2})$.
\end{proof}

%% file: main-journal.bbl
\begin{thebibliography}{10}

\bibitem{AWV16}
Amir Abboud, Virginia~Vassilevska Williams, and Joshua~R. Wang.
\newblock Approximation and fixed parameter subquadratic algorithms for radius
  and diameter in sparse graphs.
\newblock In Robert Krauthgamer, editor, {\em Proceedings of the Twenty-Seventh
  Annual {ACM-SIAM} Symposium on Discrete Algorithms, {SODA} 2016, Arlington,
  VA, USA, January 10-12, 2016}, pages 377--391. {SIAM}, 2016.
\newblock URL: \url{https://doi.org/10.1137/1.9781611974331.ch28}, \href
  {https://doi.org/10.1137/1.9781611974331.CH28}
  {\path{doi:10.1137/1.9781611974331.CH28}}.

\bibitem{AbrahamDGW2012}
Ittai Abraham, Daniel Delling, Andrew~V. Goldberg, and Renato Fonseca~F.
  Werneck.
\newblock Hierarchical hub labelings for shortest paths.
\newblock In Leah Epstein and Paolo Ferragina, editors, {\em Algorithms - {ESA}
  2012 - 20th Annual European Symposium, Ljubljana, Slovenia, September 10-12,
  2012. Proceedings}, volume 7501 of {\em Lecture Notes in Computer Science},
  pages 24--35. Springer, 2012.
\newblock \href {https://doi.org/10.1007/978-3-642-33090-2\_4}
  {\path{doi:10.1007/978-3-642-33090-2\_4}}.

\bibitem{AADr}
Muad Abu{-}Ata and Feodor~F. Dragan.
\newblock Metric tree-like structures in real-world networks: an empirical
  study.
\newblock {\em Networks}, 67(1):49--68, 2016.
\newblock URL: \url{https://doi.org/10.1002/net.21631}, \href
  {https://doi.org/10.1002/NET.21631} {\path{doi:10.1002/NET.21631}}.

\bibitem{AingworthCIM1999}
Donald Aingworth, Chandra Chekuri, Piotr Indyk, and Rajeev Motwani.
\newblock Fast estimation of diameter and shortest paths (without matrix
  multiplication).
\newblock {\em {SIAM} J. Comput.}, 28(4):1167--1181, 1999.
\newblock \href {https://doi.org/10.1137/S0097539796303421}
  {\path{doi:10.1137/S0097539796303421}}.

\bibitem{AIK15}
Takuya Akiba, Yoichi Iwata, and Yuki Kawata.
\newblock An exact algorithm for diameters of large real directed graphs.
\newblock In Evripidis Bampis, editor, {\em Experimental Algorithms - 14th
  International Symposium, {SEA} 2015, Paris, France, June 29 - July 1, 2015,
  Proceedings}, volume 9125 of {\em Lecture Notes in Computer Science}, pages
  56--67. Springer, 2015.
\newblock \href {https://doi.org/10.1007/978-3-319-20086-6\_5}
  {\path{doi:10.1007/978-3-319-20086-6\_5}}.

\bibitem{AlkassarBMR2011}
Eyad Alkassar, Sascha B{\"{o}}hme, Kurt Mehlhorn, and Christine Rizkallah.
\newblock Verification of certifying computations.
\newblock In Ganesh Gopalakrishnan and Shaz Qadeer, editors, {\em Computer
  Aided Verification - 23rd International Conference, {CAV} 2011, Snowbird, UT,
  USA, July 14-20, 2011. Proceedings}, volume 6806 of {\em Lecture Notes in
  Computer Science}, pages 67--82. Springer, 2011.
\newblock \href {https://doi.org/10.1007/978-3-642-22110-1\_7}
  {\path{doi:10.1007/978-3-642-22110-1\_7}}.

\bibitem{BackstromBRUV2012}
Lars Backstrom, Paolo Boldi, Marco Rosa, Johan Ugander, and Sebastiano Vigna.
\newblock Four degrees of separation.
\newblock In Noshir~S. Contractor, Brian Uzzi, Michael~W. Macy, and Wolfgang
  Nejdl, editors, {\em Web Science 2012, WebSci '12, Evanston, IL, {USA} - June
  22 - 24, 2012}, pages 33--42. {ACM}, 2012.
\newblock \href {https://doi.org/10.1145/2380718.2380723}
  {\path{doi:10.1145/2380718.2380723}}.

\bibitem{BBRUV12}
Lars Backstrom, Paolo Boldi, Marco Rosa, Johan Ugander, and Sebastiano Vigna.
\newblock Four degrees of separation.
\newblock In Noshir~S. Contractor, Brian Uzzi, Michael~W. Macy, and Wolfgang
  Nejdl, editors, {\em Web Science 2012, WebSci '12, Evanston, IL, {USA} - June
  22 - 24, 2012}, pages 33--42. {ACM}, 2012.
\newblock \href {https://doi.org/10.1145/2380718.2380723}
  {\path{doi:10.1145/2380718.2380723}}.

\bibitem{BlumFPRT1973}
Manuel Blum, Robert~W. Floyd, Vaughan~R. Pratt, Ronald~L. Rivest, and
  Robert~Endre Tarjan.
\newblock Time bounds for selection.
\newblock {\em J. Comput. Syst. Sci.}, 7(4):448--461, 1973.
\newblock \href {https://doi.org/10.1016/S0022-0000(73)80033-9}
  {\path{doi:10.1016/S0022-0000(73)80033-9}}.

\bibitem{Bollobas1980}
B{\'{e}}la Bollob{\'{a}}s.
\newblock A probabilistic proof of an asymptotic formula for the number of
  labelled regular graphs.
\newblock {\em Eur. J. Comb.}, 1(4):311--316, 1980.
\newblock \href {https://doi.org/10.1016/S0195-6698(80)80030-8}
  {\path{doi:10.1016/S0195-6698(80)80030-8}}.

\bibitem{Bonnet2021}
{\'{E}}douard Bonnet.
\newblock 4 vs 7 sparse undirected unweighted diameter is seth-hard at time
  n{\^{}}\{4/3\}.
\newblock In Nikhil Bansal, Emanuela Merelli, and James Worrell, editors, {\em
  48th International Colloquium on Automata, Languages, and Programming,
  {ICALP} 2021, July 12-16, 2021, Glasgow, Scotland (Virtual Conference)},
  volume 198 of {\em LIPIcs}, pages 34:1--34:15. Schloss Dagstuhl -
  Leibniz-Zentrum f{\"{u}}r Informatik, 2021.
\newblock URL: \url{https://doi.org/10.4230/LIPIcs.ICALP.2021.34}, \href
  {https://doi.org/10.4230/LIPICS.ICALP.2021.34}
  {\path{doi:10.4230/LIPICS.ICALP.2021.34}}.

\bibitem{BCHKMT15}
Michele Borassi, Pierluigi Crescenzi, Michel Habib, Walter~A. Kosters, Andrea
  Marino, and Frank~W. Takes.
\newblock Fast diameter and radius bfs-based computation in (weakly connected)
  real-world graphs: With an application to the six degrees of separation
  games.
\newblock {\em Theor. Comput. Sci.}, 586:59--80, 2015.
\newblock URL: \url{https://doi.org/10.1016/j.tcs.2015.02.033}, \href
  {https://doi.org/10.1016/J.TCS.2015.02.033}
  {\path{doi:10.1016/J.TCS.2015.02.033}}.

\bibitem{BCT17}
Michele Borassi, Pierluigi Crescenzi, and Luca Trevisan.
\newblock An axiomatic and an average-case analysis of algorithms and
  heuristics for metric properties of graphs.
\newblock In Philip~N. Klein, editor, {\em Proceedings of the Twenty-Eighth
  Annual {ACM-SIAM} Symposium on Discrete Algorithms, {SODA} 2017, Barcelona,
  Spain, Hotel Porta Fira, January 16-19}, pages 920--939. {SIAM}, 2017.
\newblock \href {https://doi.org/10.1137/1.9781611974782.58}
  {\path{doi:10.1137/1.9781611974782.58}}.

\bibitem{BringmannHM2020}
Karl Bringmann, Thore Husfeldt, and M{\aa}ns Magnusson.
\newblock Multivariate analysis of orthogonal range searching and graph
  distances.
\newblock {\em Algorithmica}, 82(8):2292--2315, 2020.
\newblock URL: \url{https://doi.org/10.1007/s00453-020-00680-z}, \href
  {https://doi.org/10.1007/S00453-020-00680-Z}
  {\path{doi:10.1007/S00453-020-00680-Z}}.

\bibitem{CGR16}
Massimo Cairo, Roberto Grossi, and Romeo Rizzi.
\newblock New bounds for approximating extremal distances in undirected graphs.
\newblock In Robert Krauthgamer, editor, {\em Proceedings of the Twenty-Seventh
  Annual {ACM-SIAM} Symposium on Discrete Algorithms, {SODA} 2016, Arlington,
  VA, USA, January 10-12, 2016}, pages 363--376. {SIAM}, 2016.
\newblock URL: \url{https://doi.org/10.1137/1.9781611974331.ch27}, \href
  {https://doi.org/10.1137/1.9781611974331.CH27}
  {\path{doi:10.1137/1.9781611974331.CH27}}.

\bibitem{CGIMPS16}
Marco~L. Carmosino, Jiawei Gao, Russell Impagliazzo, Ivan Mihajlin, Ramamohan
  Paturi, and Stefan Schneider.
\newblock Nondeterministic extensions of the strong exponential time hypothesis
  and consequences for non-reducibility.
\newblock In Madhu Sudan, editor, {\em Proceedings of the 2016 {ACM} Conference
  on Innovations in Theoretical Computer Science, Cambridge, MA, USA, January
  14-16, 2016}, pages 261--270. {ACM}, 2016.
\newblock \href {https://doi.org/10.1145/2840728.2840746}
  {\path{doi:10.1145/2840728.2840746}}.

\bibitem{ChechikLRSTW2014}
Shiri Chechik, Daniel~H. Larkin, Liam Roditty, Grant Schoenebeck, Robert~Endre
  Tarjan, and Virginia {Vassilevska Williams}.
\newblock Better approximation algorithms for the graph diameter.
\newblock In Chandra Chekuri, editor, {\em Proceedings of the Twenty-Fifth
  Annual {ACM-SIAM} Symposium on Discrete Algorithms, {SODA} 2014, Portland,
  Oregon, USA, January 5-7, 2014}, pages 1041--1052. {SIAM}, 2014.
\newblock \href {https://doi.org/10.1137/1.9781611973402.78}
  {\path{doi:10.1137/1.9781611973402.78}}.

\bibitem{Ch86}
Victor Chepoi.
\newblock Some {$d$}-convexity properties in triangulated graphs.
\newblock {\em Mathematical Research}, 87:164--177, 1986.
\newblock \c{S}tiin\c{t}a, Chi\c{s}in\u{a}u (Russian).

\bibitem{ChDr94}
Victor Chepoi and Feodor~F. Dragan.
\newblock A linear-time algorithm for finding a central vertex of a chordal
  graph.
\newblock In Jan van Leeuwen, editor, {\em Algorithms - {ESA} '94, Second
  Annual European Symposium, Utrecht, The Netherlands, September 26-28, 1994,
  Proceedings}, volume 855 of {\em Lecture Notes in Computer Science}, pages
  159--170. Springer, 1994.
\newblock URL: \url{https://doi.org/10.1007/BFb0049406}, \href
  {https://doi.org/10.1007/BFB0049406} {\path{doi:10.1007/BFB0049406}}.

\bibitem{ChDrEsHaVa}
Victor Chepoi, Feodor~F. Dragan, Bertrand Estellon, Michel Habib, and Yann
  Vax{\`{e}}s.
\newblock Diameters, centers, and approximating trees of
  delta-hyperbolicgeodesic spaces and graphs.
\newblock In Monique Teillaud, editor, {\em Proceedings of the 24th {ACM}
  Symposium on Computational Geometry, College Park, MD, USA, June 9-11, 2008},
  pages 59--68. {ACM}, 2008.
\newblock \href {https://doi.org/10.1145/1377676.1377687}
  {\path{doi:10.1145/1377676.1377687}}.

\bibitem{ChDrHaVaAlR}
Victor Chepoi, Feodor~F. Dragan, Michel Habib, Yann Vax{\`{e}}s, and Hend
  Alrasheed.
\newblock Fast approximation of eccentricities and distances in hyperbolic
  graphs.
\newblock {\em J. Graph Algorithms Appl.}, 23(2):393--433, 2019.
\newblock URL: \url{https://doi.org/10.7155/jgaa.00496}, \href
  {https://doi.org/10.7155/JGAA.00496} {\path{doi:10.7155/JGAA.00496}}.

\bibitem{CE07}
Victor Chepoi and Bertrand Estellon.
\newblock Packing and covering \emph{delta} -hyperbolic spaces by balls.
\newblock In Moses Charikar, Klaus Jansen, Omer Reingold, and Jos{\'{e}} D.~P.
  Rolim, editors, {\em Approximation, Randomization, and Combinatorial
  Optimization. Algorithms and Techniques, 10th International Workshop,
  {APPROX} 2007, and 11th International Workshop, {RANDOM} 2007, Princeton, NJ,
  USA, August 20-22, 2007, Proceedings}, volume 4627 of {\em Lecture Notes in
  Computer Science}, pages 59--73. Springer, 2007.
\newblock \href {https://doi.org/10.1007/978-3-540-74208-1\_5}
  {\path{doi:10.1007/978-3-540-74208-1\_5}}.

\bibitem{Ch88}
Victor~D Chepoi.
\newblock Centers of triangulated graphs.
\newblock {\em Mathematical Notes of the Academy of Sciences of the USSR},
  43:82--86, 1988.
\newblock URL:
  \url{https://pageperso.lis-lab.fr/~victor.chepoi/centers\_triang.pdf}.

\bibitem{Clarkson1995}
Kenneth~L. Clarkson.
\newblock Las vegas algorithms for linear and integer programming when the
  dimension is small.
\newblock {\em J. {ACM}}, 42(2):488--499, 1995.
\newblock \href {https://doi.org/10.1145/201019.201036}
  {\path{doi:10.1145/201019.201036}}.

\bibitem{CDHP01}
Derek~G. Corneil, Feodor~F. Dragan, Michel Habib, and Christophe Paul.
\newblock Diameter determination on restricted graph families.
\newblock {\em Discret. Appl. Math.}, 113(2-3):143--166, 2001.
\newblock \href {https://doi.org/10.1016/S0166-218X(00)00281-X}
  {\path{doi:10.1016/S0166-218X(00)00281-X}}.

\bibitem{CDK03}
Derek~G. Corneil, Feodor~F. Dragan, and Ekkehard K{\"{o}}hler.
\newblock On the power of {BFS} to determine a graph's diameter.
\newblock {\em Networks}, 42(4):209--222, 2003.
\newblock URL: \url{https://doi.org/10.1002/net.10098}, \href
  {https://doi.org/10.1002/NET.10098} {\path{doi:10.1002/NET.10098}}.

\bibitem{CDP19}
David Coudert, Guillaume Ducoffe, and Alexandru Popa.
\newblock Fully polynomial {FPT} algorithms for some classes of bounded
  clique-width graphs.
\newblock {\em {ACM} Trans. Algorithms}, 15(3):33:1--33:57, 2019.
\newblock \href {https://doi.org/10.1145/3310228} {\path{doi:10.1145/3310228}}.

\bibitem{CGHLM13}
Pilu Crescenzi, Roberto Grossi, Michel Habib, Leonardo Lanzi, and Andrea
  Marino.
\newblock On computing the diameter of real-world undirected graphs.
\newblock {\em Theor. Comput. Sci.}, 514:84--95, 2013.
\newblock URL: \url{https://doi.org/10.1016/j.tcs.2012.09.018}, \href
  {https://doi.org/10.1016/J.TCS.2012.09.018}
  {\path{doi:10.1016/J.TCS.2012.09.018}}.

\bibitem{DaLiWi}
Mina Dalirrooyfard, Ray Li, and Virginia {Vassilevska Williams}.
\newblock Hardness of approximate diameter: Now for undirected graphs.
\newblock {\em J. {ACM}}, 72(1):6:1--6:32, 2025.
\newblock \href {https://doi.org/10.1145/3704631} {\path{doi:10.1145/3704631}}.

\bibitem{Vien}
Fabien de~Montgolfier, Mauricio Soto, and Laurent Viennot.
\newblock Treewidth and hyperbolicity of the internet.
\newblock In {\em Proceedings of The Tenth {IEEE} International Symposium on
  Networking Computing and Applications, {NCA} 2011, August 25-27, 2011,
  Cambridge, Massachusetts, {USA}}, pages 25--32. {IEEE} Computer Society,
  2011.
\newblock \href {https://doi.org/10.1109/NCA.2011.11}
  {\path{doi:10.1109/NCA.2011.11}}.

\bibitem{DS14}
Irit Dinur and David Steurer.
\newblock Analytical approach to parallel repetition.
\newblock In David~B. Shmoys, editor, {\em Symposium on Theory of Computing,
  {STOC} 2014, New York, NY, USA, May 31 - June 03, 2014}, pages 624--633.
  {ACM}, 2014.
\newblock \href {https://doi.org/10.1145/2591796.2591884}
  {\path{doi:10.1145/2591796.2591884}}.

\bibitem{Dr_thesis}
Feodor~F. Dragan.
\newblock {\em Centers of graphs and the Helly property}.
\newblock PhD thesis, Moldova State University, 1989.
\newblock (in Russian).

\bibitem{Dr_Helly}
Feodor~F. Dragan.
\newblock Conditions for coincidence of local and global minima for
  eccentricity function on graphs and the helly property.
\newblock {\em Applied Mathematics and Information Science}, pages 49--56,
  1990.

\bibitem{DrDuGu_Helly_hyp}
Feodor~F. Dragan, Guillaume Ducoffe, and Heather~M. Guarnera.
\newblock Fast deterministic algorithms for computing all eccentricities in
  (hyperbolic) helly graphs.
\newblock In Anna Lubiw and Mohammad~R. Salavatipour, editors, {\em Algorithms
  and Data Structures - 17th International Symposium, {WADS} 2021, Virtual
  Event, August 9-11, 2021, Proceedings}, volume 12808 of {\em Lecture Notes in
  Computer Science}, pages 300--314. Springer, 2021.
\newblock \href {https://doi.org/10.1007/978-3-030-83508-8\_22}
  {\path{doi:10.1007/978-3-030-83508-8\_22}}.

\bibitem{DrGu}
Feodor~F. Dragan and Heather~M. Guarnera.
\newblock Helly-gap of a graph and vertex eccentricities.
\newblock {\em Theor. Comput. Sci.}, 867:68--84, 2021.
\newblock URL: \url{https://doi.org/10.1016/j.tcs.2021.03.022}, \href
  {https://doi.org/10.1016/J.TCS.2021.03.022}
  {\path{doi:10.1016/J.TCS.2021.03.022}}.

\bibitem{DrNiBr97}
Feodor~F. Dragan, Falk Nicolai, and Andreas Brandst{\"{a}}dt.
\newblock Lexbfs-orderings and power of graphs.
\newblock In Fabrizio d'Amore, Paolo~Giulio Franciosa, and Alberto
  Marchetti{-}Spaccamela, editors, {\em Graph-Theoretic Concepts in Computer
  Science, 22nd International Workshop, {WG} '96, Cadenabbia (Como), Italy,
  June 12-14, 1996, Proceedings}, volume 1197 of {\em Lecture Notes in Computer
  Science}, pages 166--180. Springer, 1996.
\newblock \href {https://doi.org/10.1007/3-540-62559-3\_15}
  {\path{doi:10.1007/3-540-62559-3\_15}}.

\bibitem{Dress84}
Andreas~WM Dress.
\newblock Trees, tight extensions of metric spaces, and the cohomological
  dimension of certain groups: a note on combinatorial properties of metric
  spaces.
\newblock {\em Advances in Mathematics}, 53(3):321--402, 1984.
\newblock URL: \url{https://doi.org/10.1016/0001-8708(84)90029-X}.

\bibitem{Duc21}
Guillaume Ducoffe.
\newblock Beyond helly graphs: The diameter problem on absolute retracts.
\newblock In Lukasz Kowalik, Michal Pilipczuk, and Pawel Rzazewski, editors,
  {\em Graph-Theoretic Concepts in Computer Science - 47th International
  Workshop, {WG} 2021, Warsaw, Poland, June 23-25, 2021, Revised Selected
  Papers}, volume 12911 of {\em Lecture Notes in Computer Science}, pages
  321--335. Springer, 2021.
\newblock \href {https://doi.org/10.1007/978-3-030-86838-3\_25}
  {\path{doi:10.1007/978-3-030-86838-3\_25}}.

\bibitem{Duc21JGT}
Guillaume Ducoffe.
\newblock The diameter of at-free graphs.
\newblock {\em J. Graph Theory}, 99(4):594--614, 2022.
\newblock URL: \url{https://doi.org/10.1002/jgt.22754}, \href
  {https://doi.org/10.1002/JGT.22754} {\path{doi:10.1002/JGT.22754}}.

\bibitem{Duc22}
Guillaume Ducoffe.
\newblock Obstructions to faster diameter computation: Asteroidal sets.
\newblock In Holger Dell and Jesper Nederlof, editors, {\em 17th International
  Symposium on Parameterized and Exact Computation, {IPEC} 2022, September 7-9,
  2022, Potsdam, Germany}, volume 249 of {\em LIPIcs}, pages 10:1--10:24.
  Schloss Dagstuhl - Leibniz-Zentrum f{\"{u}}r Informatik, 2022.
\newblock URL: \url{https://doi.org/10.4230/LIPIcs.IPEC.2022.10}, \href
  {https://doi.org/10.4230/LIPICS.IPEC.2022.10}
  {\path{doi:10.4230/LIPICS.IPEC.2022.10}}.

\bibitem{ducoffe2023distance}
Guillaume Ducoffe.
\newblock Distance problems within helly graphs and k-helly graphs.
\newblock {\em Theoretical Computer Science}, 946:113690, 2023.

\bibitem{DuDr21-netw}
Guillaume Ducoffe and Feodor~F. Dragan.
\newblock A story of diameter, radius, and (almost) helly property.
\newblock {\em Networks}, 77(3):435--453, 2021.
\newblock URL: \url{https://doi.org/10.1002/net.21998}, \href
  {https://doi.org/10.1002/NET.21998} {\path{doi:10.1002/NET.21998}}.

\bibitem{FunkePS2025}
Stefan Funke, Claudius Proissl, and Sabine Storandt.
\newblock Computing the exact radius of large graphs.
\newblock In Petra Mutzel and Nicola Prezza, editors, {\em 23rd International
  Symposium on Experimental Algorithms, {SEA} 2025, July 22-24, 2025, Venice,
  Italy}, volume 338 of {\em LIPIcs}, pages 17:1--17:14. Schloss Dagstuhl -
  Leibniz-Zentrum f{\"{u}}r Informatik, 2025.
\newblock URL: \url{https://doi.org/10.4230/LIPIcs.SEA.2025.17}, \href
  {https://doi.org/10.4230/LIPICS.SEA.2025.17}
  {\path{doi:10.4230/LIPICS.SEA.2025.17}}.

\bibitem{gavoille2001small}
Cyril Gavoille, David Peleg, Andr{\'e} Raspaud, and Eric Sopena.
\newblock Small k-dominating sets in planar graphs with applications.
\newblock In {\em International Workshop on Graph-Theoretic Concepts in
  Computer Science}, pages 201--216. Springer, 2001.

\bibitem{GeisbergerSSD2008}
Robert Geisberger, Peter Sanders, Dominik Schultes, and Daniel Delling.
\newblock Contraction hierarchies: Faster and simpler hierarchical routing in
  road networks.
\newblock In Catherine~C. McGeoch, editor, {\em Experimental Algorithms, 7th
  International Workshop, {WEA} 2008, Provincetown, MA, USA, May 30-June 1,
  2008, Proceedings}, volume 5038 of {\em Lecture Notes in Computer Science},
  pages 319--333. Springer, 2008.
\newblock \href {https://doi.org/10.1007/978-3-540-68552-4\_24}
  {\path{doi:10.1007/978-3-540-68552-4\_24}}.

\bibitem{goodman2000tight}
Oliver Goodman and Vincent Moulton.
\newblock On the tight span of an antipodal graph.
\newblock {\em Discrete Mathematics}, 218(1-3):73--96, 2000.

\bibitem{Gr}
Mikhael Gromov.
\newblock Hyperbolic groups.
\newblock In {\em Essays in group theory}, pages 75--263. Springer, 1987.

\bibitem{GuDrLe22}
Heather~M. Guarnera, Feodor~F. Dragan, and Arne Leitert.
\newblock Injective hulls of various graph classes.
\newblock {\em Graphs Comb.}, 38(4):112, 2022.
\newblock URL: \url{https://doi.org/10.1007/s00373-022-02512-z}, \href
  {https://doi.org/10.1007/S00373-022-02512-Z}
  {\path{doi:10.1007/S00373-022-02512-Z}}.

\bibitem{H73}
Gabriel~Y Handler.
\newblock Minimax location of a facility in an undirected tree graph.
\newblock {\em Transportation Science}, 7(3):287--293, 1973.
\newblock URL:
  \url{https://pubsonline.informs.org/doi/abs/10.1287/trsc.7.3.287}.

\bibitem{Isbell64}
John~R Isbell.
\newblock Injective envelopes of banach spaces are rigidly attached.
\newblock {\em Commentarii mathematici Helvetici}, 39:65--76, 1964.
\newblock URL:
  \url{https://projecteuclid.org/journals/bulletin-of-the-american-mathematical-society/volume-70/issue-5/Injective-envelopes-of-Banach-spaces-are-rigidly-attached/bams/1183526270.pdf}.

\bibitem{KSN16}
W.~Sean Kennedy, Iraj Saniee, and Onuttom Narayan.
\newblock On the hyperbolicity of large-scale networks and its estimation.
\newblock In James Joshi, George Karypis, Ling Liu, Xiaohua Hu, Ronay Ak,
  Yinglong Xia, Weijia Xu, Aki{-}Hiro Sato, Sudarsan Rachuri, Lyle~H. Ungar,
  Philip~S. Yu, Rama Govindaraju, and Toyotaro Suzumura, editors, {\em 2016
  {IEEE} International Conference on Big Data {(IEEE} BigData 2016), Washington
  DC, USA, December 5-8, 2016}, pages 3344--3351. {IEEE} Computer Society,
  2016.
\newblock URL: \url{https://doi.org/10.1109/BigData.2016.7840994}, \href
  {https://doi.org/10.1109/BIGDATA.2016.7840994}
  {\path{doi:10.1109/BIGDATA.2016.7840994}}.

\bibitem{Kunnemann2018}
Marvin K{\"{u}}nnemann.
\newblock On nondeterministic derandomization of freivalds' algorithm:
  Consequences, avenues and algorithmic progress.
\newblock In Yossi Azar, Hannah Bast, and Grzegorz Herman, editors, {\em 26th
  Annual European Symposium on Algorithms, {ESA} 2018, August 20-22, 2018,
  Helsinki, Finland}, volume 112 of {\em LIPIcs}, pages 56:1--56:16. Schloss
  Dagstuhl - Leibniz-Zentrum f{\"{u}}r Informatik, 2018.
\newblock \href {https://doi.org/10.4230/LIPICS.ESA.2018.56}
  {\path{doi:10.4230/LIPICS.ESA.2018.56}}.

\bibitem{Li2020}
Ray Li.
\newblock Improved seth-hardness of unweighted diameter.
\newblock {\em CoRR}, abs/2008.05106, 2020.
\newblock URL: \url{https://arxiv.org/abs/2008.05106}, \href
  {https://arxiv.org/abs/2008.05106} {\path{arXiv:2008.05106}}.

\bibitem{MLH09}
Cl{\'{e}}mence Magnien, Matthieu Latapy, and Michel Habib.
\newblock Fast computation of empirically tight bounds for the diameter of
  massive graphs.
\newblock {\em {ACM} J. Exp. Algorithmics}, 13, 2008.
\newblock \href {https://doi.org/10.1145/1412228.1455266}
  {\path{doi:10.1145/1412228.1455266}}.

\bibitem{McConnellMNS2011}
Ross~M. McConnell, Kurt Mehlhorn, Stefan N{\"{a}}her, and Pascal Schweitzer.
\newblock Certifying algorithms.
\newblock {\em Comput. Sci. Rev.}, 5(2):119--161, 2011.
\newblock \href {https://doi.org/10.1016/J.COSREV.2010.09.009}
  {\path{doi:10.1016/J.COSREV.2010.09.009}}.

\bibitem{Monier1980}
Louis Monier.
\newblock Combinatorial solutions of multidimensional divide-and-conquer
  recurrences.
\newblock {\em J. Algorithms}, 1(1):60--74, 1980.
\newblock \href {https://doi.org/10.1016/0196-6774(80)90005-X}
  {\path{doi:10.1016/0196-6774(80)90005-X}}.

\bibitem{NaP90}
R.~Nandakumar and K.~Parthasarathy.
\newblock Eccentricity-preserving spanning trees.
\newblock {\em J. Math. Phys. Sci.}, 24(1):33--35, 1990.

\bibitem{RV13}
Liam Roditty and Virginia~Vassilevska Williams.
\newblock Fast approximation algorithms for the diameter and radius of sparse
  graphs.
\newblock In Dan Boneh, Tim Roughgarden, and Joan Feigenbaum, editors, {\em
  Symposium on Theory of Computing Conference, STOC'13, Palo Alto, CA, USA,
  June 1-4, 2013}, pages 515--524. {ACM}, 2013.
\newblock \href {https://doi.org/10.1145/2488608.2488673}
  {\path{doi:10.1145/2488608.2488673}}.

\bibitem{SM10}
Sandeep Sen and V.~N. Muralidhara.
\newblock The covert set-cover problem with application to network discovery.
\newblock In Md.~Saidur Rahman and Satoshi Fujita, editors, {\em {WALCOM:}
  Algorithms and Computation, 4th International Workshop, {WALCOM} 2010, Dhaka,
  Bangladesh, February 10-12, 2010. Proceedings}, volume 5942 of {\em Lecture
  Notes in Computer Science}, pages 228--239. Springer, 2010.
\newblock \href {https://doi.org/10.1007/978-3-642-11440-3\_21}
  {\path{doi:10.1007/978-3-642-11440-3\_21}}.

\bibitem{TK11}
Frank~W. Takes and Walter~A. Kosters.
\newblock Determining the diameter of small world networks.
\newblock In Craig Macdonald, Iadh Ounis, and Ian Ruthven, editors, {\em
  Proceedings of the 20th {ACM} Conference on Information and Knowledge
  Management, {CIKM} 2011, Glasgow, United Kingdom, October 24-28, 2011}, pages
  1191--1196. {ACM}, 2011.
\newblock \href {https://doi.org/10.1145/2063576.2063748}
  {\path{doi:10.1145/2063576.2063748}}.

\bibitem{TK13}
Frank~W. Takes and Walter~A. Kosters.
\newblock Computing the eccentricity distribution of large graphs.
\newblock {\em Algorithms}, 6(1):100--118, 2013.
\newblock URL: \url{https://doi.org/10.3390/a6010100}, \href
  {https://doi.org/10.3390/A6010100} {\path{doi:10.3390/A6010100}}.

\bibitem{sagemathGD2025}
{The Sage Developers}.
\newblock {\em {S}ageMath, the {S}age {M}athematics {S}oftware {S}ystem
  ({V}ersion 10.6), diameter function in the Undirected graphs library}, 2025.
\newblock URL:
  \url{https://doc.sagemath.org/html/en/reference/graphs/sage/graphs/graph.html#sage.graphs.graph.Graph.diameter}.

\end{thebibliography}
